\documentclass[bib]{ba}
\usepackage{bayes}

\usepackage{multirow}

\usepackage{graphicx} 
\usepackage{subfigure}
\usepackage{algorithm}
\usepackage{algorithmic}
\usepackage{hyperref}
\usepackage{amsmath, amssymb, amsthm}

\usepackage{color}

\newcommand{\given}{\,|\,}

\newcommand{\beq}{\begin{equation}}
\newcommand{\eeq}{\end{equation}}
\newcommand{\beqs}{\begin{eqnarray}}
\newcommand{\eeqs}{\end{eqnarray}}
\newcommand{\barr}{\begin{array}}
\newcommand{\earr}{\end{array}}
\newcommand{\bali}{\begin{aligned}}
\newcommand{\eali}{\end{aligned}}

\newcommand{\Nc}[0]{\ensuremath{\mathcal{N}} }

\newcommand{\Oc}[0]{\ensuremath{\mathcal{O}} }

\newcommand{\Rbb}[0]{\ensuremath{\mathbb{R}} }

\newcommand{\Ebb}[0]{\ensuremath{\mathbb{E}} }
\newcommand{\Zbb}[0]{\ensuremath{\mathbb{Z}} }

\newcommand{\Amat}[0]{\ensuremath{{\bf A}}}
\newcommand{\Bmat}[0]{\ensuremath{{\bf B}} }
\newcommand{\Cmat}[0]{\ensuremath{{\bf C}} }
\newcommand{\Dmat}[0]{\ensuremath{{\bf D}} }

\newcommand{\Gmat}[0]{\ensuremath{{\bf G}} }
\newcommand{\Hmat}[0]{\ensuremath{{\bf H}} }
\newcommand{\Imat}[0]{\ensuremath{{\bf I}} }

\newcommand{\Mmat}[0]{\ensuremath{{\bf M}} }

\newcommand{\Qmat}[0]{\ensuremath{{\bf Q}} }

\newcommand{\bds}[1]{\boldsymbol{#1}}

\newcommand{\dv}[0]{\ensuremath{\boldsymbol{d}} }
\newcommand{\ev}[0]{\ensuremath{\boldsymbol{e}} }
\newcommand{\fv}[0]{\ensuremath{\boldsymbol{f}} }

\newcommand{\lv}[0]{\ensuremath{\boldsymbol{l}} }

\newcommand{\nv}[0]{\ensuremath{\boldsymbol{n}} }

\newcommand{\rv}[0]{\ensuremath{\boldsymbol{r}} }

\newcommand{\tv}[0]{\ensuremath{\boldsymbol{t}} }
\newcommand{\uv}[0]{\ensuremath{\boldsymbol{u}} }

\newcommand{\xv}[0]{\ensuremath{\boldsymbol{x}} }
\newcommand{\yv}[0]{\ensuremath{\boldsymbol{y}} }
\newcommand{\zv}[0]{\ensuremath{\boldsymbol{z}} }

\newcommand{\Lambdamat}[0]{\ensuremath{\boldsymbol{\Lambda}} }

\newcommand{\Sigmamat}[0]{\ensuremath{\boldsymbol{\Sigma}} }

\newcommand{\Phimat}[0]{\ensuremath{\boldsymbol{\Phi}}}

\newcommand{\Omegamat}[0]{\ensuremath{\boldsymbol{\Omega}}}

\newcommand{\alphav}[0]{\ensuremath{\boldsymbol{\alpha}} }
\newcommand{\betav}[0]{\ensuremath{\boldsymbol{\beta}} }

\newcommand{\muv}[0]{\ensuremath{\boldsymbol{\mu}} }

\newcommand{\phiv}[0]{\ensuremath{\boldsymbol{\phi}} }

\newcommand{\varphiv}[0]{\ensuremath{\boldsymbol{\varphi}} }

\newcommand{\cdotv}[0]{\ensuremath{\boldsymbol{\cdot}}}

\newcommand{\zz}[0]{\color{black}}

\newcommand{\Dir}[0]{\ensuremath{\mathrm{Dir}} }

\newcommand{\Pois}[0]{\ensuremath{\mathrm{Pois}} }
\newcommand{\Mult}[0]{\ensuremath{\mathrm{Mult}} }

\newcommand{\diag}[0]{\ensuremath{\mathrm{diag}} }
\newcommand{\Nor}[0]{\ensuremath{\mathcal{N}} }

\newtheorem{thm}{Theorem} 
\newtheorem{cor}[thm]{Corollary}

\begin{document}

\inserttype[ba0001]{article}
\renewcommand{\thefootnote}{\fnsymbol{footnote}}
\author{Y. Cong and B. Chen and M. Zhou }{
 \fnms{Yulai}
 \snm{Cong}\footnotemark[1]\ead{yulai_cong@163.com},
  \fnms{Bo}
  \snm{Chen}\footnotemark[2]\footnotemark[4]\ead{bchen@mail.xidian.edu.cn}, 
  and
  \fnms{Mingyuan}
  \snm{Zhou}\footnotemark[3]\footnotemark[4]\ead{mingyuan.zhou@mccombs.utexas.edu}

}

\title[Fast Simulation of Hyperplane-Truncated MVN Distributions]{Fast Simulation of Hyperplane-Truncated Multivariate Normal Distributions}

\maketitle

\footnotetext[1]{
National Laboratory of Radar Signal Processing and Collaborative Innovation Center of Information Sensing \& Understanding, 
Xidian University, Xi'an, Shaanxi 710071, China,
 \href{mailto:yulai\_cong@163.com}{yulai\_cong@163.com} 
}
\footnotetext[2]{
National Laboratory of Radar Signal Processing and Collaborative Innovation Center of Information Sensing \& Understanding, 
Xidian University, Xi'an, Shaanxi 710071, China,
\href{mailto:bchen@mail.xidian.edu.cn}{bchen@mail.xidian.edu.cn}
 }
\footnotetext[3]{
  McCombs School of Business, 
 The University of Texas at Austin, Austin, TX 78712, USA,
 \href{mailto:mingyuan.zhou@mccombs.utexas.edu}{mingyuan.zhou@mccombs.utexas.edu}
}
\footnotetext[4]{
Corresponding authors.} 
\renewcommand{\thefootnote}{\arabic{footnote}}

\begin{abstract}
 We introduce a fast and easy-to-implement simulation algorithm for a multivariate normal distribution truncated on the intersection of a set of hyperplanes, and further generalize it to efficiently simulate random variables from a multivariate normal distribution whose covariance (precision) matrix can be decomposed as a positive-definite matrix minus (plus) a low-rank symmetric matrix. Example results illustrate the correctness and efficiency of the proposed simulation algorithms. 

\end{abstract}

\section{Introduction}

We investigate the problem of simulation from a multivariate normal (MVN) distribution whose samples are restricted to the intersection of a set of hyperplanes, 
which is shown to be inherently related to the simulation of a conditional distribution of a 
MVN distribution. A naive approach, which linearly transforms a random variable drawn from the conditional distribution of a related MVN distribution, 
requires a large number of intermediate variables that are often computationally expensive to instantiate.
To address this issue, we propose a fast and exact simulation algorithm that directly projects a MVN random variable onto the intersection of a set of hyperplanes.
We further show that sampling from 
a MVN distribution, whose covariance (precision) matrix can be decomposed as the sum (difference) of a positive-definite matrix, whose inversion is known or easy to compute, and a low-rank symmetric matrix, may also be made significantly fast by exploiting this newly proposed stimulation algorithm for hyperplane-truncated MVN distributions, avoiding the need of Cholesky decomposition that has a computational complexity of $O(k^3)$ \citep{golub2012matrix}, where $k$ is the dimension of the MVN random variable.  
%
%

Related to the problems under study, 
the simulation of MVN random variables subject to certain constraints 
\citep{gelfand1992bayesian} has been investigated in many other different settings, such as multinomial probit and logit models 
\citep{albert1993bayesian,mcculloch2000bayesian,imai2005bayesian,train2009discrete,holmes2006bayesian,johndrow2013diagonal}, 
Bayesian isotonic regression \citep{neelon2004bayesian}, Bayesian bridge 
 \citep{polson2014bayesian}, blind source separation \citep{schmidt2009linearly}, and unmixing of hyperspectral data \citep{altmann2014sampling,
dobigeon2009joint}. A typical example arising in these different settings  is to sample a random vector $\xv\in\mathbb{R}^k$ from a MVN distribution subject to $k$ inequality constraints as
\beq
\xv\sim\Nor_{\mathcal{S}}(\muv,\Sigmamat),~~\mathcal{S}=\{\xv:\lv\le \Gmat \xv \le \uv\},
\eeq
where $\Nor_{\mathcal{S}}(\muv,\Sigmamat)$ represents a MVN distribution truncated on the sample space $\mathcal{S}$, $\muv\in\mathbb{R}^k$ is the mean, $\Sigmamat\in\mathbb{R}^{k\times k}$ is the covariance matrix, $\Gmat\in\mathbb{R}^{k\times k}$ is a full-rank matrix, $\lv\in \mathbb{R}^k$, $\uv\in \mathbb{R}^k$, and $\lv<\uv$. If the elements of $\lv$ and $\uv$ are permitted to be $-\infty$ and $+\infty$, respectively, then both single sided and fewer than $k$ inequality constraints are allowed.
Equivalently, as in \citet{geweke1991efficient,geweke1996bayesian}, one may let $\xv = \muv + \Gmat^{-1} \zv$
and use Gibbs sampling \citep{Geman,gelfand1990sampling} to sample the $k$ elements of $\zv$ one at a time conditioning on all the others from a univariate truncated normal distribution, for which efficient algorithms exist \citep{robert1995simulation,damien2001sampling,chopin2011fast}.
To deal with the case that the number of linear constraints imposed on $\xv$ exceed its dimension $k$ and to obtain better mixing,
one may consider the Gibbs sampling algorithm for truncated MVN distributions proposed in \citet{rodriguez2004efficient}.
 In addition to Gibbs sampling, to sample truncated MVN random variables, one may also consider Hamiltonian Monte Carlo \citep{pakman2014exact,lan2014spherical} and a minimax tilting method proposed in \citet{botev2016normal}.

\section{Hyperplane-truncated and conditional MVNs}\label{sec:2}
For the problem under study, we express a $k$-dimensional MVN distribution truncated on the intersection of $k_2<k$ hyperplanes as
\beq\label{eq:x_S}
\xv\sim\Nor_{\mathcal{S}}(\muv,\Sigmamat),~~\mathcal{S}=\{\xv: \Gmat \xv = \rv\},
\eeq
where
$$
\Gmat\in\mathbb{R}^{k_2\times k}, ~~\rv\in\mathbb{R}^{k_2},
$$
and $\mbox{Rank}(\Gmat)=k_2$. The probability density function can be expressed as
\begin{align}\label{eq:px}
p( \xv \given \muv,\Sigmamat, \Gmat,\rv) =\frac{1}{Z} \exp\left[-\frac{1}{2}(\xv- \muv)^T \Sigmamat^{-1} (\xv- \muv) \right] \delta(\Gmat\xv = \rv),
\end{align}
where $Z$ is a constant ensuring $\int p( \xv \given \muv,\Sigmamat, \Gmat,\rv)d\xv = 1$, and $\delta(x)=1$ if the condition $x$ is satisfied and $\delta(x)=0$ otherwise. Let us partition $\Gmat$, $\xv$, $\muv$, $\Sigmamat$, and $\Lambdamat = \Sigmamat^{-1}$ as
$$
\Gmat = (\Gmat_1,\Gmat_2),~
\xv=\begin{bmatrix}
 \xv_1 \\
 \xv_2
\end{bmatrix},~
\muv = \begin{bmatrix}
 \muv_1 \\
 \muv_2
\end{bmatrix},~
\Sigmamat = \begin{bmatrix}
\Sigmamat_{11} &\Sigmamat_{12} \\
\Sigmamat_{21} & \Sigmamat_{22}
\end{bmatrix},~
\mbox{and}~
\Lambdamat = \begin{bmatrix}
\Lambdamat_{11} & \Lambdamat_{12} \\
\Lambdamat_{21} & \Lambdamat_{22}
\end{bmatrix},
$$
whose sizes are
$$
(k_2\times k_1,k_2\times k_2),~
\begin{bmatrix}
 k_1 \times 1 \\
 k_2 \times 1
\end{bmatrix},~
 \begin{bmatrix}
 k_1 \times 1\\
 k_2\times 1
\end{bmatrix},~
 \begin{bmatrix}
 k_1\times k_1&k_1\times k_2 \\
 k_2\times k_1& k_2\times k_2
\end{bmatrix},~\mbox{and}~
\begin{bmatrix}
k_1\times k_1&k_1\times k_2 \\
k_2\times k_1& k_2\times k_2
\end{bmatrix},
$$
respectively, where $k=k_1+k_2$, $\Sigmamat_{21} = \Sigmamat_{12}^T$, and $\Lambdamat_{21} = \Lambdamat_{12}^T$.

A special case that frequently arises in real applications is when $\Gmat_1=\mathbf{0}_{k_2\times k_1}$ and $\Gmat_2= \Imat_{k_2}$, which means $(\mathbf{0}_{k_2\times k_1},\Imat_{k_2})\xv=\xv_2=\rv$ and 
the need is to simulate $\xv_1$ 
given $\xv_2=\rv$. 
For a MVN random variable 
$
\xv\sim\mathcal{N}(\muv,\Sigmamat)
$, it is well known, $e.g.$, in \citet{tong2012multivariate}, that the conditional distribution of $\xv_1$ given $\xv_2=\rv$, $i.e.$, the distribution of $\xv$ restricted to $\mathcal{S}=\{\xv: (\mathbf{0}_{k_2\times k_1},\Imat_{k_2})\xv = \rv\}$, can be expressed as
\beq\label{eq:conditional}
\xv_1 \given \xv_2 = \rv \,\sim\,\Nor \left[\muv_1 + \Sigmamat_{12}\Sigmamat_{22}^{-1}(\rv-\muv_2), ~~\Sigmamat_{11} - \Sigmamat_{12}\Sigmamat_{22}^{-1}\Sigmamat_{21}\right].
\eeq
%
Alternatively, applying the Woodbury matrix identity to relate the entries of the covariance matrix $\Sigmamat$ to those of the precision matrix $\Lambdamat$, one may obtain the following equivalent expression as 
\beq\label{eq:conditionalPrecision}
\xv_1 \given \xv_2 = \rv \,\sim\,\Nor \left[\muv_1 - \Lambdamat_{11}^{-1} \Lambdamat_{12} (\rv-\muv_2), ~\Lambdamat_{11}^{-1}\right].
\eeq


In a general setting where $\Gmat\neq ( \mathbf{0}_{k_2\times k_1}, \Imat_{k_2})$,
let us define a full rank linear transformation matrix $\Hmat\in\mathbb{R}^{k\times k}$, with $(\Hmat_1,\Hmat_2)$ as the $(k\times k_1,k\times k_2)$ partition of $\Hmat$, where the columns of $\Hmat_1\in\mathbb{R}^{k\times k_1}$ span the null space of the $k_2$ rows of $\Gmat$, making $\Gmat\Hmat=(\Gmat\Hmat_1,\Gmat\Hmat_2)=(\mathbf{0}_{k_2\times k_1},\Gmat\Hmat_2)$, where $\Gmat\Hmat_2$ is a $k_2\times k_2$ full rank matrix. \linebreak[4]
For example, a linear transformation matrix $\Hmat$ that makes $\Gmat\Hmat=(\mathbf{0}_{k_2\times k_1},\Imat_{k_2})$ can be constructed using the  command $\mathrm{\Hmat= inv([null(\Gmat)';\Gmat]}$) in Matlab and \linebreak[4]
$\mathrm{\Hmat<\!\!-solve(rbind(t(Null(t(\Gmat))),\Gmat))}$ in R. With $\Hmat $ and $\Hmat^{-1}$, one may  re-express the constraints as $\mathcal{S}=\{\xv: (\mathbf{0}_{k_2\times k_1},\Gmat\Hmat_2) (\Hmat^{-1}\xv) = \rv\}$.
Denote $\zv=\Hmat^{-1}\xv$, then we can generate $\xv$ by letting $\xv=\Hmat \zv$, where
\beq \label{eq:pztruccond}
 \zv\sim\Nor_{\mathcal{D}}[\Hmat^{-1}\muv, \Hmat^{-1}\Sigmamat (\Hmat^{-1})^T],~~ \mathcal{D}=\{\zv: \Gmat\Hmat_2\zv_2=\rv\}=\{\zv: \zv_2=(\Gmat\Hmat_2)^{-1}\rv\}.
\eeq
More specifically, denoting $\Lambdamat = [\Hmat^{-1}\Sigmamat(\Hmat^{-1})^T]^{-1} = \Hmat^T \Sigmamat^{-1}\Hmat$ as the precision matrix for $\zv$,  we have
\begin{align}\label{eq:precision}
\begin{bmatrix}
 \Lambdamat _{11} &\Lambdamat _{12} \\
 \Lambdamat _{21} & \Lambdamat _{22}
\end{bmatrix}
&= \Hmat^T\Sigmamat^{-1} \Hmat = \begin{bmatrix}
 \Hmat_1^T\Sigmamat^{-1}\Hmat_1 &\Hmat^T_1\Sigmamat^{-1} \Hmat_2\\
 \Hmat_2^T\Sigmamat^{-1}\Hmat_1 &\Hmat_2^T\Sigmamat^{-1} \Hmat_2
\end{bmatrix},
\end{align}
and hence $\xv$ truncated on $\mathcal{S}$ can be naively generated using the following algorithm, whose computational complexity is described in Table \ref{tab:CompAlg1} of the Appendix. 

\begin{algorithm}[H]
\begin{itemize}
\item Find $\Hmat=(\Hmat_1,\Hmat_2)$ that satisfies $\Gmat\Hmat=(\Gmat\Hmat_1,\Gmat\Hmat_2)=(\mathbf{0}_{k_2\times k_1},\Gmat\Hmat_2)$, where $\Gmat\Hmat_2$ is a full rank matrix;
\item Let $\zv_2 =(\Gmat \Hmat_2)^{-1}\rv$, $\Lambdamat _{11} = \Hmat_1^T\Sigmamat^{-1}\Hmat_1$, and $\Lambdamat _{12}=\Hmat_1^T\Sigmamat^{-1}\Hmat_2$;
\item Sample $\zv_1\given \zv_2=(\Gmat \Hmat_2)^{-1}\rv \sim\Nor(
\muv_{z_1}, \Lambdamat^{-1}_{11})$, where
$$\muv_{z_1}=(\Imat_{k_1},\mathbf{0}_{k_1\times k_2})\Hmat^{-1}\muv - \Lambdamat^{-1}_{11}\Lambdamat_{12}\left[(\Gmat\Hmat_2)^{-1}\rv - (\mathbf{0}_{k_2\times k_1},\Imat_{k_2})\Hmat^{-1}\muv\right];$$
\item Return $\xv = \Hmat\zv=\Hmat_1 \zv_1 + \Hmat_2 (\Gmat \Hmat_2)^{-1}\rv $.
\end{itemize}
 \caption{\label{alg:1} Simulation of the hyperplane truncated MVN distribution $\xv\sim\Nor_{\mathcal{S}}(\muv,\Sigmamat)$, where $\mathcal{S}=\{\xv: \Gmat \xv = \rv\}$, by transforming a random variable drawn from the conditional distribution of another MVN distribution.}
\end{algorithm}\vspace{2mm}

For illustration, we consider a simple 2-dimensional example with $\muv = (1,1.2)^T$,
$\Sigmamat = [(1,0.3)^T,(0.3,1)^T]
$, $\Gmat = (1,1)$, and $\rv = 1$. If we choose $\Hmat_1 = (-0.7071,0.7071)^T$ and $\Hmat_2=(1.3,1.3)^T$, then we have $z_2 = (\Gmat \Hmat_2)^{-1}\rv = (2.6)^{-1} = 0.3846$, $\Lambdamat_{11} = 1.4285$, and $\Lambdamat_{12}=0$; as shown in Figure~\ref{fig:PxPz}, we may generate $\xv$ using 
$$\xv = (-0.7071,0.7071)^T z_1 + (1.3,1.3)^Tz_2$$
where 
$z_1\sim\Nor(0.1414,0.7)$ and $z_2=0.3846$.
%

\begin{figure}[!t]
 \centering
 \subfigure[]{
 \includegraphics[height = 4.5 cm]{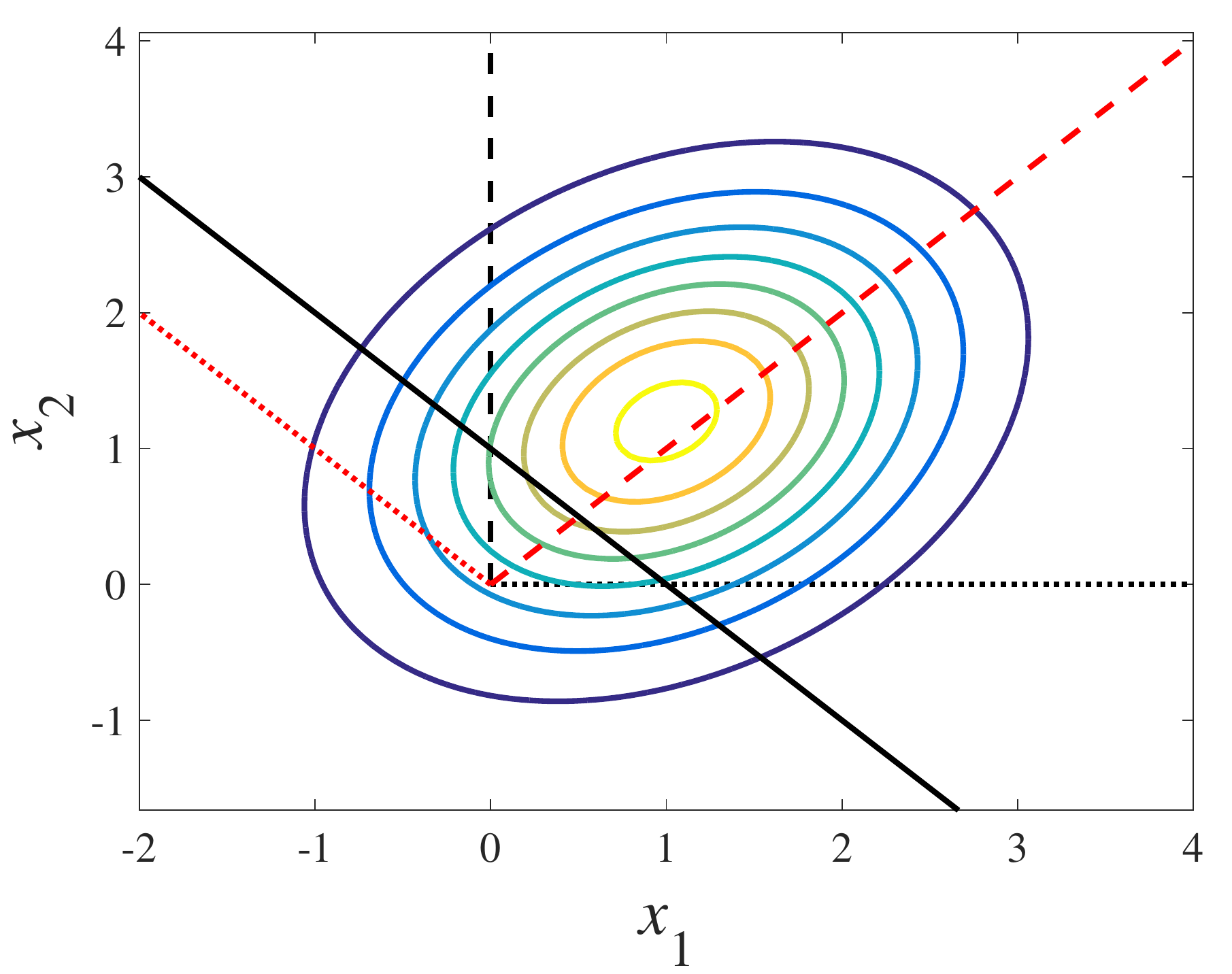} 
 }
 \subfigure[]{
 \includegraphics[height = 4.5 cm]{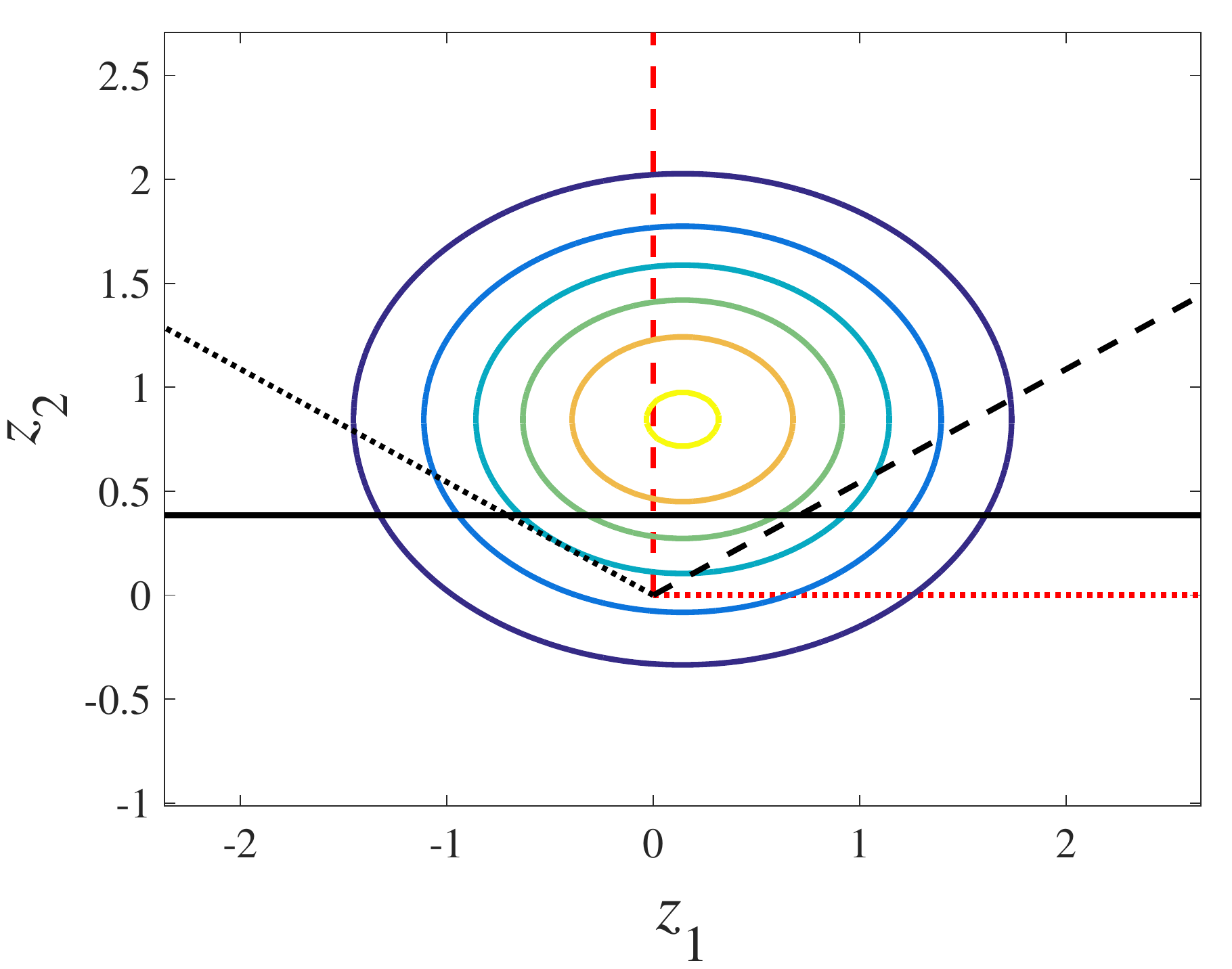} 
 }
 \caption{\small Illustration of (a) $p(\xv)$ in \eqref{eq:px}, where $\muv = (1,1.2)^T$,
$\Sigmamat = [(1,0.3)^T,(0.3,1)^T]
$, $\Gmat = (1,1)$, and $\rv = 1$, and (b) $p(\zv)$ in \eqref{eq:pztruccond}, where $\Hmat_1 = (-0.7071,0.7071)^T$, $\Hmat_2=(1.3,1.3)^T$, and $\Hmat^{-1}=[(-0.7071,0.3846)^T, (0.7071,0.3846)^T]$. The coordinate systems of $\xv$ and $\zv$ are shown in black and red, respectively, and the first and second axes of a coordinate system are shown as dotted and dashed lines, respectively.}
 \label{fig:PxPz}
\end{figure}

For high dimensional problems, however, Algorithm \ref{alg:1} in general requires a large number of intermediate variables that could be computationally expensive to compute. In the following discussion, we will show how to completely avoid instantiating these intermediate variables.

\section{Fast and exact simulation of MVN distributions }

Instead of using Algorithm \ref{alg:1},
we first provide a theorem to show how to efficiently and exactly simulate from a hyperplane-truncated MVN distribution. In the Appendix, we provide two different proofs. The first proof facilitates the derivations by 
employing an existing algorithm of \citet{hoffman1991constrained} and \citet{doucet2010note}, which describes how to simulate from the conditional distribution of a MVN distribution shown in \eqref{eq:conditional} 
without computing $\Sigmamat_{11} - \Sigmamat_{12}\Sigmamat_{22}^{-1}\Sigmamat_{21}$
and its Cholesky decomposition. 
Note it is straightforward to verify that the algorithm in \citet{hoffman1991constrained} and \citet{doucet2010note}, 
as shown in the Appendix,  can be considered as  
a special case of the proposed algorithm with $\Gmat = [\bds 0, \Imat]$.

\vspace{2mm}
\begin{algorithm}[H]
\begin{itemize}
\item Sample $\yv\sim\Nor (\muv,\Sigmamat)$;
\item Return $\xv = \yv + \Sigmamat \Gmat^T ( \Gmat \Sigmamat \Gmat^T )^{-1} (\rv - \Gmat \yv)$, which can be realized using \begin{itemize}
\item Solve $\alphav$ such that $( \Gmat \Sigmamat \Gmat^T ) \alphav = \rv - \Gmat \yv$;
\item Return $\xv = \yv + \Sigmamat \Gmat^T \alphav$. 
\end{itemize}
\end{itemize}
 \caption{\label{alg:2} Simulation of the hyperplane truncated MVN distribution $\xv\sim\Nor_{\mathcal{S}}(\muv,\Sigmamat)$, where $\mathcal{S}=\{\xv: \Gmat \xv = \rv\}$, by transforming a random variable drawn from $\yv\sim\Nor(\muv,\Sigmamat)$. }
\end{algorithm}\vspace{2mm}

\begin{thm} \label{maintheorem1}
Suppose $\xv$ is simulated with Algorithm \ref{alg:2}, then it is distributed as
$
\xv\sim\Nor_{\mathcal{S}}(\muv,\Sigmamat),~~\mathcal{S}=\{\xv: \Gmat \xv = \rv\},
$
where
$
\Gmat\in\mathbb{R}^{k_2\times k}, ~~\rv\in\mathbb{R}^{k_2},
$
and $\mbox{Rank}(\Gmat)=k_2<k$.
%
\end{thm}

The above algorithm and theorem, whose computational complexity is described in Table \ref{tab:CompAlg2} of the Appendix, show that one may draw $\yv$ from the unconstrained MVN as $\yv\sim\Nor ( \muv , \Sigmamat )$ and directly map it to a vector $\xv$ on the intersection of hyperplanes using
$ 
\xv 
=\Sigmamat \Gmat^T (\Gmat \Sigmamat \Gmat^T )^{-1} \rv +\left[ \Imat-\Sigmamat \Gmat^T ( \Gmat \Sigmamat \Gmat^T )^{-1}\Gmat\right]\yv.\notag
$ 
For illustration, with the same $\muv$, $\Sigmamat$, $\Gmat$, and $\rv$ as those in Figure \ref{fig:PxPz}, we show in Figure \ref{fig:DemoAlg2} a simple two dimensional example, where the unrestricted Gaussian distribution $ \Nor ( \muv , \Sigmamat )$ is represented with a set of ellipses, and 
the constrained sample space $\mathcal{S}$ is represented as a straight line in the two-dimensional setting.
With $\Sigmamat \Gmat^T (\Gmat \Sigmamat \Gmat^T )^{-1}  \rv = (0.5,0.5)^T$, $\left[ \Imat-\Sigmamat \Gmat^T ( \Gmat \Sigmamat \Gmat^T )^{-1}\Gmat\right] = [(0.5,-0.5)^T,(-0.5,0.5)^T]$, 
one may directly maps a sample $\yv \sim \Nor ( \muv , \Sigmamat )$ to a vector on the constrained space. For example, if $\yv = (1,2)^T$, then it would be mapped to $\xv = (0,1)^T$ on the straight line.

\begin{figure}[h] 
  \centering
  \includegraphics[height = 5 cm]{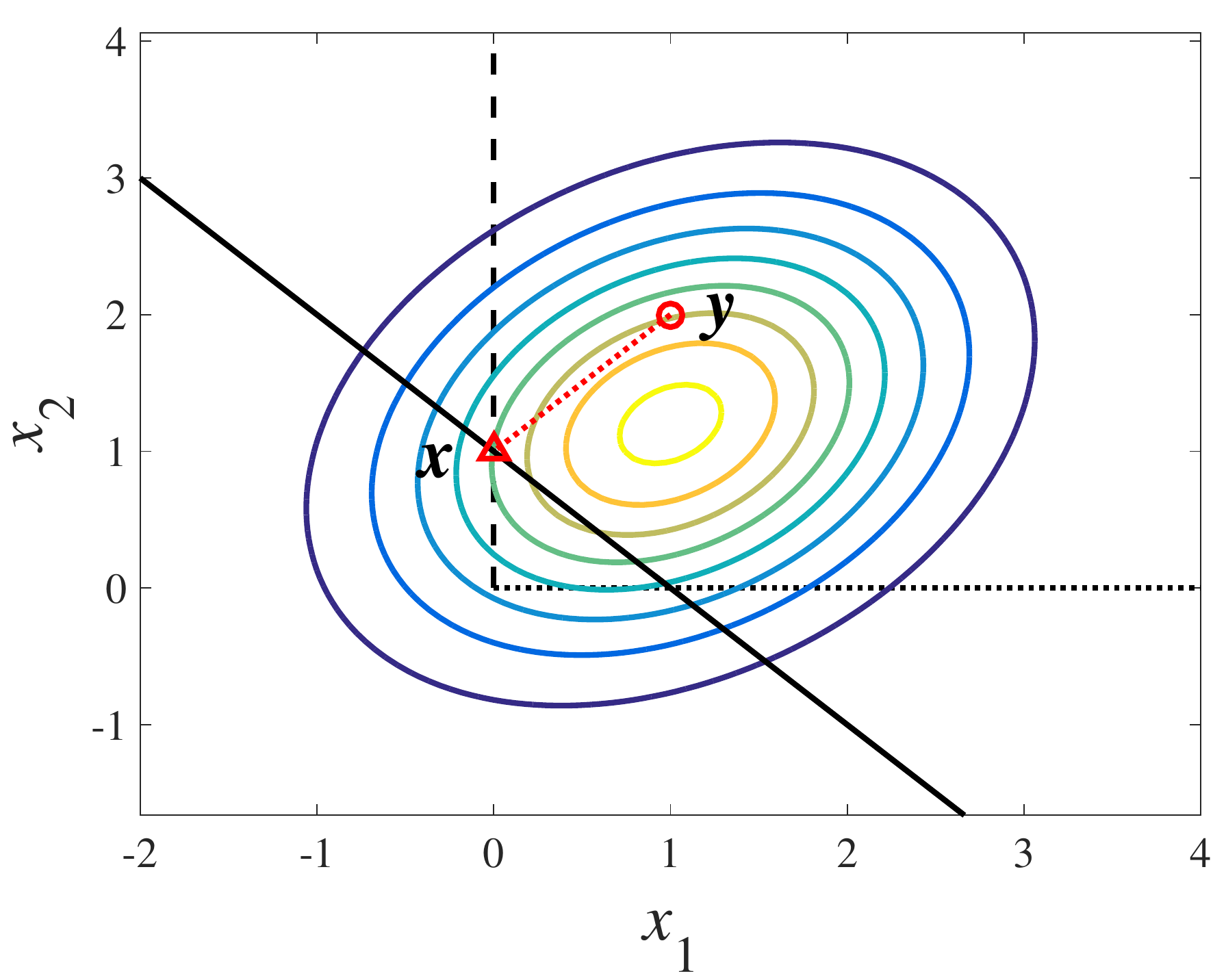}\\
  \caption{A two dimensional demonstration of Algorithm \ref{alg:2} that maps a random sample from $\yv\sim\Nor ( \muv , \Sigmamat )$  to a sample in the constrained space using $ 
\xv 
=\Sigmamat \Gmat^T (\Gmat \Sigmamat \Gmat^T )^{-1} \rv +\left[ \Imat-\Sigmamat \Gmat^T ( \Gmat \Sigmamat \Gmat^T )^{-1}\Gmat\right]\yv
$.  For example, if $\muv = (1,1.2)^T$,
  $\Sigmamat = [(1,0.3)^T,(0.3,1)^T]$,
    $\Gmat = (1,1)$, and $\rv = 1$, then
      $\yv = (1,2)^T$ would be mapped to $\xv=(0,1)^T$ on a straight line using Algorithm \ref{alg:2}.
}
  \label{fig:DemoAlg2}
\end{figure}

\subsection{Fast simulation of MVN distributions with structured covariance or precision matrices}
For fast simulation of MVN distributions with structured covariance or precision matrices, our idea is to relate them to higher-dimensional 
hyperplane-truncated MVN distributions, with block-diagonal covariance matrices, that can be efficiently simulated with Algorithm \ref{alg:2}.
We first introduce an efficient algorithm for the simulation of a MVN distribution, whose covariance matrix is a positive-definite matrix subtracted by a low-rank symmetric matrix. Such kind of covariance matrices commonly arise in the conditional distributions of MVN distributions, as shown in~\eqref{eq:conditional}. 
We then further extend this algorithm to the simulation of a MVN distribution whose precision (inverse covariance) matrix is 
the sum of a positive-definite matrix and a low-rank symmetric matrix. Such kind of precision matrices commonly arise in the conditional posterior distributions of the regression coefficients in both linear regression and generalized linear models. 


\begin{thm}\label{thm_condition}
The probability density function (PDF) of the MVN distribution
\beq\label{eq:px1}
\xv_1\sim
 \Nor ( \muv_1,~ \Sigmamat_{11} - \Sigmamat_{12}\Sigmamat_{22}^{-1}\Sigmamat_{21}),
\eeq
is the same as the PDF of the marginal distribution of $\xv_1=(x_1,\ldots,x_{k_1})^T$
in $\xv=(\xv_1^T,x_{k_1+1},\ldots,x_{k})^T$, whose PDF is expressed as
\begin{align}
p( \xv\given \muv,\tilde \Sigmamat,\Gmat,\rv ) &= \Nor_{\{ \xv: \Gmat \xv = \rv\}}\left(
 \muv,
\tilde \Sigmamat \right)
 \notag\\
 &=\frac{1}{Z} \exp\left[-\frac{1}{2}(\xv- \muv)^T\tilde \Sigmamat^{-1} (\xv- \muv) \right] \delta(\Gmat\xv = \rv),
\end{align}
where $Z$ is a normalization constant, $\Gmat_1 = \Sigmamat_{21}\Sigmamat_{11}^{-1} $ is a matrix of size $k_2\times k_1$,
$\Gmat_2 $ is a user-specified full rank invertible matrix of size $k_2\times k_2$, $\rv\in\mathbb{R}^{k_2}$ is a user-specified vector, 
and
\beq\label{eq:14}
 \Gmat=(\Gmat_1,\Gmat_2) 
\in\mathbb{R}^{k_2\times k},~~~~
 \muv = \begin{bmatrix}
 \muv_1\\
 \muv_2
 \end{bmatrix}\in\mathbb{R}^{k},
 ~~~~\tilde \Sigmamat = \begin{bmatrix}
 \Sigmamat_{11} & \mathbf{0} \\
\mathbf{0} & \tilde \Sigmamat_{22}
 \end{bmatrix}\in\mathbb{R}^{k\times k},
 \eeq
 where
 \begin{align}\label{eq:tilde_mu2}
 & \muv_2= \Gmat_{2}^{-1}(\rv - \Sigmamat_{21}\Sigmamat_{11}^{-1} \muv_1),\\
 \label{eq:tilde_sig22}
 &\tilde \Sigmamat_{22} = \Gmat_2^{-1}\left(\Sigmamat_{22} - \Sigmamat_{21} \Sigmamat_{11}^{-1} \Sigmamat_{12}\right)(\Gmat_{2}^{-1})^T. 
 \end{align}
\end{thm}

The above theorem shows how the simulation of a MVN distribution, whose covariance matrix is a positive-definite matrix minus a symmetric  matrix, can be realized by the simulation of a higher-dimensional hyperplane-truncated MVN distribution. By construction, it makes the covariance matrix $\tilde \Sigmamat$ of the truncated-MVN be block diagonal, but still preserves the flexibility  to customize the full-rank matrix $\Gmat_2$ and the vector~$\rv$. While there are infinitely many choices for 
both $\Gmat_2$ and $\rv$, 
in the following discussion, we remove that flexibility by specifying $\Gmat_2 = \Imat_{k_2} $, leading to $\Gmat=(\Gmat_1,\Gmat_2)=(\Sigmamat_{21}\Sigmamat_{11}^{-1} ,\Imat_{k_2})$, and $\rv = \Sigmamat_{21}\Sigmamat_{11}^{-1} \muv_1$. This specific  setting of $\Gmat_2$ and $\rv$ leads to the following Corollary that is a special case of Theorem~\ref{thm_condition}. Note that while we choose this specific setting in the paper, depending on the problems under study, other settings may lead to even more efficient simulation algorithms. 

\begin{cor}\label{cor1}
The PDF of the MVN distribution
\beq
\xv_1 \sim \Nor ( \muv_1, \Sigmamat_{11} - \Sigmamat_{12}\Sigmamat_{22}^{-1}\Sigmamat_{21})
\eeq
is the same as the PDF of the marginal distribution of $\xv_1$
in $\xv=(\xv_1^T,x_{k_1+1},\ldots,x_{k})^T$, whose PDF is expressed as
\begin{align}
p( \xv) &= \Nor_{\{ \xv: \, \Sigmamat_{21}\Sigmamat_{11}^{-1}\xv_1 + \xv_2 = \Sigmamat_{21}\Sigmamat_{11}^{-1} \muv_1\}}\left(
 \muv,
\tilde \Sigmamat \right)
 \notag\\
 &=\frac{1}{Z} \exp\left[-\frac{1}{2}(\xv- \muv)^T\tilde \Sigmamat^{-1} (\xv- \muv) \right] \delta( \Sigmamat_{21}\Sigmamat_{11}^{-1}\xv_1 + \xv_2 = \Sigmamat_{21}\Sigmamat_{11}^{-1} \muv_1),
\end{align}
where $\xv_2 = (x_{k_1+1},\ldots,x_{k})^T$, $Z$ is a normalization constant, 
and
\beq 
 \muv = \begin{bmatrix}
 \muv_1\\
\mathbf{0}
 \end{bmatrix}\in\mathbb{R}^{k},
 ~~~~\tilde \Sigmamat = \begin{bmatrix}
 \Sigmamat_{11} & \mathbf{0} \\
\mathbf{0} & \Sigmamat_{22} - \Sigmamat_{21} \Sigmamat_{11}^{-1} \Sigmamat_{12}
 \end{bmatrix}\in\mathbb{R}^{k\times k}.
 \eeq
 \end{cor}

%

Further applying  Theorem \ref{maintheorem1}
to Corollary \ref{cor1}, as described in detail in the Appendix, a MVN random variable $\xv$ with a structured covariance matrix can be generated as in Algorithm \ref{alg:3}, where there is no need to compute $\Sigmamat_{11}-\Sigmamat_{12}\Sigmamat_{22}^{-1}\Sigmamat_{21}$ and its Cholesky decomposition.
Suppose the covariance matrix $\Sigmamat_{11}$ admits some special structure that makes it easy to invert and computationally efficient to simulate from $\Nor (\mathbf{0},\Sigmamat_{11})$, 
then Algorithm \ref{alg:3} could lead to a significant saving in computation if $k_2\ll k_1$.
On the other hand, when $k_2\gg k_1$ and $\Sigmamat_{22}- \Sigmamat_{21} \Sigmamat_{11}^{-1} \Sigmamat_{12}$ admits no special structures, Algorithm \ref{alg:3} may not bring any computational advantage and  hence one may resort to the naive Cholesky decomposition based procedure. 
Detailed computational complexity analyses for both methods are provided in Tables \ref{tab:CompNaive3} and \ref{tab:CompAlg3} of the Appendix, respectively. 

\vspace{2mm}
\begin{algorithm}[H]
\begin{itemize}
\item Sample $\yv_1\sim\Nor (\mathbf{0},\Sigmamat_{11})$ and $\yv_2\sim\Nor (\mathbf{0},\Sigmamat_{22}- \Sigmamat_{21} \Sigmamat_{11}^{-1} \Sigmamat_{12})$ ;
\item Return $\xv_1 = \muv_1+\yv_1 - \Sigmamat_{12} \Sigmamat_{22}^{-1} (\Sigmamat_{21} \Sigmamat_{11}^{-1}\yv_1 + \yv_2)$, which can be realized using
\begin{itemize}
\item Solve $\alphav$ such that $ \Sigmamat_{22}\alphav = \Sigmamat_{21} \Sigmamat_{11}^{-1}\yv_1 + \yv_2$;
\item Return $\xv_1 = \muv_1+\yv_1 - \Sigmamat_{12} \alphav$. 
\end{itemize}

\end{itemize}
 \caption{\label{alg:3} Simulation of the MVN distribution $$\xv_1\sim\Nor(\muv_1,\Sigmamat_{11}-\Sigmamat_{12}\Sigmamat_{22}^{-1}\Sigmamat_{21}).$$ }
\end{algorithm}
\vspace{2mm}

\begin{cor} \label{cor2}
A random variable simulated with Algorithm \ref{alg:3} is distributed as
$
\xv_1\sim\Nor (\muv_1, \Sigmamat_{11} - \Sigmamat_{12}\Sigmamat_{22}^{-1}\Sigmamat_{21}). 
$

\end{cor}

The efficient simulation algorithm for a MVN distribution with a structured covariance matrix can also be further extended to a MVN distribution with a structured precision matrix, as described below, 
 where $\betav \in \Rbb^p$, $\muv_{\beta} \in \Rbb^{p}$, $\Phimat \in \Rbb^{n \times p}$, and both $\Amat \in \Rbb^{p \times p}$ and $\Omegamat \in \Rbb^{n \times n}$ are positive-definite matrices.
Computational complexity analyses for both the naive Cholesky decomposition based implementation and Algorithm \ref{alg:4} are provided in Table \ref{tab:CompNaive4} and \ref{tab:CompAlg4} of the Appendix, respectively.
Similar to Algorithm~\ref{alg:3}, Algorithm \ref{alg:4} may bring a significant saving in computation when $p \gg n$  and 
$\Amat$ 
admits some special structure that makes it easy to invert and computationally efficient to simulate $\yv_1$. 

\vspace{2mm}
\begin{algorithm}[H]
 \begin{itemize}
\item Sample $\yv_1\sim\Nor (\mathbf{0},\Amat^{-1})$ and $\yv_2\sim\Nor (\mathbf{0},\Omegamat^{-1})$ ;
\item
 Return
$
 \betav 
 =\muv_{\beta}+\yv_1 - \Amat^{-1} \Phimat^T (\Omegamat^{-1}+\Phimat\Amat^{-1}\Phimat^T)^{-1} \left( \Phimat \yv_1+\yv_2\right)$, which can be realized using
\begin{itemize}
\item Solve $\alphav$ such that $ (\Omegamat^{-1}+\Phimat\Amat^{-1}\Phimat^T)\alphav = \Phimat \yv_1+\yv_2$.
\item Return $\betav =\muv_{\beta}+\yv_1 - \Amat^{-1} \Phimat^T \alphav$. 
\end{itemize}

\end{itemize}

\caption{\label{alg:4} Simulation of the MVN distribution 
$$\betav\sim\Nor\left[\muv_{\beta}, (\Amat + \Phimat^T \Omegamat \Phimat)^{-1}\right].$$ 
}
\end{algorithm}
\vspace{2mm}

\begin{cor}\label{cor3}
The random variable obtained with Algorithm \ref{alg:4} is distributed as
$\betav\sim\Nor(\muv_{\beta},\Sigmamat_{\beta})$, where 
$
\Sigmamat_{\beta} = (\Amat + \Phimat^T \Omegamat \Phimat)^{-1}$.
\end{cor}

\section{Illustrations} 

Below we provide several examples to illustrate Theorem \ref{maintheorem1}, which shows how to efficiently simulate from a hyperplane-truncated MVN distribution, and Corollary \ref{cor2} (Corollary \ref{cor3}), which shows how to efficiently simulate from a MVN distribution with a structured covariance (precision) matrix. We run all our experiments 
on a 2.9 GHz computer.


\subsection{Simulation of hyperplane-truncated MVNs}

We first compare Algorithms \ref{alg:1} and \ref{alg:2}, whose generated random samples follow the same distribution, as suggested by Theorem \ref{maintheorem1}, to highlight the advantages of 
Algorithm~\ref{alg:2}
 over Algorithm  \ref{alg:1}. 
 We then employ Algorithm \ref{alg:2} for a real application whose data dimension is high and sample size is large.

\subsubsection{Comparison of Algorithms \ref{alg:1} and \ref{alg:2}}
\label{sec:A1vsA2}

We compare Algorithms \ref{alg:1} and \ref{alg:2} in a wide variety of settings by varying the data dimension $k$, varying the number of hyperplane constraints $k_2$, and choosing either a diagonal covariance matrix  $\Sigmamat$ or a non-diagonal one. We generate random diagonal covariance matrices using the MATLAB command  $\diag(0.05+\text{rand}(k,1))$ and random non-diagonal ones using $U.'*\diag(0.05+\text{rand}(k,1))*U$, where $\text{rand}(k,1)$ is a vector of $k$ uniform random numbers and $U$ consists of a set of $k$ orthogonal basis vectors. 
The elements of $\muv$, $\rv$, and $\Gmat$ are all sampled from $\Nor(0,1)$, 
with 
 the singular value decomposition applied to $\Gmat$  to check whether  $\mbox{Rank}(\Gmat)=k_2$.

First, to verify Theorem \ref{maintheorem1},  we conduct an experiment with $k=5000$ data dimension, $k_2=20$ hyperplanes, and a diagonal $\Sigmamat$. 
Contour plots of two randomly selected dimensions of the 10,000 random samples simulated with Algorithms \ref{alg:1} and \ref{alg:2} are shown in the top and bottom rows of Figure \ref{fig:A1vsA2CutPatch}, respectively.   
The clear matches between the contour plots of these two different algorithms 
 suggest the correctness of Theorem \ref{maintheorem1}.

\begin{figure}[t]
	\centering
	\subfigure[]{
		\includegraphics[width = 2.2 cm]{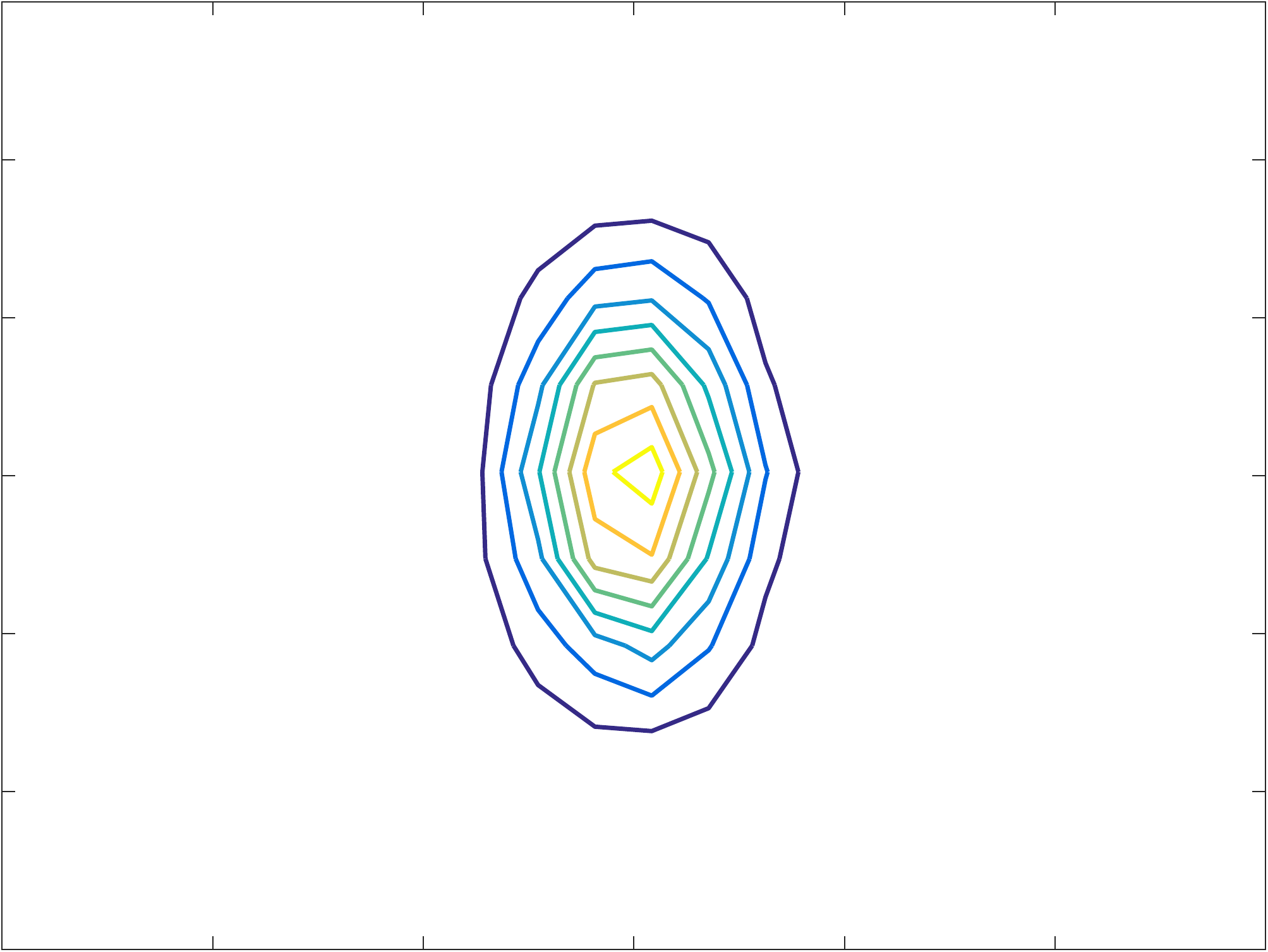}
	}
	\subfigure[]{
		\includegraphics[width = 2.2 cm]{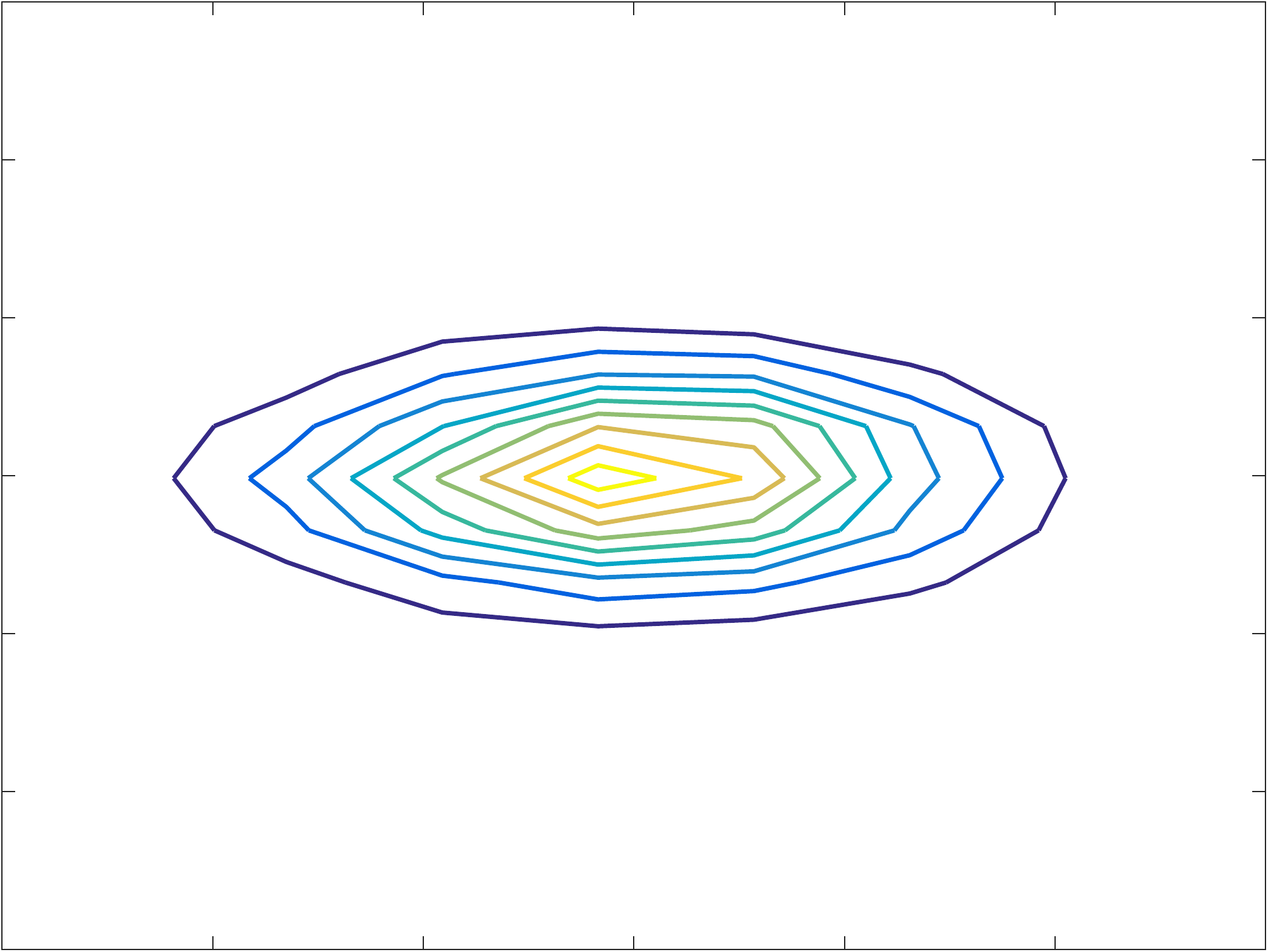}
	}
	\subfigure[]{
		\includegraphics[width = 2.2 cm]{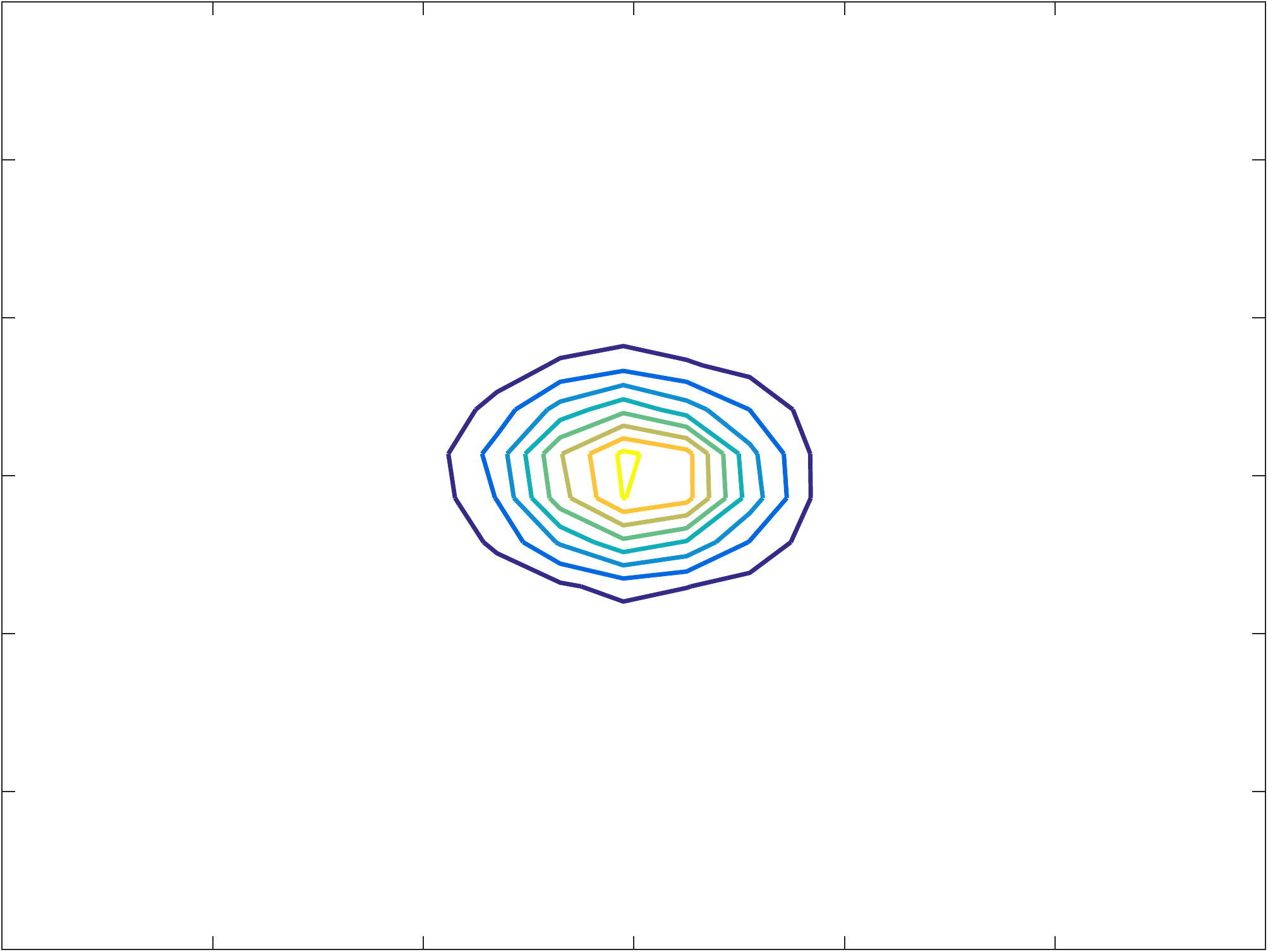}
	}
	\subfigure[]{
		\includegraphics[width = 2.2 cm]{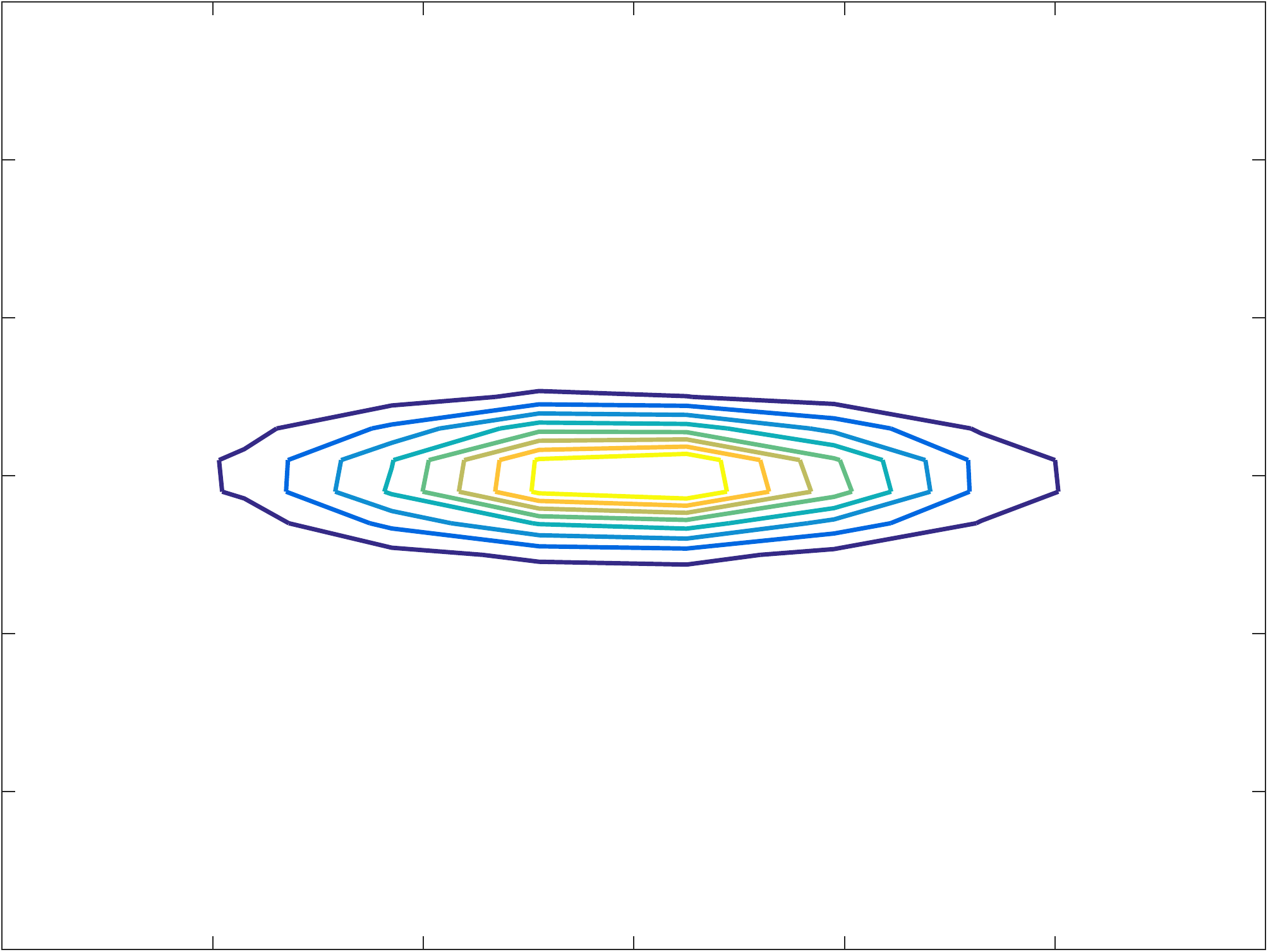}
	}
	\subfigure[]{
		\includegraphics[width = 2.2 cm]{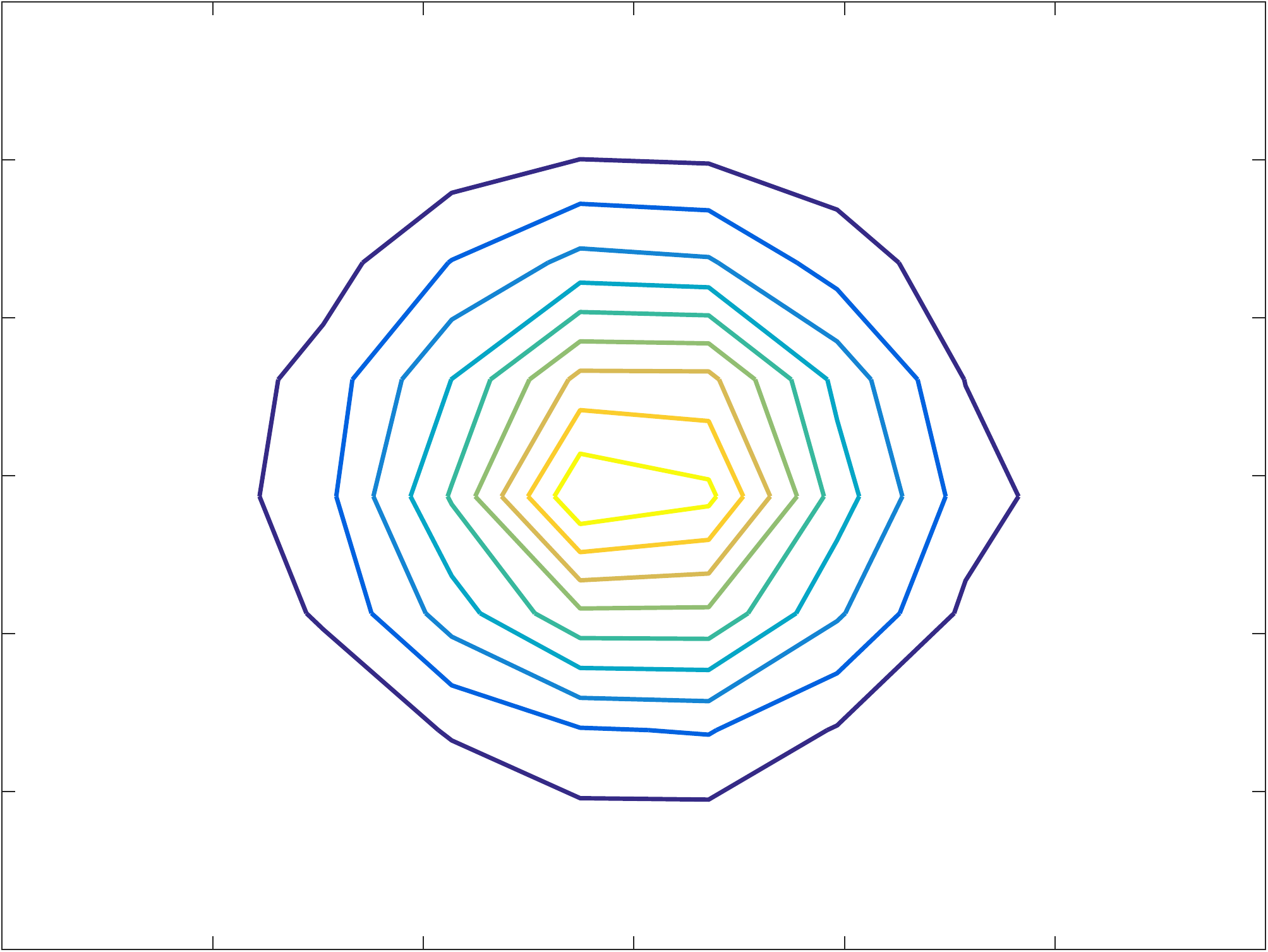}
	}
	\subfigure[]{
		\includegraphics[width = 2.2 cm]{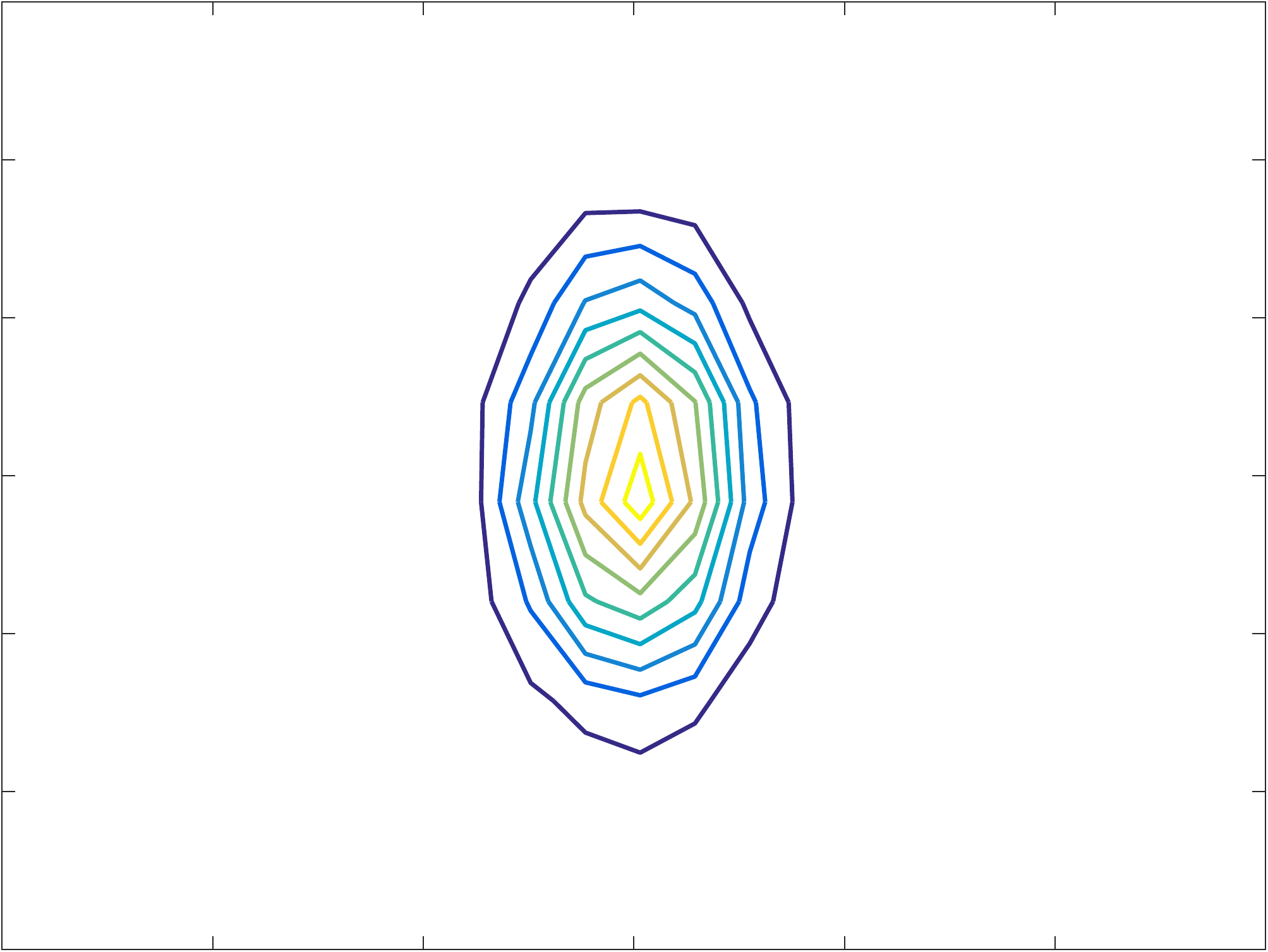}
	}
	\subfigure[]{
		\includegraphics[width = 2.2 cm]{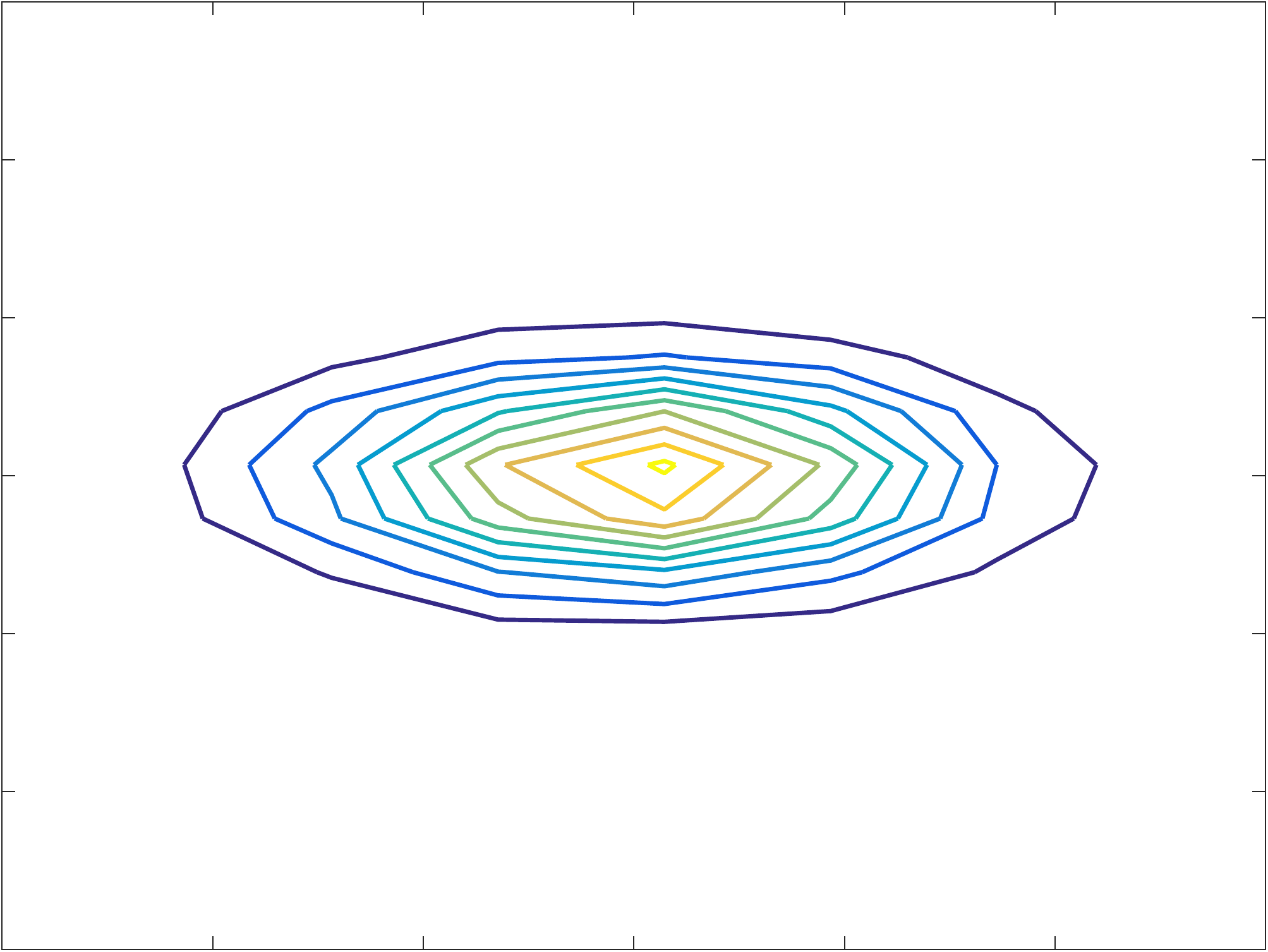}
	}
	\subfigure[]{
		\includegraphics[width = 2.2 cm]{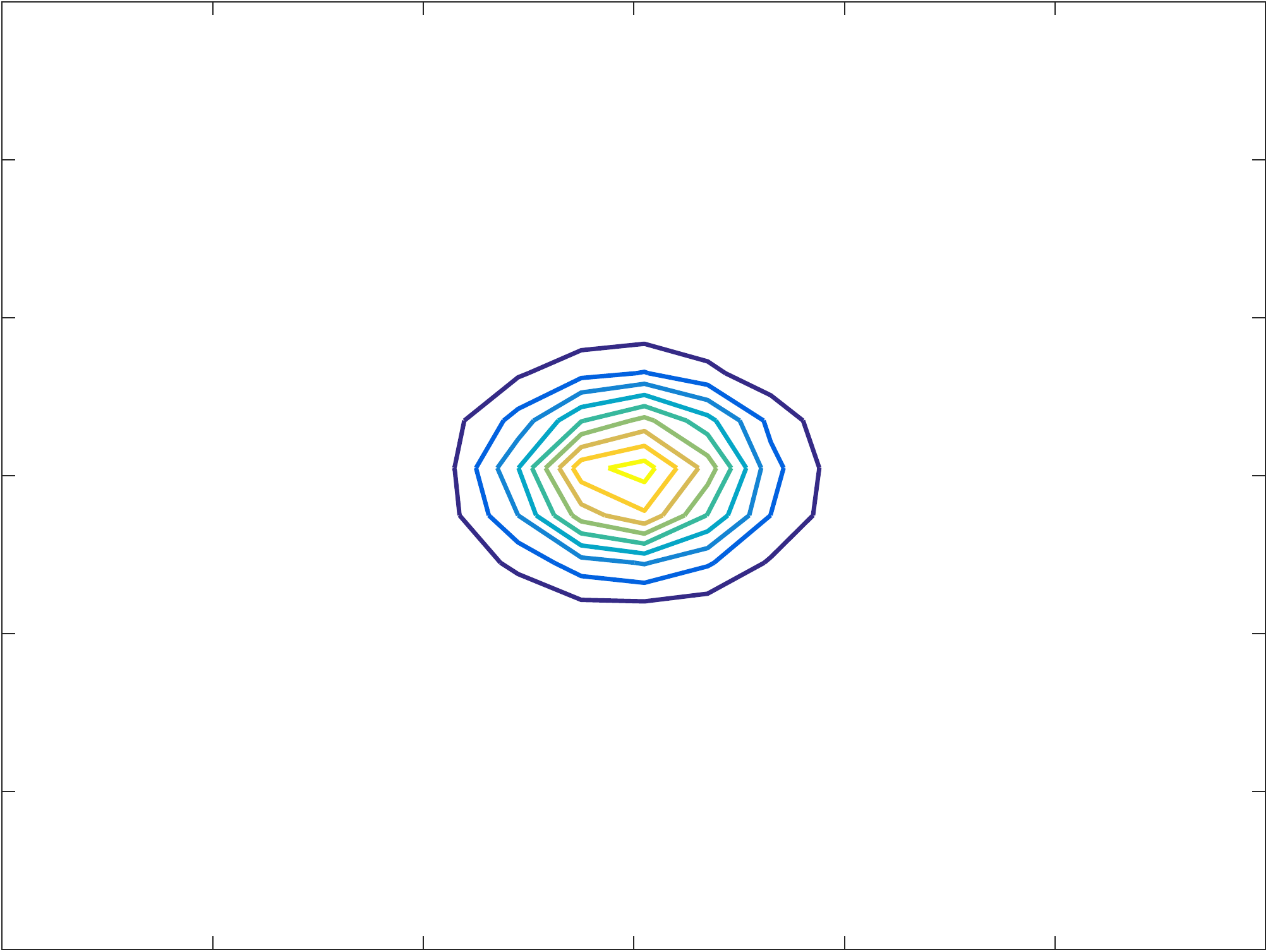}
	}
	\subfigure[]{
		\includegraphics[width = 2.2 cm]{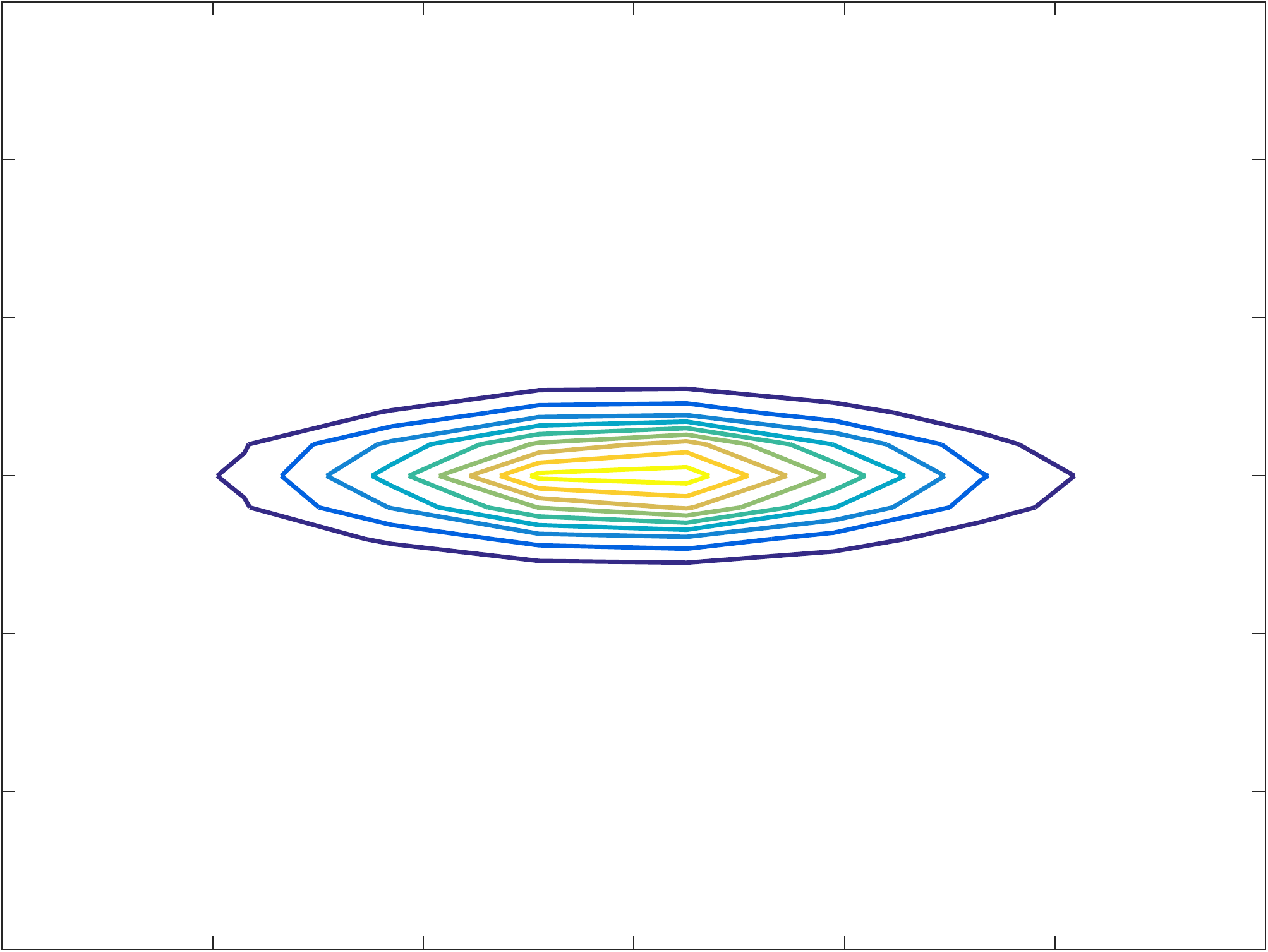}
	}
	\subfigure[]{
		\includegraphics[width = 2.2 cm]{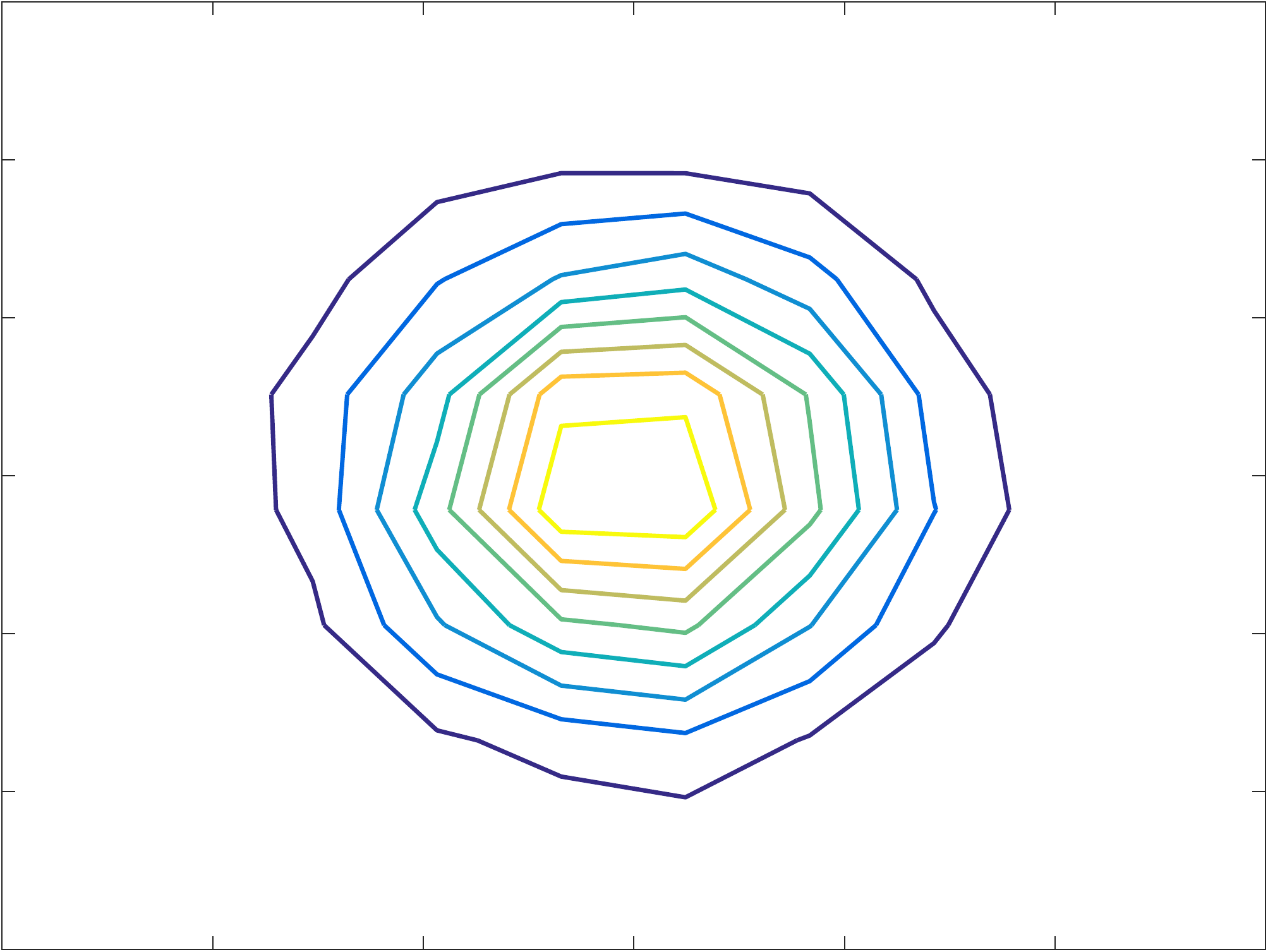}
	}
	\caption{\small Comparison of the contour plots of two randomly selected dimensions of the 10,000 $k=5000$ dimensional random samples simulated with Algorithm \ref{alg:1} (top row) and Algorithm \ref{alg:2} (bottom row). Each of the five columns corresponds to a random trial. 
	}
	\label{fig:A1vsA2CutPatch}
\end{figure}

To demonstrate the efficiency of Algorithm \ref{alg:2}, 
we first carry out a series of experiments 
with the number of hyperplane constraints fixed at $k_2=20$ and the data dimension increased from $k=50$ to $k=5000$.
The computation time of simulating 10,000 samples averaged over five random trials is shown in Figure \ref{fig:A1A2K220GenCov} for non-diagonal $\Sigmamat$'s and in Figure \ref{fig:A1A2K220DiagCov} for diagonal ones. It is clear that, when the data dimension $k$ is high,  Algorithm \ref{alg:2} has a clear advantage over Algorithm \ref{alg:1} by avoiding computing 
unnecessary  intermediate variables, which is especially evident when $\Sigmamat$ is diagonal. 
We then carry out a series of experiments where we vary not only  $k$, but also $k_2$ from $0.1k$ to $0.9k$ for each~$k$. 
As shown in Figure \ref{fig:A1vsA2Time}, 
it is evident that Algorithm \ref{alg:2} dominates 
 Algorithm \ref{alg:1} in all scenarios, which can be explained by the fact that  Algorithm \ref{alg:2}  needs to compute much fewer intermediate variables. 
Also observed is that a larger $k_2$ leads to slower simulation for both algorithms, but to a much lesser extent for Algorithm \ref{alg:2}. 
Moreover, the curvatures of those curves indicate that Algorithm \ref{alg:2} is more practical in a high dimensional setting. 
Note that since Algorithm \ref{alg:2} can naturally exploit the structure of the covariance matrix $\Sigmamat$ for fast simulation, it is clearly more capable of benefiting from having a diagonal or block-diagonal $\Sigmamat$, demonstrated by comparing  Figures \ref{fig:Alg1GenCovVaryK2} and \ref{fig:Alg2GenCovVaryK2} with Figures \ref{fig:Alg1DiagCovVaryK2}  and \ref{fig:Alg2DiagCovVaryK2}. 
All these observations agree with 
our computational complexity analyses for Algorithms  \ref{alg:1}  and  \ref{alg:2}, as  shown in Table \ref{tab:CompAlg1} and  \ref{tab:CompAlg2} of the Appendix, respectively.

\begin{figure}[t]
	\centering
	\subfigure[]{\label{fig:A1A2K220GenCov}
		\includegraphics[width = 4.1 cm]{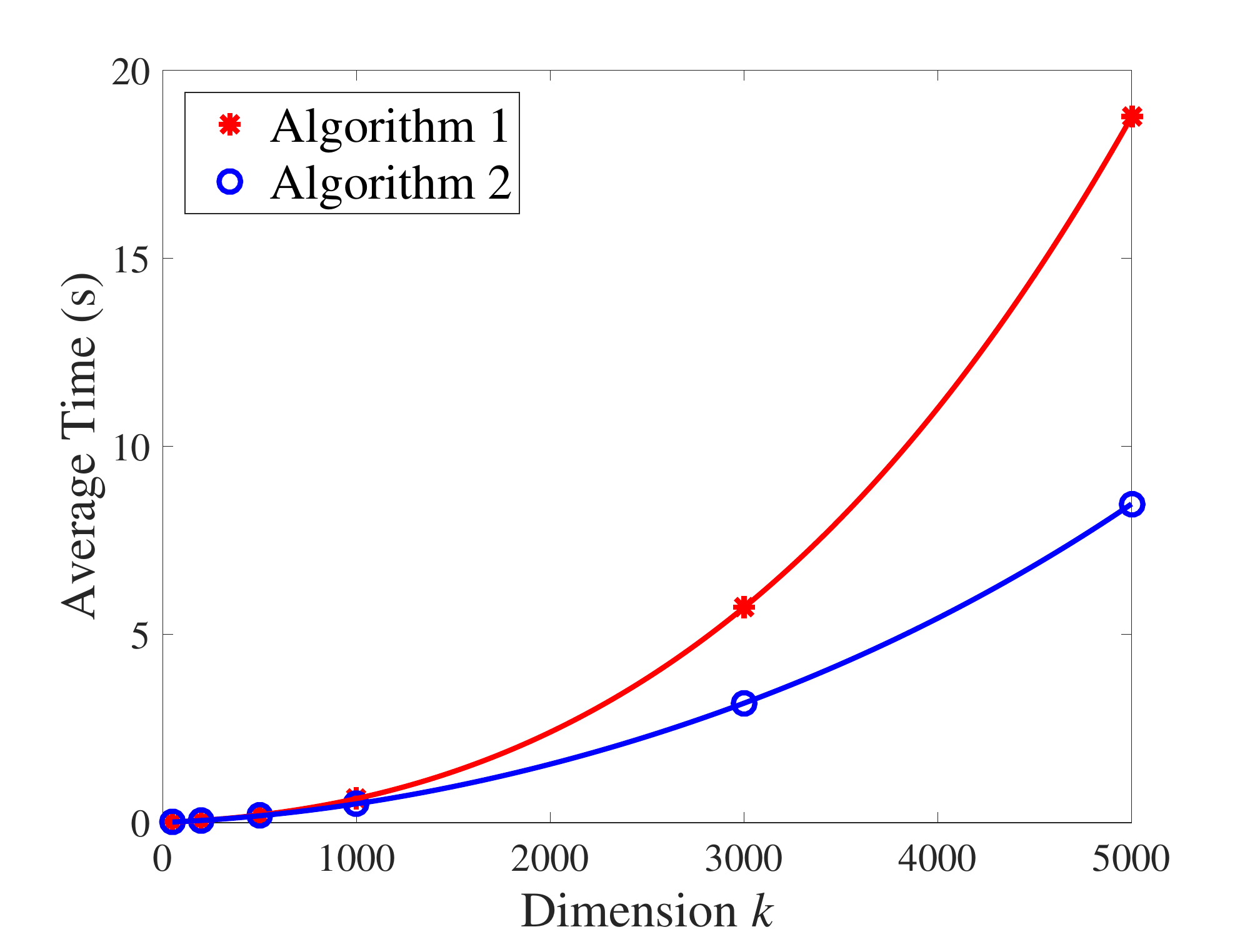}
	}
	\subfigure[]{\label{fig:Alg1GenCovVaryK2}
		\includegraphics[width = 4.1 cm]{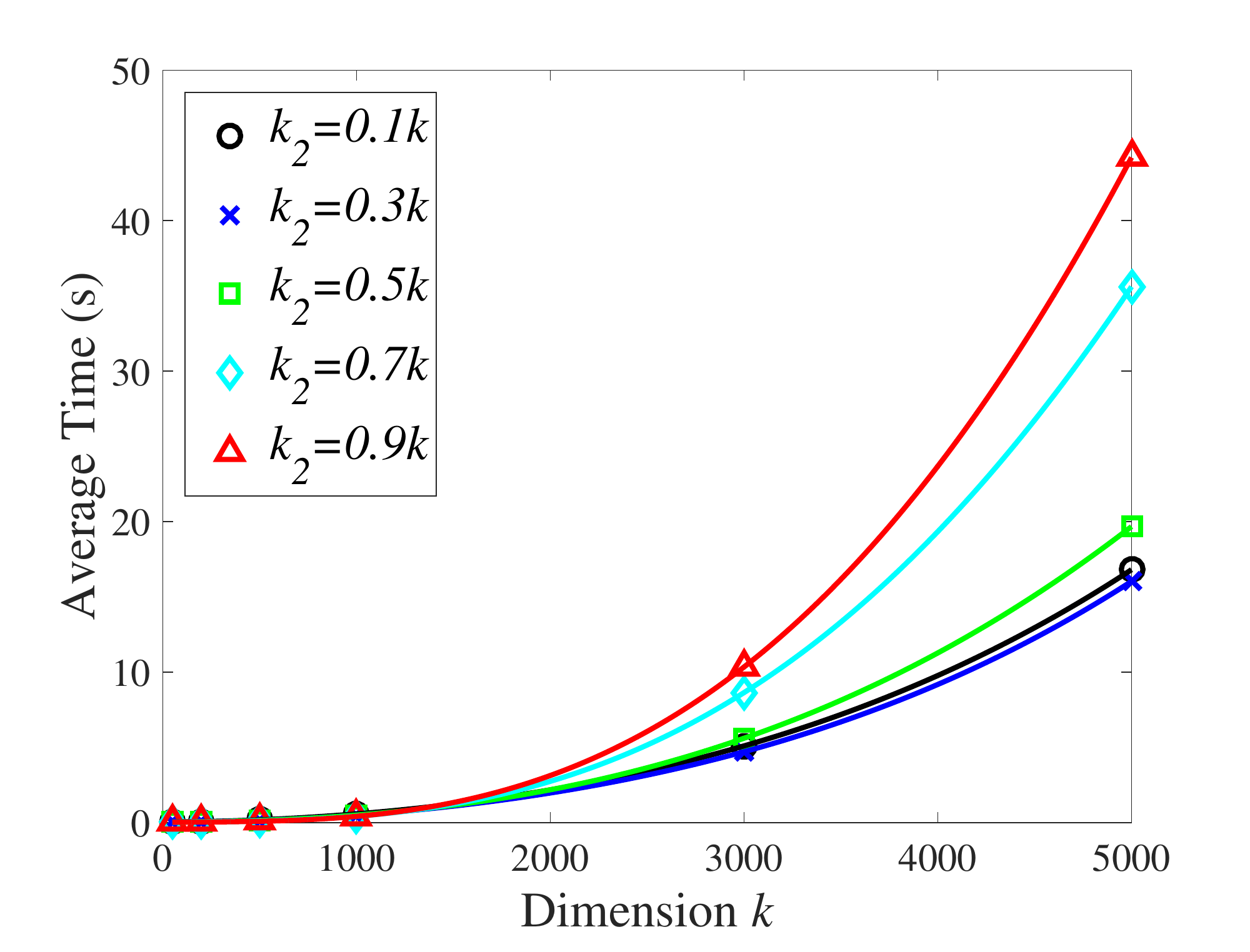}
	}
	\subfigure[]{\label{fig:Alg2GenCovVaryK2}
		\includegraphics[width = 4.1 cm]{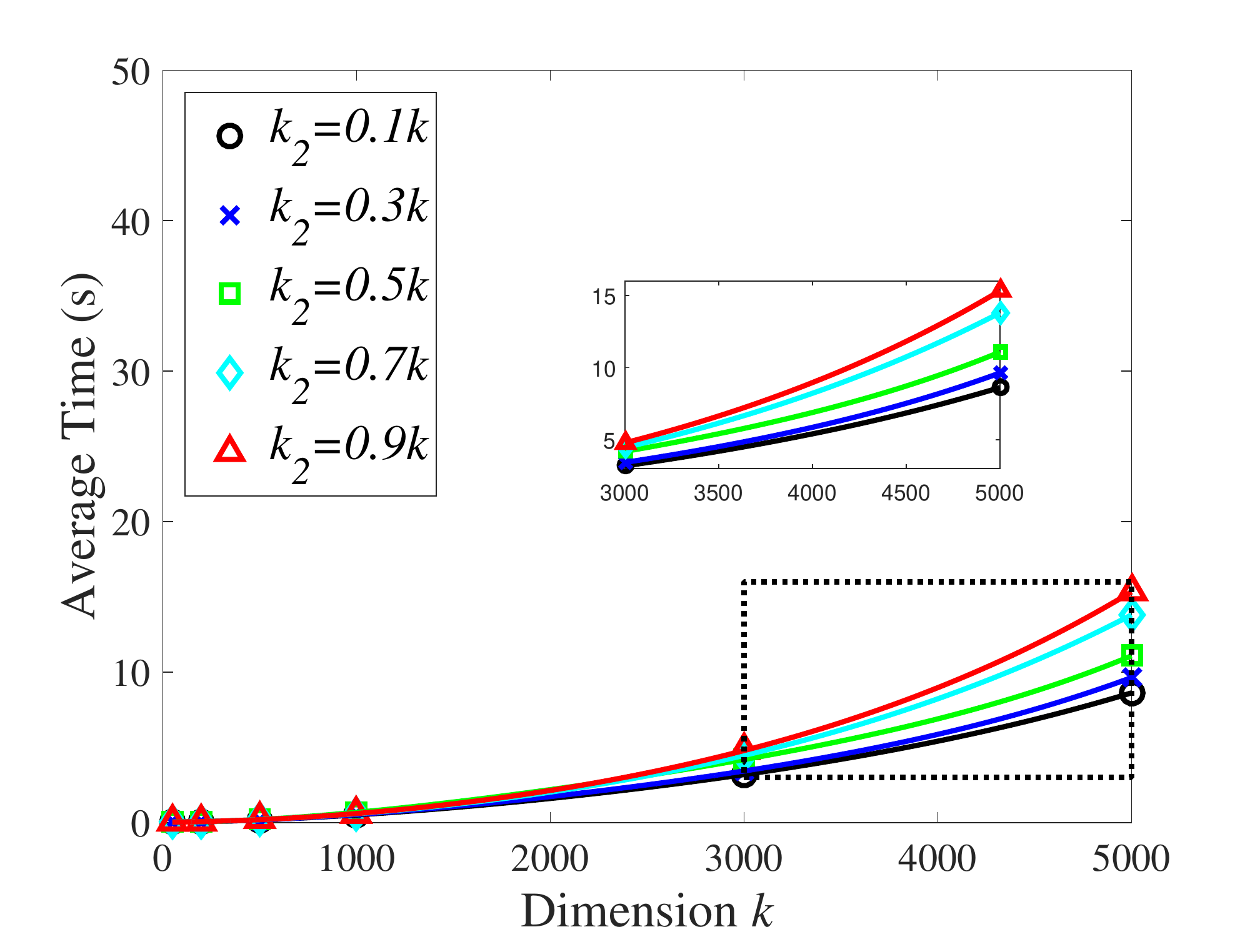}
	}
	\subfigure[]{\label{fig:A1A2K220DiagCov}
		\includegraphics[width = 4.1 cm]{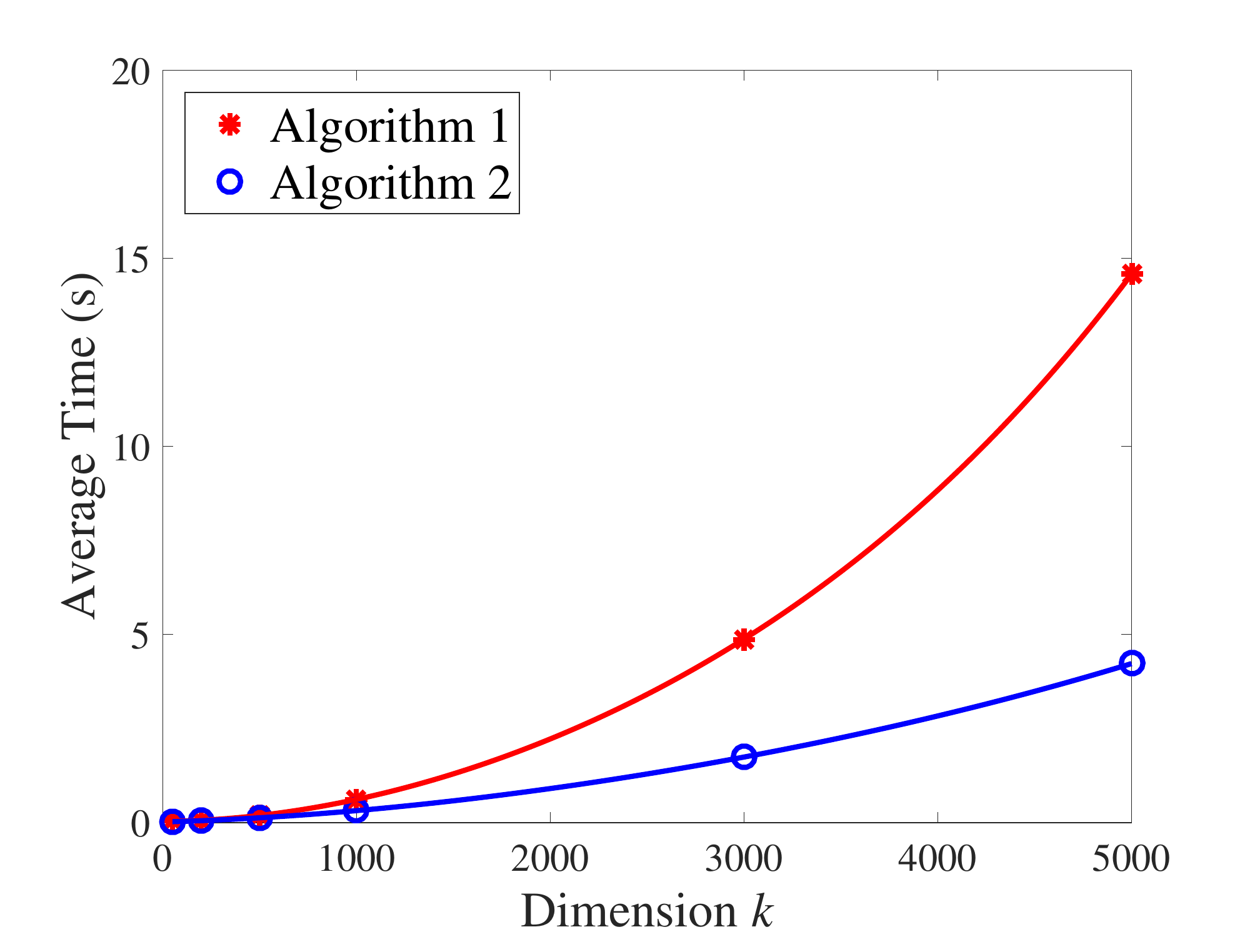}
	}
	\subfigure[]{\label{fig:Alg1DiagCovVaryK2}
		\includegraphics[width = 4.1 cm]{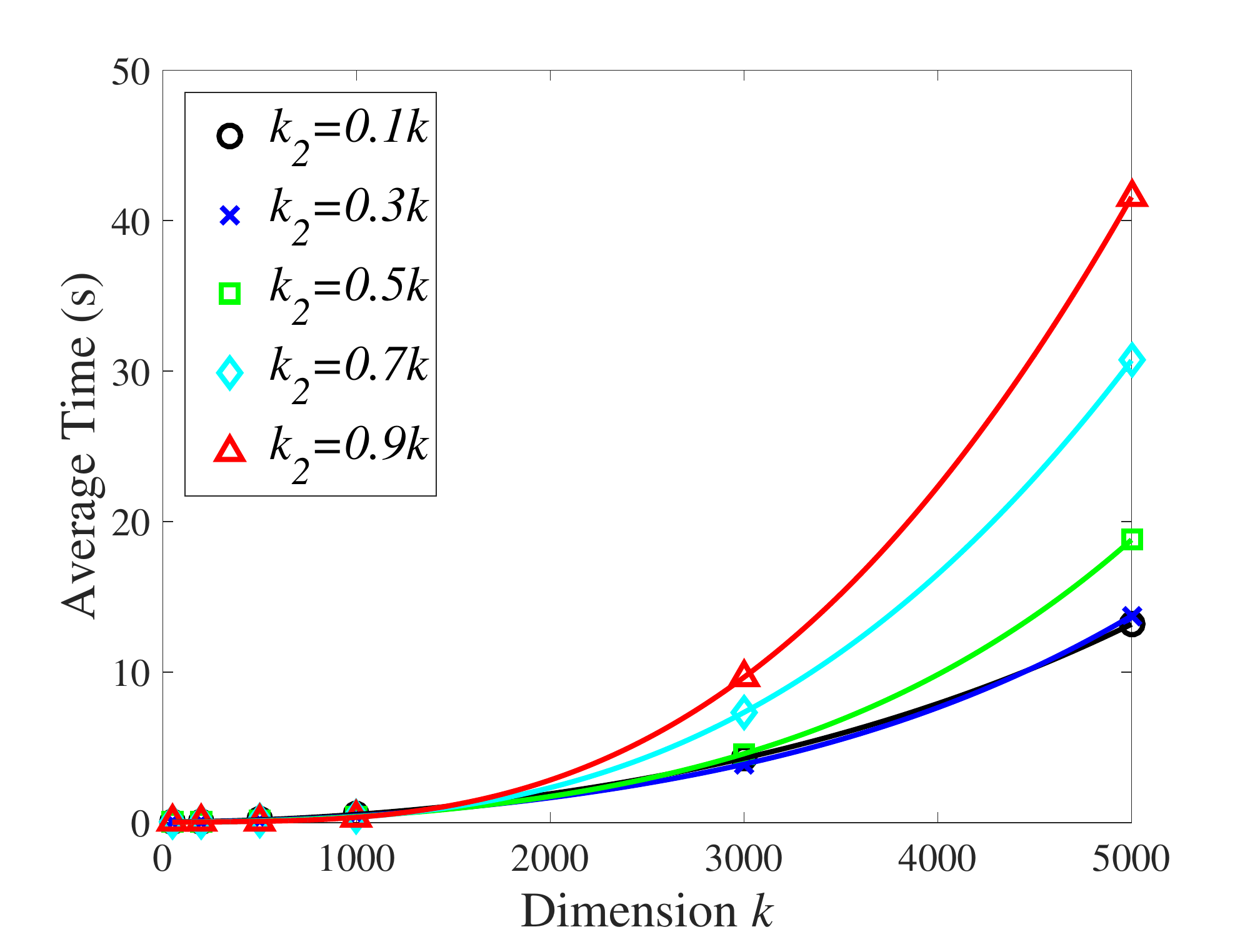}
	}
	\subfigure[]{\label{fig:Alg2DiagCovVaryK2}
		\includegraphics[width = 4.1 cm]{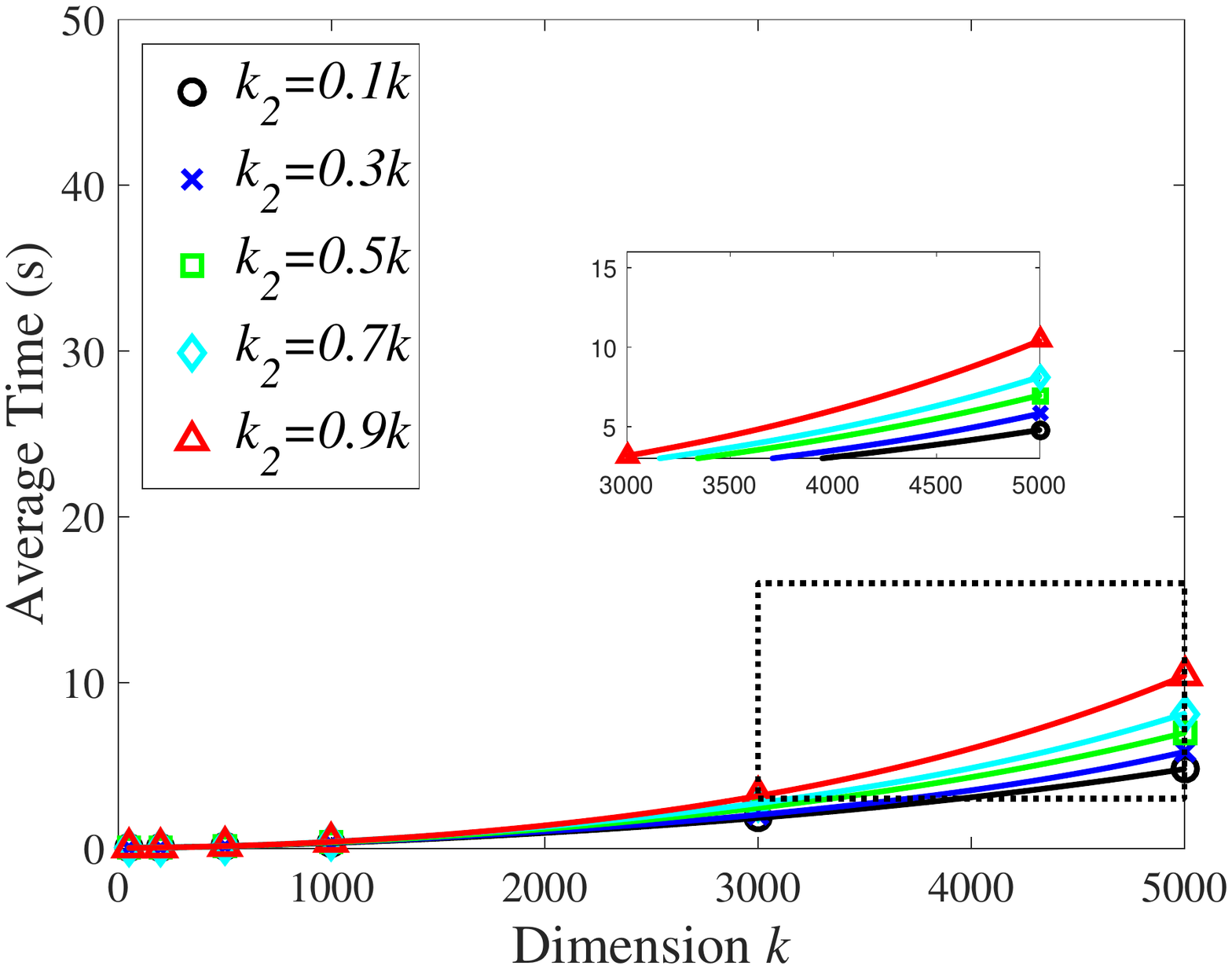}
	}
	\caption{\small Average time of simulating 10,000 hyperplane-truncated MVN samples over five random trials in different dimensions with non-diagonal covariance matrixes (top row) and diagonal ones (bottom row). (a)(d) Comparison with fixed $k_2 = 20$. (b)(e) Algorithm \ref{alg:1} with varying $k_2$. (c)(f) Algorithm~\ref{alg:2} with varying $k_2$.
	}
	\label{fig:A1vsA2Time}
\end{figure}

\subsubsection{A practical application of Algorithm \ref{alg:2}}

In what follows, we extend Algorithm \ref{alg:2} to facilitate simulation from a MVN distribution truncated on a probability simplex $\mathbb{S}^k = \{ \xv : \xv \in \Rbb^k, \bds 1^T \xv = 1, x_i \ge 0, i= 1,\cdots,k \} $.
This problem frequently arises when unknown parameters can be interpreted as fractions or probabilities, for instance, in topic models \citep{blei2003latent}, 
admixture models \citep{pritchard2000inference,dobigeon2009bayesian,bazot2013unsupervised}, and discrete directed graphical models \citep{heckerman1998tutorial}.
With Algorithm \ref{alg:2}, one may remove the equality constraint 
to greatly simplify the problem. 

More specifically, we focus on a big data setting in which the globally shared simplex-constrained model parameters  could be linked to some latent counts via the multinomial likelihood. 
When there are tens of thousands or millions of observations in the dataset, scalable  Bayesian inference  for the  simplex-constrained globally shared model parameters  
is highly desired, for example, for 
inferring the topics' distributions over words  in 
 latent Dirichlet allocation \citep{blei2003latent,OnlineLDA} and Poisson factor analysis \citep{zhou2012beta,GBN}. 

Let us denote the $\kappa$th model parameter vector constrained on a $V$-dimensional simplex by $\phiv_\kappa\in \mathbb{S}^V$, which could be linked to the latent counts $n_{vj\kappa}\in \Zbb$ of the $j$th document under a multinomial likelihood as $(n_{1j\kappa},\ldots, n_{Vj\kappa})\sim \mbox{Mult}(n_{\cdotv j\kappa}, \phiv_\kappa)$, where $\Zbb = \{ 0,1,2,\cdots\}$, $v\in\{1,\ldots,V\}$, $\kappa\in\{1,\ldots,K\}$, and $j\in\{1,\ldots,N\}$. In topic modeling, one may consider $K$ as the total number of latent topics and $n_{vj\kappa}$ as the number of words at the $v$th vocabulary term  in the $j$th document that are associated with  the $\kappa$th latent topic. Note that the dimension $V$ in real applications is often large, such as tens of thousands in topic modeling. 
 Given the observed counts $n_{vj}$ for the whole dataset, in a batch-learning setting, one typically iteratively updates the latent counts $n_{vj\kappa}$ conditioning on $\phiv_\kappa$, and updates $\phiv_\kappa$ conditioning on $n_{vj\kappa}$. 
 
 However, this batch-learning inference procedure would become inefficient and even impractical when the dataset size $N$ grows to a level that makes it too time consuming to finish even a single iteration of updating all local variables $n_{vj\kappa}$.
To address this issue, we consider constructing a mini-batch based Bayesian inference procedure that could make substantial progress in posterior simulation while the batch-learning one may still be waiting to finish a single iteration. 

Without loss of generality, in the following discussion, we drop the latent factor/topic index $\kappa$ to simplify the notation, focusing on the update of a single simplex-constrained global parameter vector. 
More specifically, 
we let  the latent local count vector $\nv_j =(n_{1j},\ldots,n_{Vj})^T$ be linked to the simplex-constrained global parameter vector $\phiv \in \mathbb{S}^V$   via the multinomial likelihood as $\nv_j \sim \Mult \left( n_{\cdotv j}, \phiv \right)$, and impose  a Dirichlet distribution prior on $\phiv$ as $\phiv \sim \Dir \left( \eta \mathbf{1}_V\right)$. 

Instead of waiting for all $\nv_j$ to be updated before performing a single update of $\phiv$, we develop a mini-batch based Bayesian inference algorithm under a general framework for constructing stochastic gradient Markov chain Monte Carlo (SG-MCMC) \citep{ma2015complete}, allowing 
$\phiv$ to be updated every time a mini-batch of $\nv_j$ are processed. 
For the sake of completeness, we concisely describe the derivation for a SG-MCMC algorithm, as outlined  below, for simplex-constrained globally shared model parameters. We refer the readers to 
\citet{TLASGR}
 for more details on the derivation and 
  its application to scalable inference for topic modeling. 

Using the reduced-mean parameterization of the simplex constrained vector $\phiv$, namely $\varphiv = (\phi_1,\cdots,\phi_{V-1})^T$, where $\varphiv \in \mathbb{R}_{+}^{V-1}$ is constrained with $\varphi_{\cdotv} \le 1$, 
we develop a SG-MCMC algorithm that updates $\varphiv$ for the $t$th mini-batch as
\beq \label{eq:upvarphi} 
	\varphiv_{t + 1} \!=\! \left[ 
	\varphiv_t \!+\! \frac{\varepsilon _t}{M}
	\!\left[ \left( \rho \bar \nv_{: \cdotv} \!+\! \eta \right) \!-\! \left(\rho n_{\cdotv \cdotv} \!+\! \eta V\right) \varphiv_t \right]
	\!+\! \Nc \left( \bds 0,\frac{2\varepsilon _t}{M} \!\left[ \diag \left( \varphiv_t \right) \!-\! \varphiv_t \varphiv_t ^T \right] \right)
	\right]_{\triangle} \!,
\eeq
where $\varepsilon _t$ are annealed step sizes, $\rho$ 
is the ratio of the dataset size $N$ to the mini-batch size, 
\zz $\nv_{: \cdotv}=(\nv_{1 \cdotv},\cdots,\nv_{V \cdotv})^T = \sum_{j\in I_t} \nv_{j}$, $\bar \nv_{: \cdotv} = (\nv_{1 \cdotv},\cdots,\nv_{(V-1) \cdotv})^T$, $[\cdotv]_{\triangle}$ denotes the constraint that $\varphiv \in \mathbb{R}_{+}^{V-1}$ and $\varphi_{\cdotv} \le 1$, and $M := \Ebb \left[\sum_{j=1}^N  n_{\cdotv j} \right]$ is approximated along the updating using $ M = \left( 1 - \varepsilon_t \right) M + {\varepsilon_t} \rho \Ebb \left[ n_{ \cdotv \cdotv } \right] $. 
Alternatively, we have an equivalent update equation for $\phiv$ as
\beq \label{eq:upphi} 
\phiv_{t + 1} \!=\! \left[ 
\phiv_t \!+\! \frac{\varepsilon _t}{M}
\!\left[ \left( \rho \nv_{: \cdotv} \!+\! \eta \right) \!-\! \left(\rho n_{\cdotv \cdotv} \!+\! \eta V\right) \phiv_t \right]
\!+\! \Nc \left( \bds 0,\frac{2\varepsilon _t}{M} \diag \left( \phiv_t \right)  \right)
\right]_{\angle} \!,
\eeq
where $[\cdotv]_{\angle}$ represents the constraint that $\phiv \in \Rbb_{+}^{V}$ and $\bds 1 ^T \phiv = 1$.

It is clear that \eqref{eq:upvarphi} corresponds to simulation of a $V-1$ dimensional truncated MVN distribution with $V$ inequality constraints. 
Since the number of constraints is larger than the dimension, previously proposed iterative simulation methods such as the one in \citet{botev2016normal} are often inappropriate. 
Note that, 
by omitting the non-negative constraints, 
the update in \eqref{eq:upphi} corresponds to simulation of  a hyperplane-truncated MVN simulation with a diagonal covariance matrix, which can be efficiently sampled as described in the following example.

\noindent\textbf{Example 1:} \textit{
	Simulation of a hyperplane-truncated MVN distribution as
	$$
	\xv\sim\Nor_{\mathcal{S}}[\muv,a\,\diag(\phiv)],~~\mathcal{S}=\left\{\xv: \mathbf{1}^T \xv = 1\right\},
	$$
	where $\xv\in\mathbb{R}^{k}$, $\muv\in\mathbb{R}^{k}$, $\mathbf{1}^T\xv =\sum_{i=1}^k x_i$, $\phiv\in\mathbb{R}^{k}$, $a>0$, $\phi_i > 0$ for $i \in\{ 1,\cdots,k\}$, and $\mathbf{1}^T \phiv = \sum_{i=1}^k \phi_i=1$,
	can be realized as follows.
	\begin{itemize}
		\item Sample $\yv\sim\Nor [ \muv , a \diag(\phiv) ]$;
		\item Return $\xv = \yv + ( 1 - \mathbf{1}^T \yv) \phiv$.
	\end{itemize}
}

\noindent The sampling steps in Example 1 directly follow Algorithm \ref{alg:2} and Theorem~\ref{maintheorem1} with the distribution parameters specified as
$\Sigmamat = a \diag (\phiv)$, $\Gmat = \bds 1^T$, and $\rv = 1$.
Accordingly, we present the following fast sampling procedure for \eqref{eq:upvarphi}.


\noindent\textbf{Example 2:} \textit{
	Simulation from \eqref{eq:upvarphi} can be approximately but rapidly realized as
	\begin{itemize}
		\item Sample $\yv \sim \Nor \big[ \phiv_t + \frac{\varepsilon _t}{M} \left[ \left( \rho \nv_{: \cdotv} + \eta \right) - \left(\rho n_{\cdotv \cdotv} + \eta V\right) \phiv_t \right] ,\frac{2\varepsilon _t}{M} \diag \left( \phiv_t \right) \big]$;
		\item Calculate $\zv = \yv + ( 1 - \mathbf{1}^T \yv) \phiv_t$;
		\item If $\zv \in \mathbb{S}$, return $\varphiv_{t+1} = (z_1,\cdots,z_{V-1})^T$; else
		calculate $\dv = \max(\epsilon,\zv)$ with a small constant $\epsilon \ge 0$, let $\ev = \dv / \sum\nolimits_{i=1}^V d_i$, and return $\varphiv_{t+1} = (e_1,\cdots,e_{V-1})^T$.
	\end{itemize}
}

To verify Example 2, we conduct an experiment using multinomial-distributed data vectors of $V=2000$ dimensions, which are generated as follows: 
considering that the simplex-constrained vector $\phiv$ is usually sparse in a high-dimensional application, we  sample a $V=2000$ dimensional vector $\fv$ whose elements are uniformly distributed between 0 and 1, randomly select 40 dimensions and reset their values to be 100, and set $\phiv = \fv / \sum\nolimits_{i=1}^V f_i$; we simulate $N=10,000$ samples, each $\nv_j$ of which is generated from the multinomial distribution $\mbox{Mult}(n_{\cdotv j}, \phiv)$, where the number of trials is random and generated as $n_{\cdotv j} \sim \Pois (50)$. 
We set $\varepsilon_t = t^{-0.99}$ and use mini-batches, each of which consists of 10 data samples, to stochastically update global parameters  
 via SG-MCMC. 

For comparison, we choose the same SG-MCMC inference procedure but consider  simulating \eqref{eq:upvarphi}, as performed every time a mini-batch of data samples are provided,  either as in Example 2 or with the Gibbs sampler of \citet{rodriguez2004efficient}.
Simulating \eqref{eq:upvarphi} with the Gibbs sampler of \citet{rodriguez2004efficient} is realized by updating all the $V$ dimensions, one dimension at a time, in each Gibbs sampling iteration. We set the total number of Gibbs sampling iterations for  \eqref{eq:upvarphi}  in each mini-batch based update as 1, 5, or 10.  Note that in practice, the $\nv_j$ belonging to the current mini-batch are often latent and are updated conditioning on the data samples in the mini-batch and $\phiv$. For simplicity,  all $\nv_j$ here are simulated once and then fixed.  



\begin{figure}[!t]
	\centering
	\subfigure[]{\label{fig:ResErrIterationComp}
		\includegraphics[height = 4.6 cm]{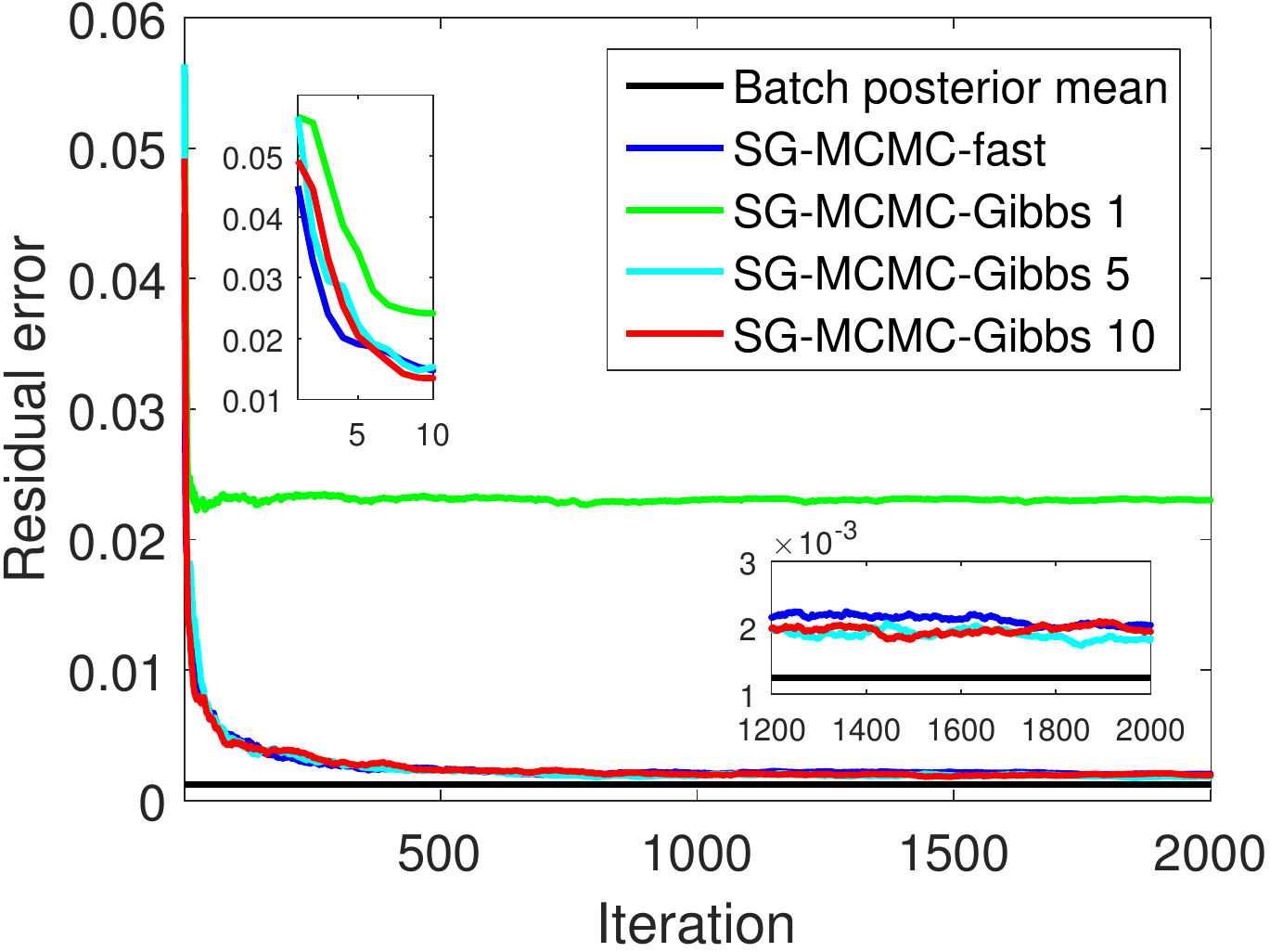}
	}
	\subfigure[]{\label{fig:ResErrTimeComp}
		\includegraphics[height = 4.6 cm]{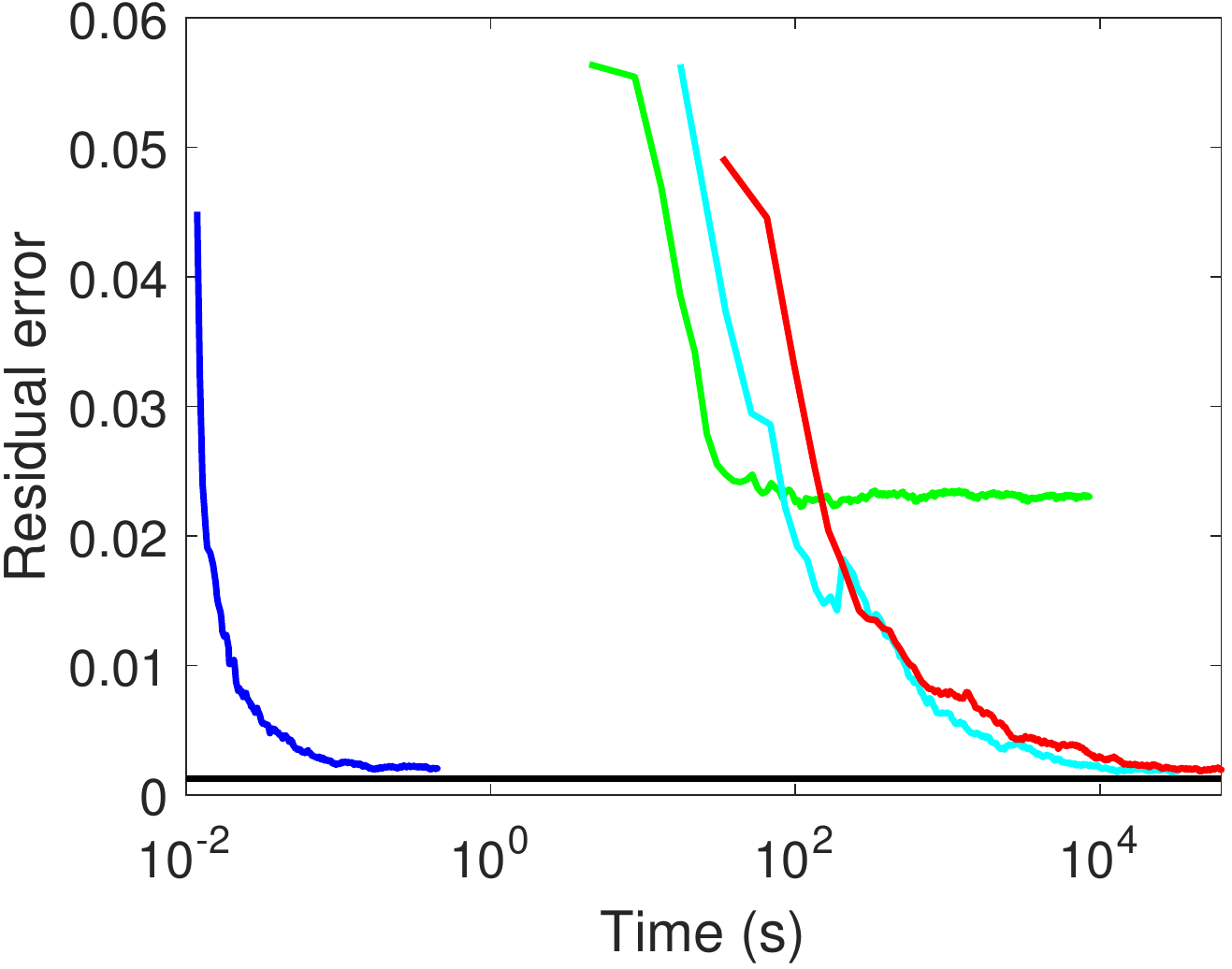}
	}
	\caption{\small Comparisons of the residual errors of the simplex-constrained parameter vector, estimated under various settings of the stochastic-gradient MCMC (SG-MCMC) algorithm, as a function of (a) the number of processed mini-batches and (b) time.
		The curves labeled as ``Batch posterior mean'', ``SG-MCMC-fast'', and ``SG-MCMC-Gibbs'' correspond to the batch posterior mean, SG-MCMC with \eqref{eq:upvarphi} simulated as in Example 2, and SG-MCMC with \eqref{eq:upvarphi} simulated with the Gibbs sampler of  \citet{rodriguez2004efficient},  respectively. 
		 The digit following ``SG-MCMC-Gibbs'' represents the number of Gibbs sampling iterations to simulate  \eqref{eq:upvarphi} for each mini-batch. 
	}
	\label{fig:ResErrComp}
\end{figure}

Using $\phiv_{post}^* = ( \sum_{j=1}^N \nv_{j} + \eta ) / ( \sum_{j=1}^N n_{\cdotv j} + \eta V ) $, the posterior mean of $\phiv$ in a batch-learning setting, as the reference, we show in Figure \ref{fig:ResErrComp} how the residual errors for the estimated $\phiv^*$, defined as $\left\| \phiv^* - \phiv \right\|_2 $, change both as a function of the number of processed mini-batches and as a function of computation time under various settings of the mini-batch based SG-MCMC algorithm. 
The curves shown in Figure \ref{fig:ResErrComp} suggest that for each mini-batch, to simulate \eqref{eq:upvarphi} with the Gibbs sampler of \citet{rodriguez2004efficient},  it is necessary to have more than one Gibbs sampling iteration to achieve satisfactory results. 
It is clear from Figure \ref{fig:ResErrIterationComp} that the Gibbs sampler with 5 or 10 iterations for each mini-batch, even though each mini-batch has only 10 data samples, provides residual errors that quickly approach that of the batch posterior mean with a tiny gap, indicating the effectiveness of the SG-MCMC updating in \eqref{eq:upvarphi}.  
While simulating \eqref{eq:upvarphi} with Gibbs sampling could in theory lead to unbiased samples if the number of Gibbs sampling iterations is large enough, it is much more efficient to simulate \eqref{eq:upvarphi} with  the procedure described in Example 2, which provides a performance that is undistinguishable from those of the Gibbs sampler with as many as 5 or 10 iterations for each mini-batch, but at the expense of a tiny fraction of a single Gibbs sampling iteration. 

\subsection{Simulation of MVNs with structured covariance matrices}

To illustrate Corollary \ref{cor2}, we mimic the truncated MVN simulation in \eqref{eq:upvarphi} and present the following simulation example with a structured covariance matrix. 

\noindent\textbf{Example 3:} 
Simulation of a MVN distribution as
$$
\xv_1\sim\Nor [\muv_1,a\,\diag(\phiv_1)-a\,\phiv_1\phiv_1^T],
$$
where $\xv_1\in\mathbb{R}^{k-1}$, $\muv_1\in\mathbb{R}^{k-1}$, $a>0$, $\phiv_1=(\phi_1,\ldots,\phi_{k-1})^T$, $\phi_{i}>0$ for $i\in\{1,\ldots,k-1\}$, and $\sum_{i=1}^{k-1}\phi_i <1$, can be realized as follows.
\begin{itemize}
	\item Sample $\yv_1\sim\Nor [ \mathbf{0} , a\,\diag(\phiv_1) ]$ and $\yv_2\sim\Nor ( 0 , a^{-1}\,\phi_k )$, where $\phi_k=1- \sum_{i=1}^{k-1} \phi_i$; 
	\item Return $\xv_1 = \muv_1+\yv_1 - (\mathbf{1}^T \yv_1 + a\yv_2) \phiv_1$. 
\end{itemize}
Denoting $\xv=(\xv_1^T,x_{k})^T$, $\phiv=(\phiv_1^T,\phi_{k})^T$, $\muv=(\muv_1^T,\mu_{k})^T$, and $\mu_k=1-\mathbf{1}^T \muv_1$,
the above sampling steps can also be equivalently expressed as follows.
\begin{itemize}
	\item Sample $\yv\sim\Nor [ \muv , a\,\diag(\phiv) ]$; 
	\item Return $\xv_1 = \yv_1 + (1- \mathbf{1}^T \yv ) \phiv_1$. 
\end{itemize}

Directly following Algorithm \ref{alg:3} and Corollary \ref{cor2}, the first sampling approach for the above example can be derived by specifying
the distribution parameters as
$\Sigmamat_{11} = a \diag (\phiv_1)$, $\Sigmamat_{12} = \phiv_1$, $\Sigmamat_{21} = \phiv_1^T$, and $\Sigmamat_{22} = a^{-1}$, while the second approach can be derived by specifying 
$\Sigmamat_{11} = a \diag (\phiv_1)$, $\Sigmamat_{12} = a\phiv_1$, $\Sigmamat_{21} = a\phiv_1^T$, and $\Sigmamat_{22} = a$.

To illustrate the efficiency of the proposed algorithms in Example 3, 
we simulate from the MVN distribution $\xv_1\sim\Nor [\muv_1,a\,\diag(\phiv_1)-a\,\phiv_1\phiv_1^T]$
using both a naive implementation via Cholesky decomposition of the covariance matrix and the fast simulation algorithm for a hyperplane-truncated MVN random variable described in Example 3. 
We set the dimension from $k = 10^2$ 
up to $k=10^4$ and set $\muv = (1/ k,\ldots,1/k)$ and $a = 0.5$. For each $k$ and each simulation algorithm, we perform 100 independent random trials, in each of which $\phiv$ is sampled from the Dirichlet distribution $\Dir ( 1,\ldots,1)$ and 10,000 independent random samples are simulated using that same $\phiv$. 

\begin{figure}[!t]
	\centering
	\includegraphics[width = 6.3 cm]{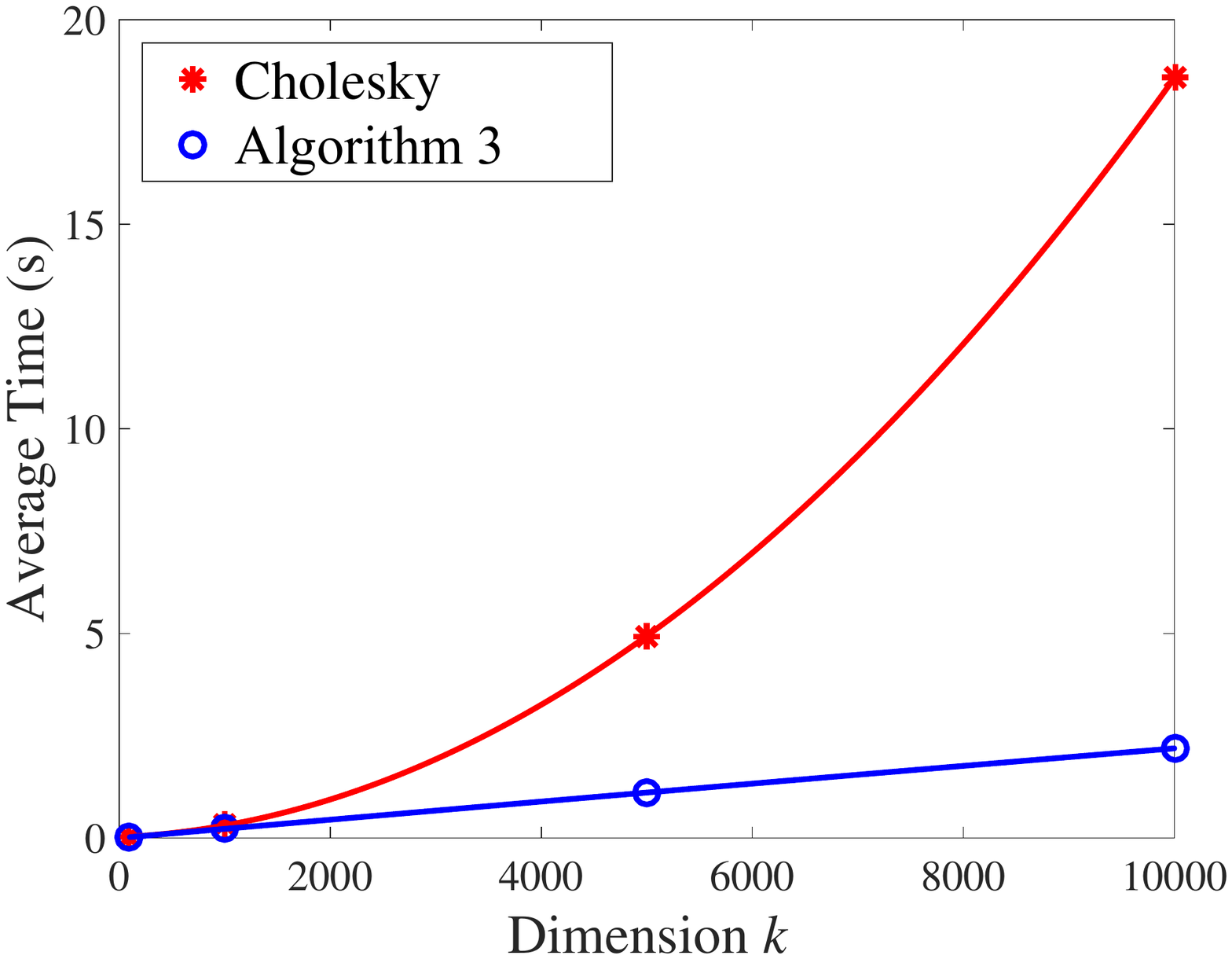}
	\caption{\small Comparison of the naive Cholesky decomposition based implementation and Algorithm \ref{alg:3} in terms of the average time of generating 10,000 $k$-dimensional random samples from $
		\xv_1\sim\Nor [\muv_1,a\,\diag(\phiv_1)-a\,\phiv_1\phiv_1^T]
		$. The distribution parameters are randomly generated and computation time averaged over 100 random trials is displayed.		
	}
	\label{fig:TimeVsDim}
\end{figure}

As shown in Figure \ref{fig:TimeVsDim}, for the proposed 
Algorithm \ref{alg:3}, the average time of simulating 10,000 random samples increases linearly in the dimension $k$. By contrast, for the naive Cholesky decomposition based simulation algorithm, whose computational complexity is $O(k^3)$ \citep{golub2012matrix}, the average simulation time increases at a significantly faster rate as the dimension $k$ increases. 

\begin{figure}[t]
	\centering
	\subfigure[]{
		\includegraphics[width = 2.2cm]{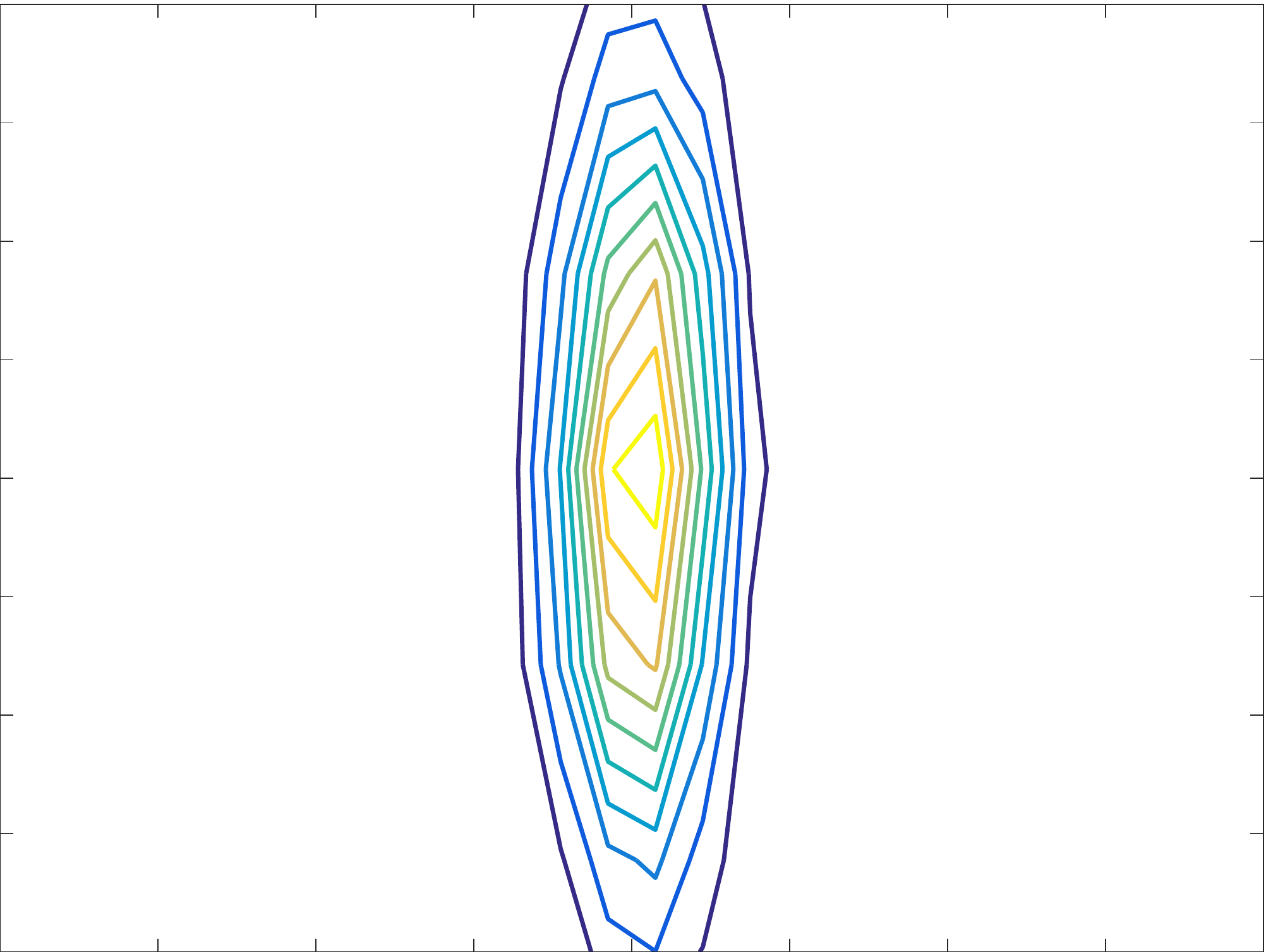}
	}
	\subfigure[]{
		\includegraphics[width = 2.2 cm]{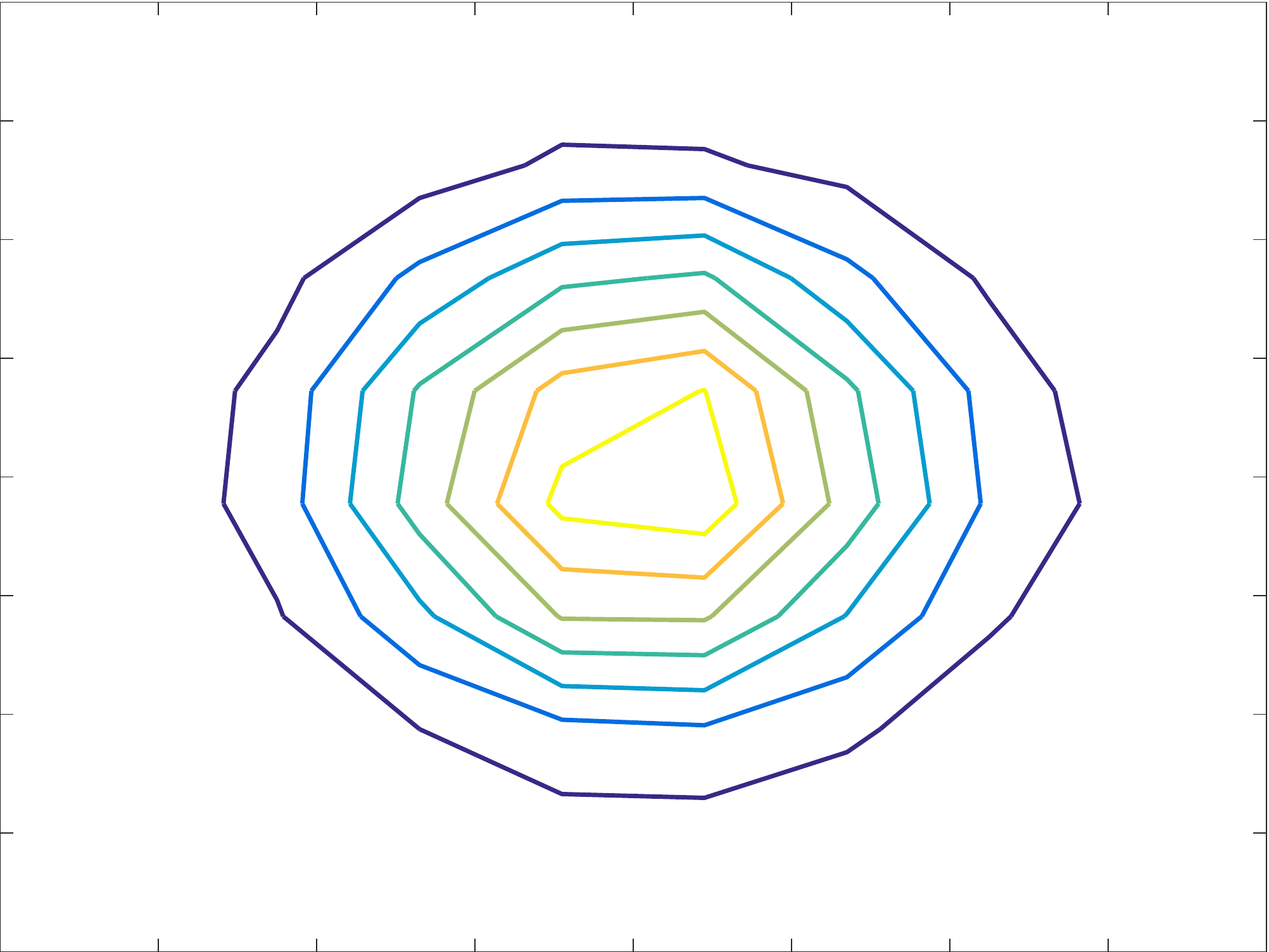}
	}
	\subfigure[]{
		\includegraphics[width = 2.2 cm]{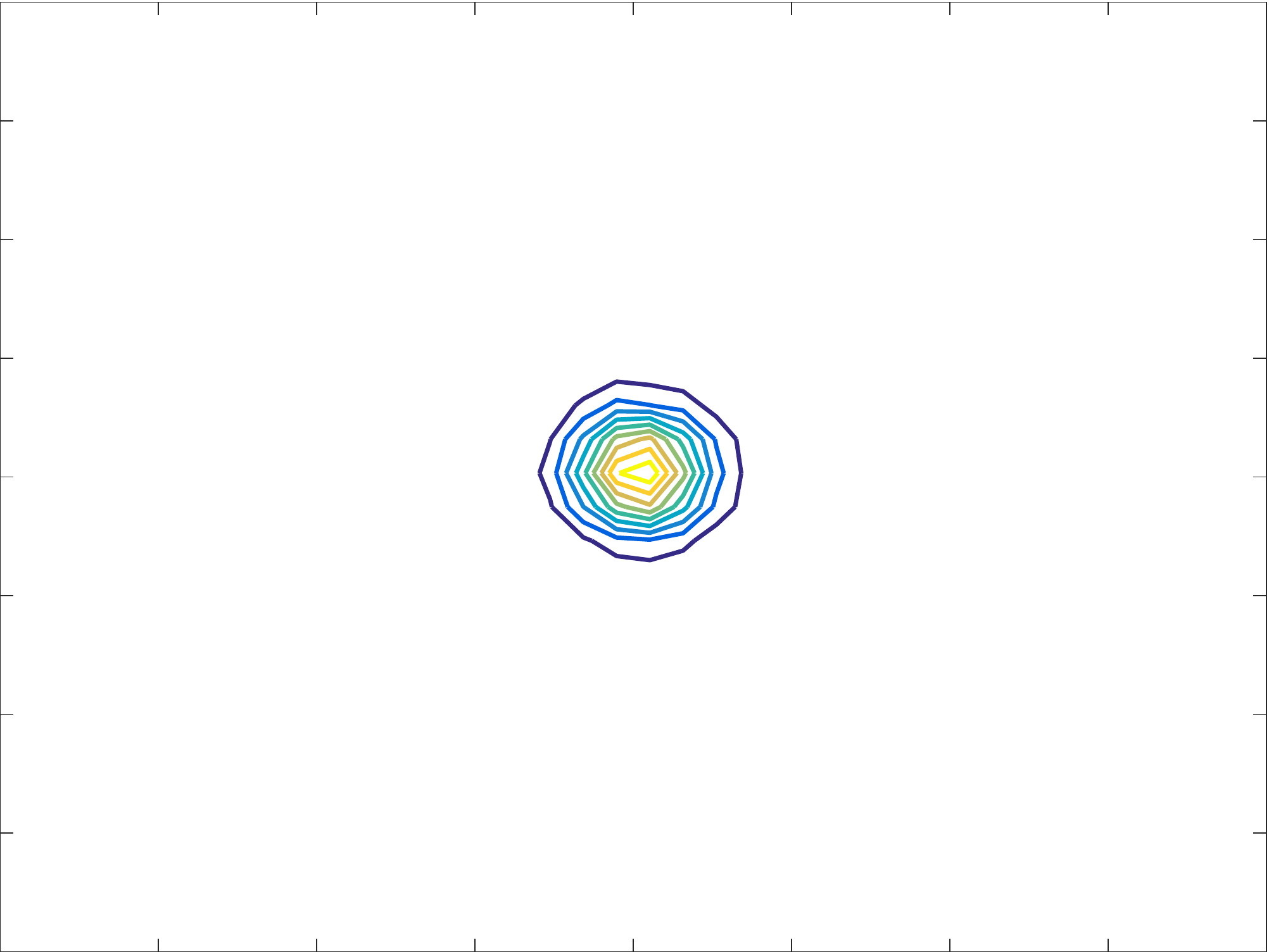}
	}
	\subfigure[]{
		\includegraphics[width = 2.2 cm]{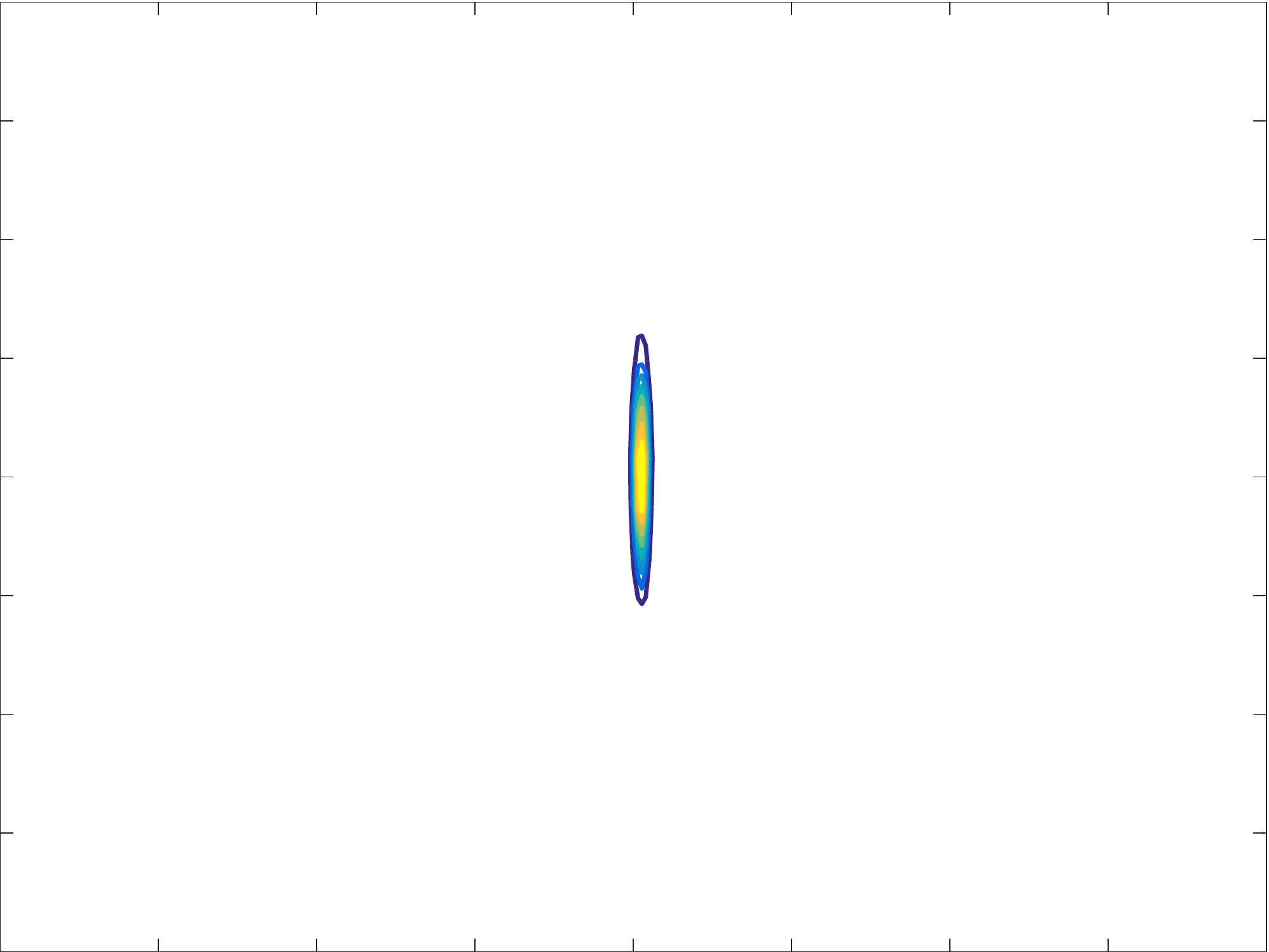}
	}
	\subfigure[]{
		\includegraphics[width = 2.2 cm]{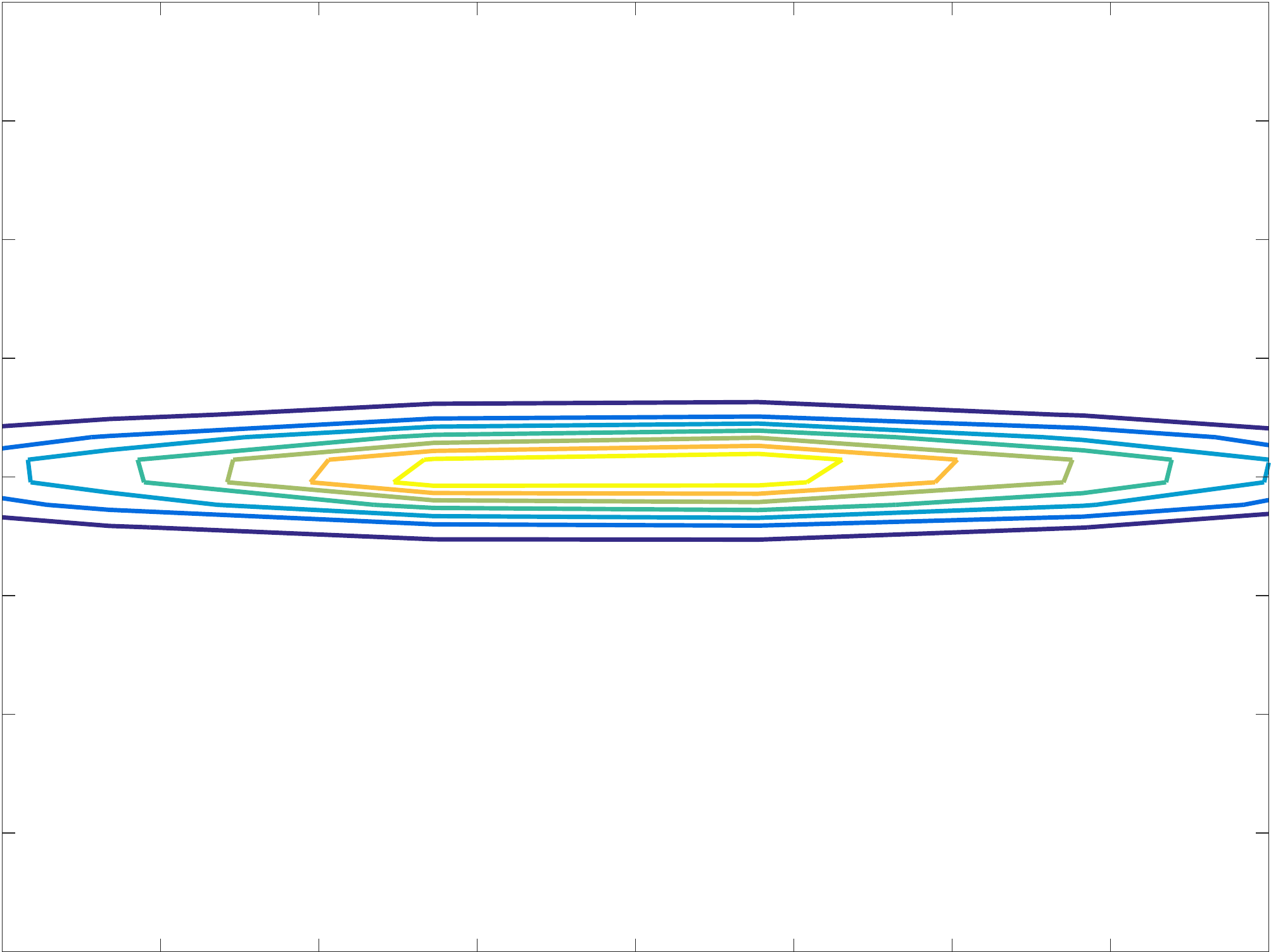}
	}
	\subfigure[]{
		\includegraphics[width = 2.2 cm]{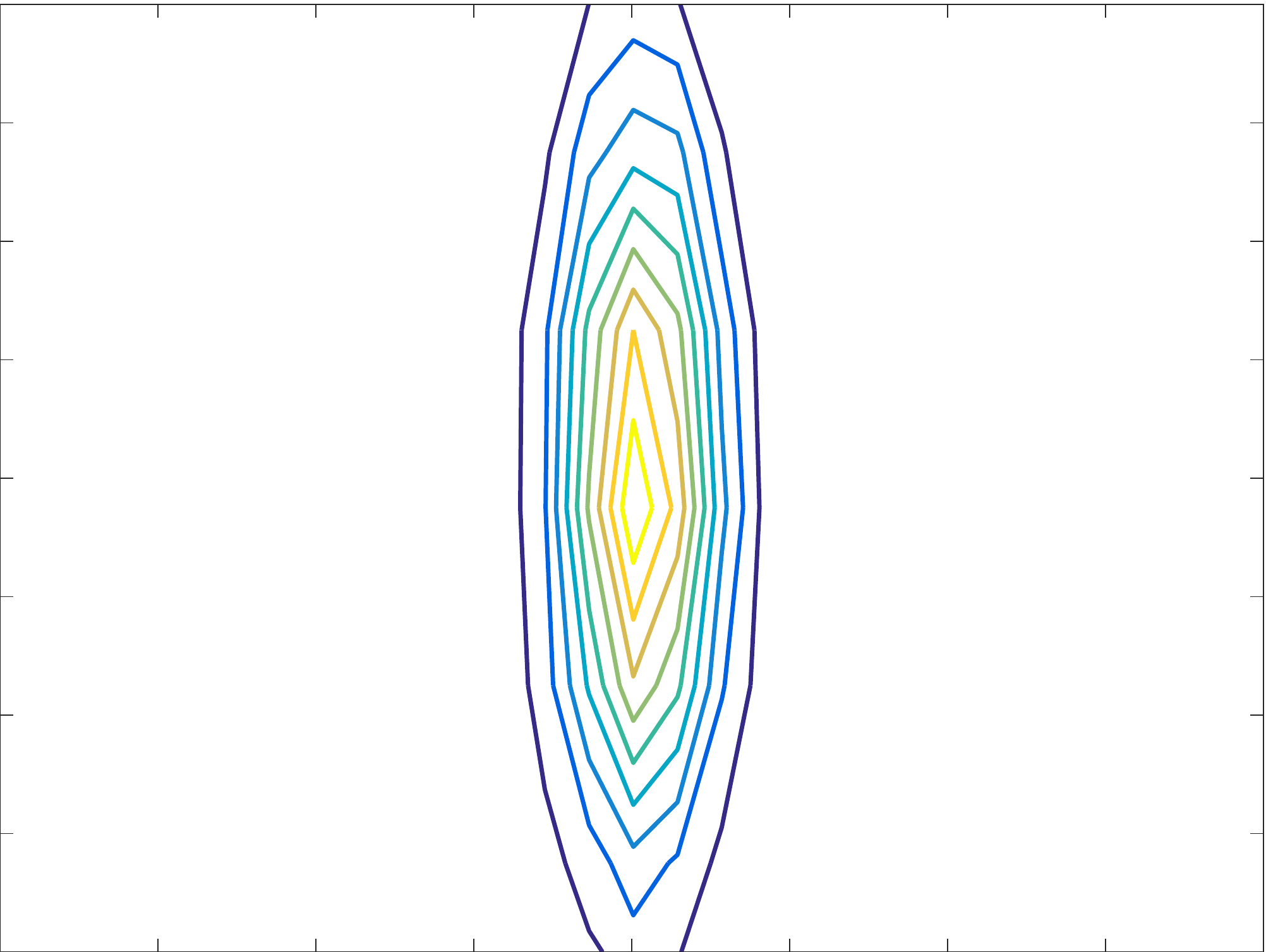}
	}
	\subfigure[]{
		\includegraphics[width = 2.2 cm]{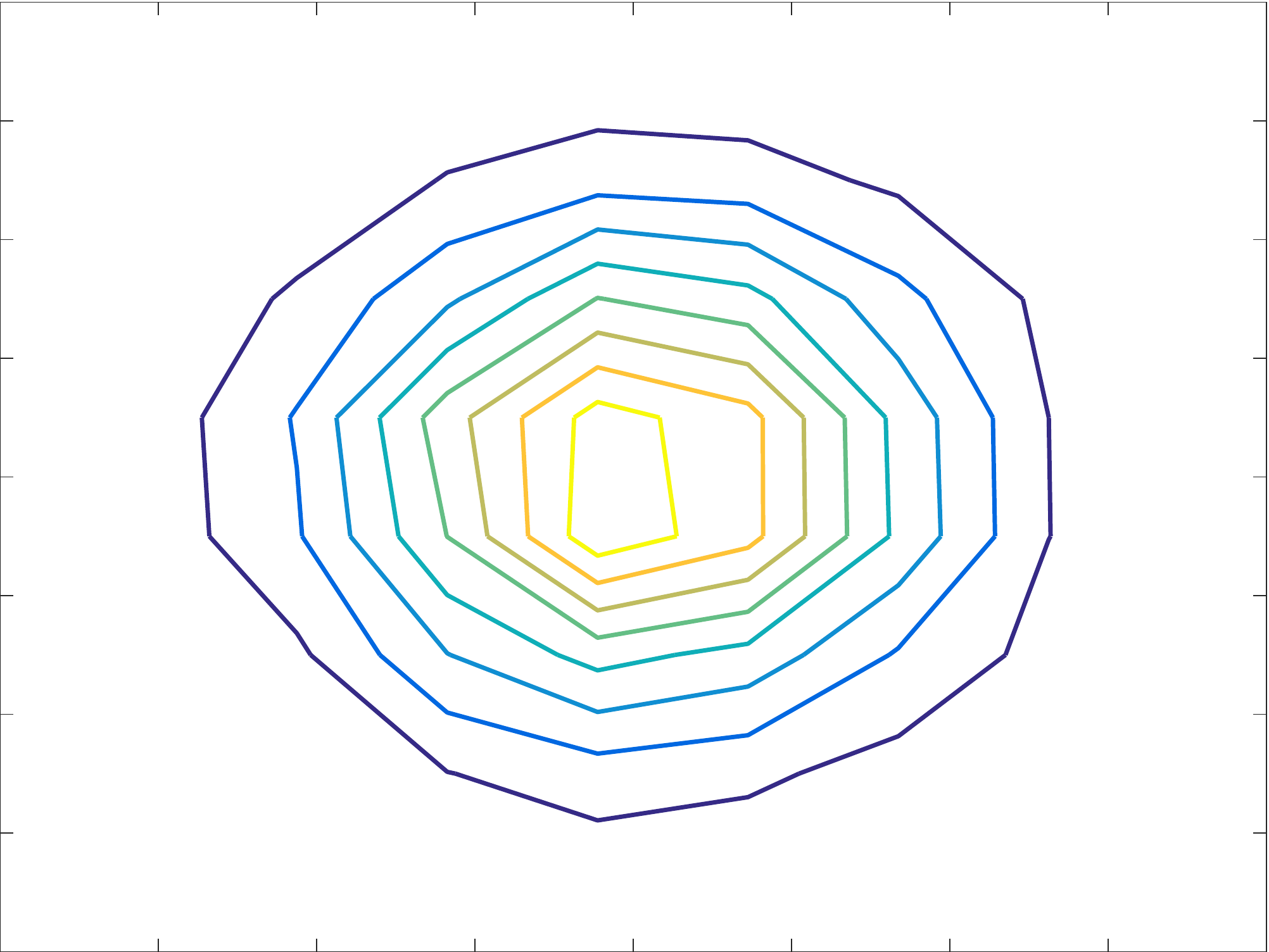}
	}
	\subfigure[]{
		\includegraphics[width = 2.2 cm]{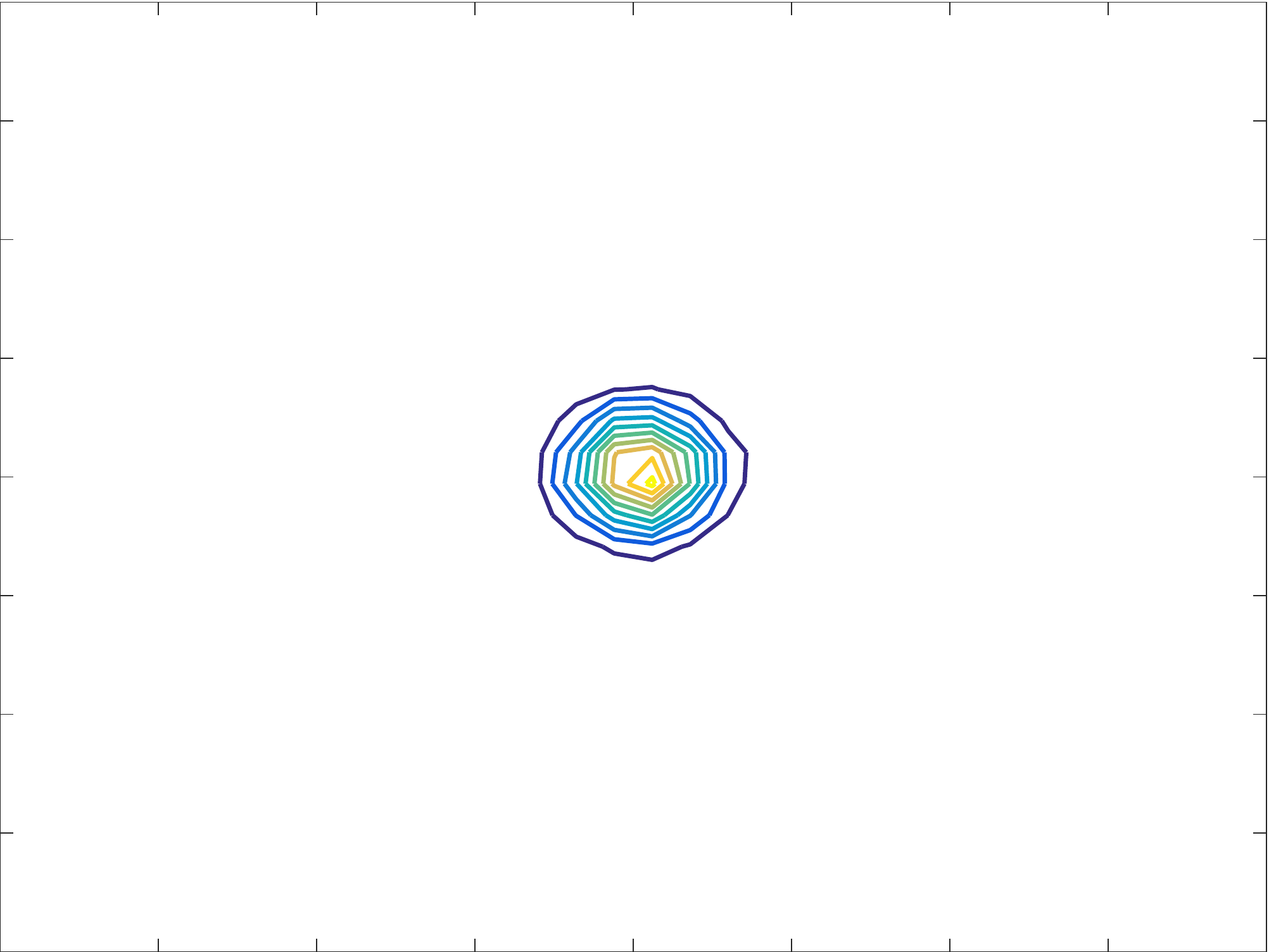}
	}
	\subfigure[]{
		\includegraphics[width = 2.2 cm]{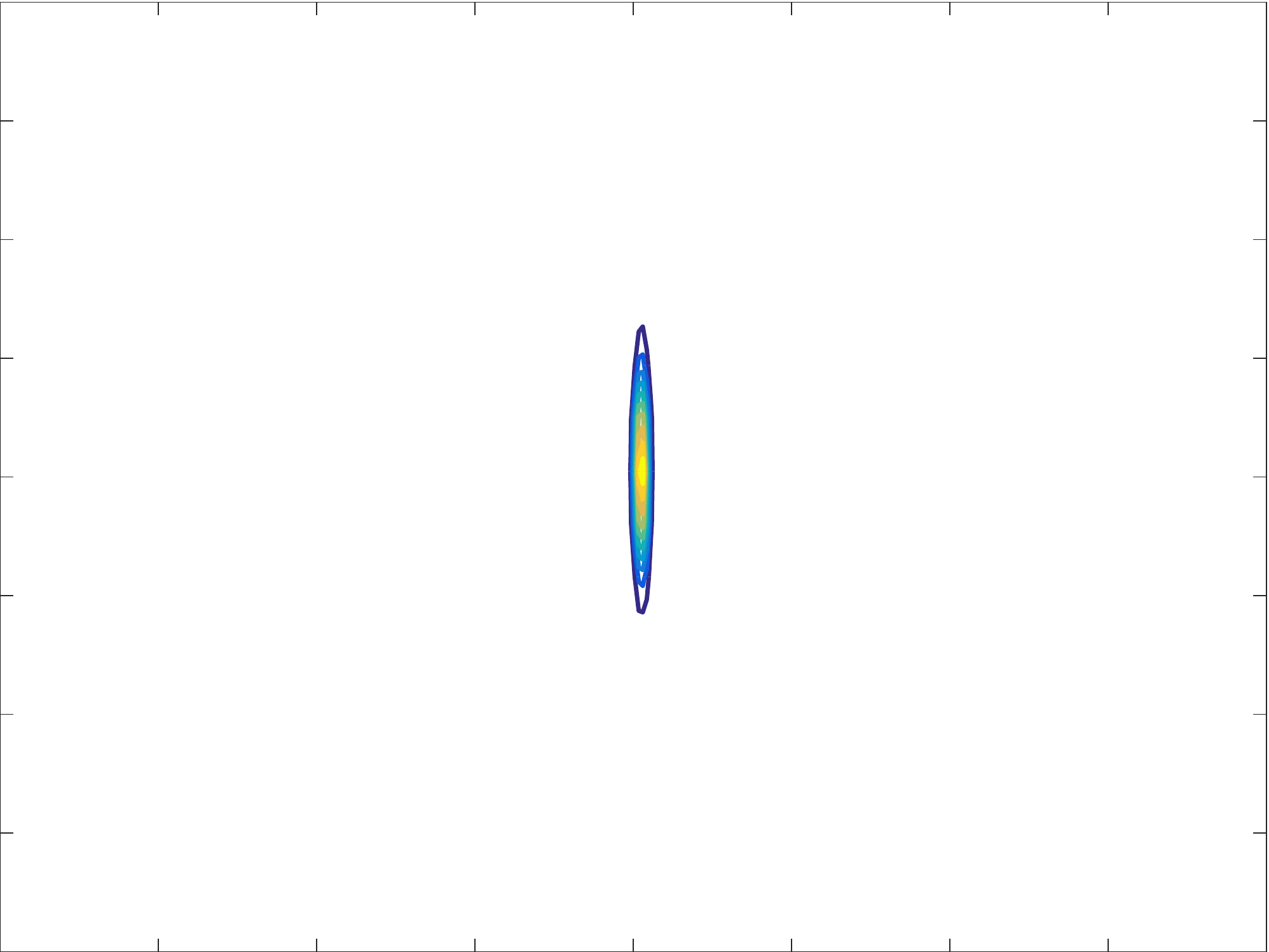}
	}
	\subfigure[]{
		\includegraphics[width = 2.2 cm]{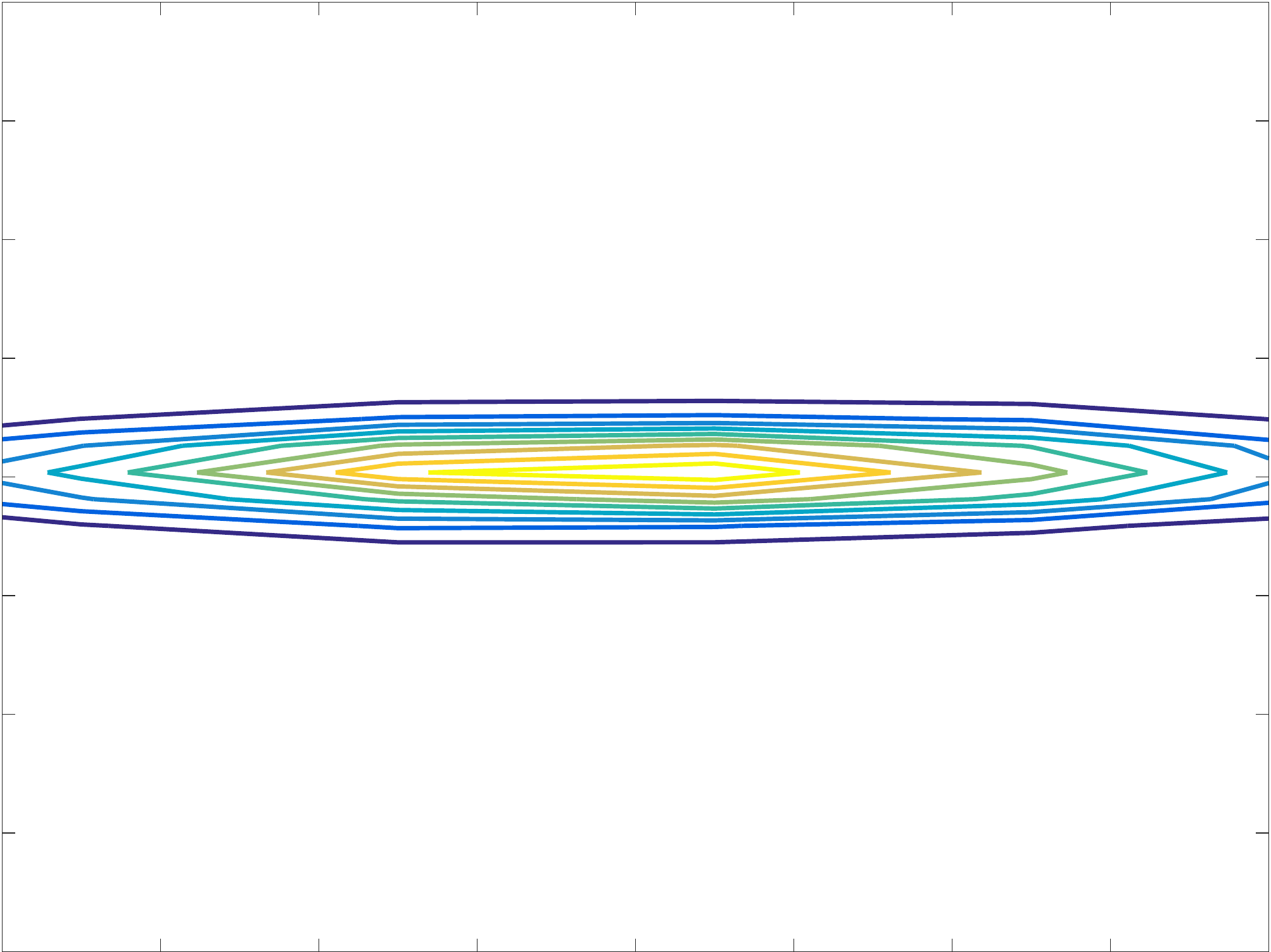}
	}
	\caption{\small Comparison of the contour plots of two randomly selected dimensions of the 10,000 $k=10^4$ dimensional random samples simulated with the naive Cholesky implementation (top row) and Algorithm \ref{alg:3} (bottom row). Each of the five columns corresponds to a random trial. 
	}
	\label{fig:ContourCompare}
\end{figure}

For explicit verification, with the 10,000 simulated $k=10^4$ dimensional random samples in a random trial, we randomly choose two dimensions and display their joint distribution using a contour plot. 
As in Figure \ref{fig:ContourCompare}, shown in the first row are the contour plots of five different random trials for the naive Cholesky implementation, whereas shown in the second row are the corresponding ones for the proposed Algorithm \ref{alg:3}. As expected, the contour lines of the two figures in the same column closely match each other. 

\begin{figure}[!th]
	\centering
	\includegraphics[width = 8 cm]{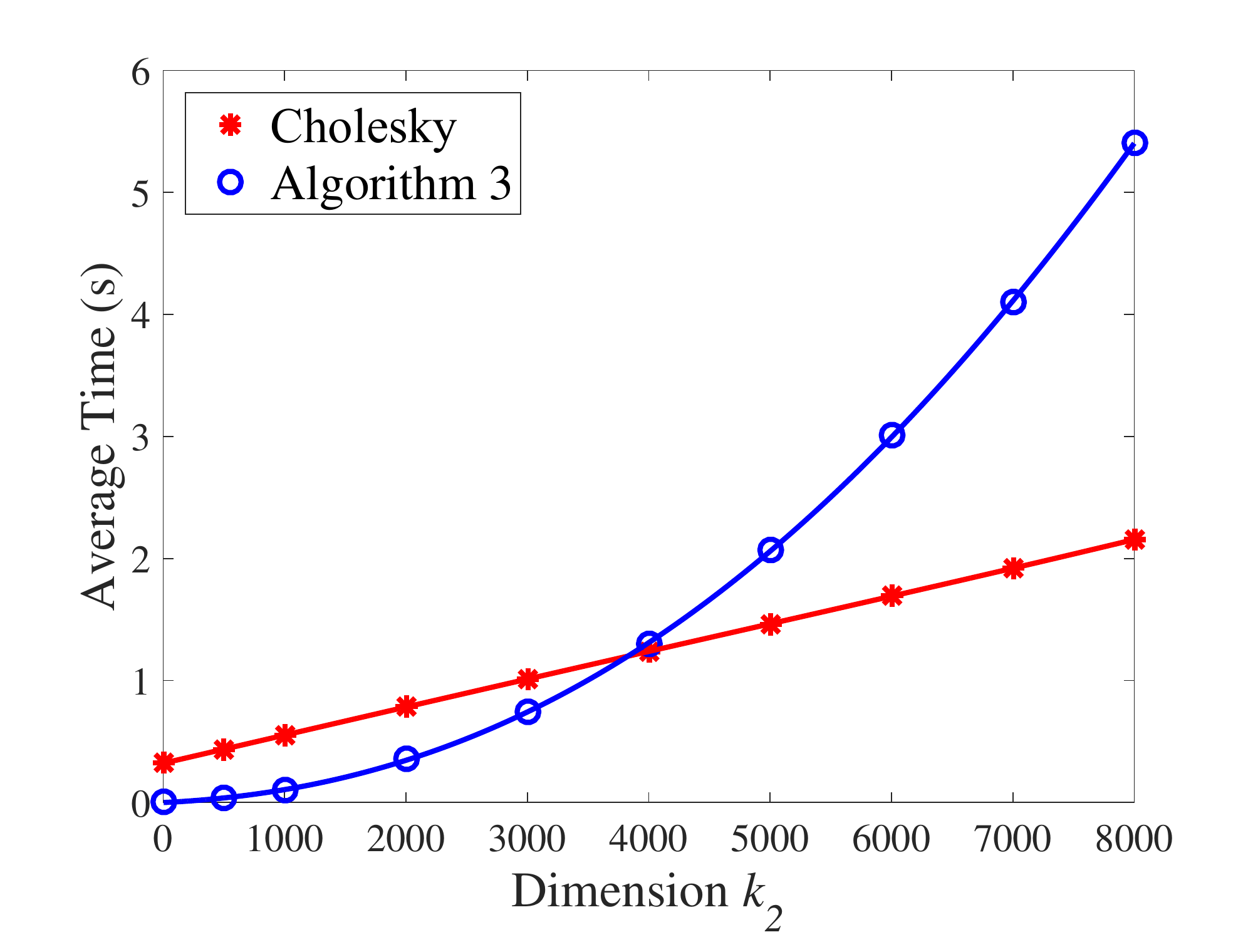}
	\caption{\small Comparison of the naive Cholesky decomposition based implementation and Algorithm \ref{alg:3} in terms of the average time of generating one $k_1=4000$ dimensional sample from $\xv_1\sim\Nor(\muv_1,\Sigmamat_{11}-\Sigmamat_{12}\Sigmamat_{22}^{-1}\Sigmamat_{21})$, with diagonal $\Sigmamat_{11}$ and $\Sigmamat_{22}$. 
		The distribution parameters are randomly generated and computation time averaged over 50 random trials is displayed.		
	}
	\label{fig:A3}
\end{figure}

To further examine  when to apply Algorithm \ref{alg:3} instead of the naive Cholesky decomposition based implementation in a general setting, 
we present the computational complexity analyses in Tables \ref{tab:CompNaive3} and \ref{tab:CompAlg3} of the Appendix for the naive approach and Algorithm \ref{alg:3}, respectively. In addition, 
we mimic the settings in Section \ref{sec:A1vsA2} to conduct a set of experiments with randomly generated $\Sigmamat_{12}$, diagonal $\Sigmamat_{11}$, and diagonal $\Sigmamat_{22}$. We fix $k_1 = 4000$ and vary $k_2$ from 1 to 8000.
The computation time for one sample averaged over 50 random trials is presented in Figure \ref{fig:A3}.
It is clear from 
Tables \ref{tab:CompNaive3} and \ref{tab:CompAlg3} and  Figure \ref{fig:A3} that, as a general guideline, one may choose Algorithm \ref{alg:3} when $k_2$ is  smaller than  $ k_1$ and $\Sigmamat_{11}$ admits some special structure that makes it easy to invert and computationally efficient to simulate from $\Nor(\bf{0},\Sigmamat_{11})$. 

\subsection{Simulation of MVNs with structured precision matrices}

To examine when to apply Algorithm \ref{alg:4} instead of the naive Choleskey decomposition based procedure, we first consider a series of random simulations in which the sample size $n$ is fixed while the data dimension $p$ is varying. We then show that Algorithm \ref{alg:4} can be applied for high-dimensional regression whose $p$ is often much larger than $n$. 

\begin{figure}[!th]
	\centering
	\includegraphics[width = 8 cm]{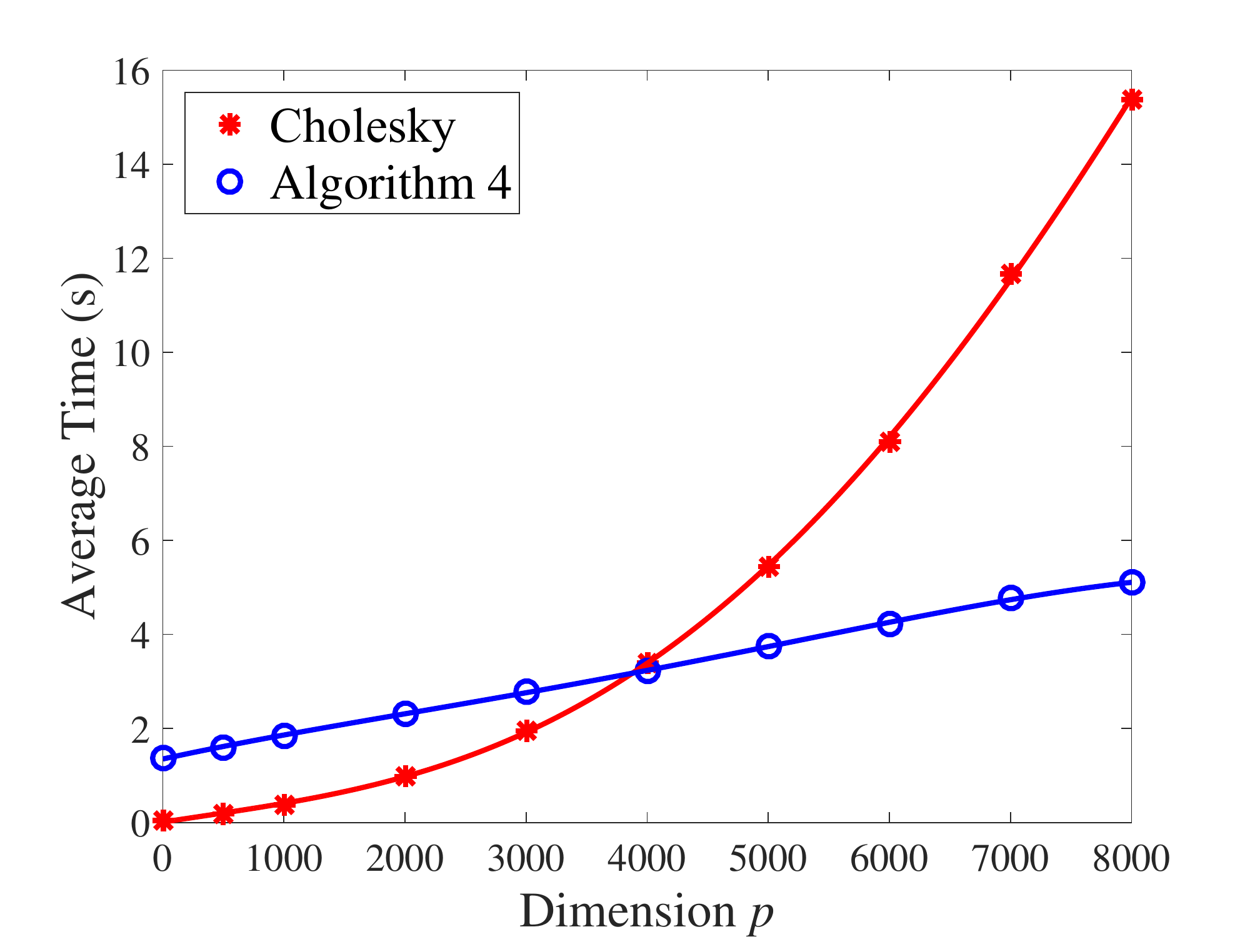}
	\caption{\small Comparison of the naive Cholesky decomposition based implementation and Algorithm \ref{alg:4} in terms of the average time of generating one $p$ dimensional sample from $\betav\sim\Nor\left[\muv_{\beta}, (\Amat + \Phimat^T \Omegamat \Phimat)^{-1}\right]$, with diagonal $\Amat$ and $\Omegamat$.
		The distribution parameters are randomly generated and computation time averaged over 50 random trials is displayed.	
	}
	\label{fig:A4}
\end{figure}

We fix $n = 4000$, vary $p$ from 1 to 8000, and mimic the settings in Section \ref{sec:A1vsA2} to randomly generate $\Phimat$, diagonal $\Amat$, and diagonal $\Omegamat$. 
As a function of dimensions $p$, the computation time for one sample averaged over 50 random trials is shown in Figure \ref{fig:A4}.
It is evident that, identical to the complexity analysis in Tables \ref{tab:CompNaive4} and \ref{tab:CompAlg4}, Algorithm \ref{alg:4} has a linear complexity with respect to $p$ under these settings, which will bring significant acceleration in a high-dimensional setting with $p \gg n$.
If the sample size $n$ is large enough that $n>p$, 
then one may directly apply the naive Cholesky decomposition based implementation.

Algorithm  \ref{alg:4} could be slightly modified to be applied to high-dimensional regression, 
where the main objective is to efficiently sample from the conditional posterior of $\betav \in \Rbb^{p \times 1}$ in the linear regression model as
\beq
\tv \sim\Nor(\Phimat \betav,\Omegamat^{-1}),~\betav\sim\Nor(\mathbf{0},\Amat^{-1}),
\eeq
where $\Phimat\in\mathbb{R}^{n\times p}$, $\Omegamat \in \Rbb^{n \times n}$, and different constructions on $\Amat\in\mathbb{R}^{p\times p}$ lead to a wide variety of regression models \citep{caron2008sparse,carvalho2010horseshoe,polson2014bayesian}.
The conditional posterior of $\betav$ is directly derived and shown in the following example, where its simulation algorithm is summarized by further generalizing Corollary~\ref{cor3}.


\noindent\textbf{Example 4:} \textit{
	Simulation of the MVN distribution 
	$$\betav\sim\Nor\left[(\Amat + \Phimat^T \Omegamat \Phimat)^{-1}\Phimat^T\Omegamat\tv, ~~(\Amat + \Phimat^T \Omegamat \Phimat)^{-1}\right]$$
	can be realized as follows.
	\begin{itemize}
		\item Sample $\yv_1\sim\Nor (\mathbf{0},\Amat^{-1})$ and $\yv_2\sim\Nor (\mathbf{0},\Omegamat^{-1})$ ;
		\item
		Return
		$
		\betav =\yv_1 + \Amat^{-1} \Phimat^T (\Omegamat^{-1}+\Phimat\Amat^{-1}\Phimat^T)^{-1} \left(\tv - \Phimat \yv_1-\yv_2\right)
		$, which can be realized using
		\begin{itemize}
			\item Solve $\alphav$ such that $ (\Omegamat^{-1}+\Phimat\Amat^{-1}\Phimat^T)\alphav = \tv - \Phimat \yv_1-\yv_2$;
			\item Return $\betav =\yv_1 + \Amat^{-1} \Phimat^T \alphav$. 
		\end{itemize}
	\end{itemize}
}

Note that if $\Omegamat = \Imat_n$, then the simulation algorithm in Example 4 reduces to the one in Proposition 2.1 of \citet{bhattacharya2015fast}, which is shown there to be significantly more efficient than that of \citet{rue2001fast} for high-dimensional regression if $p\gg n$.

\section{Conclusions}

A fast and exact simulation algorithm is developed for a multivariate normal (MVN) distribution whose sample space is constrained on the intersection of a set of hyperplanes, which is shown to be inherently related to the conditional distribution of a unconstrained MVN distribution. The proposed simulation algorithm is further generalized to efficiently simulate from a MVN distribution, whose covariance (precision) matrix can be decomposed as the sum (difference) of a positive-definite matrix and a low-rank symmetric matrix, using a higher dimensional hyperplane-truncated MVN distribution whose covariance matrix is block-diagonal. 

%




\bibliographystyle{ba}
\bibliography{References052016,References112016}

\begin{acknowledgement}
 The authors would like to thank the editor-in-chief, editor, associate editor, and two anonymous referees for their 
comments and suggestions, which have helped us improve the paper
substantially.  
M. Zhou thanks Yingbo Li and Xiaojing Wang for helpful discussions. 
B. Chen thanks the support of the Thousand Young Talent Program of China, NSFC
(61372132), NCET-13-0945, and NDPR-9140A07010115DZ01015.
\end{acknowledgement}

\newpage
\appendix 
\begin{center}
\Large{\bf{Appendix}}
\end{center}

\vspace{2mm}
\begin{algorithm}[H]
\begin{itemize}
\item Sample $\yv\sim\Nor (\muv,\Sigmamat)$ and denote $\yv_1 = (y_1,\ldots,y_{k_1})^T$ and $\yv_2 = (y_{k_1+1},\ldots,y_k)^T$;
\item Return $\xv_1 = \yv_1 + \Sigmamat_{12}\Sigmamat_{22}^{-1}( \rv- \yv_2)$.
\end{itemize}
 \caption{\label{thm:1} \citep{hoffman1991constrained,doucet2010note} Simulation of the conditional distribution of $\xv_1$ given $\xv_2=\rv$ as
$
\xv_1 \given \xv_2 = \rv \sim\Nor \left[\muv_1 + \Sigmamat_{12}\Sigmamat_{22}^{-1}(\rv-\muv_2), ~~\Sigmamat_{11} - \Sigmamat_{12}\Sigmamat_{22}^{-1}\Sigmamat_{21}\right]$, where the joint distribution of $\xv=(\xv_1^T,\xv_2^T)$ follows $\xv\sim(\muv,\Sigmamat)$.} 
\end{algorithm}\vspace{2mm}


\subsection*{Proofs}
\begin{proof}[Proof of Theorem \ref{maintheorem1}]
Let us denote $\Lambdamat = \Hmat^T\Sigmamat^{-1}\Hmat$
as a precision matrix that can be partitioned as in \eqref{eq:precision}. 
For Algorithm \ref{alg:1}, instead of directly simulating $\zv_1$ given $\zv_2$ using the conditional distribution of the MVN, we apply Algorithm \ref{thm:1} \citep{hoffman1991constrained,doucet2010note} to modify its sampling steps as follows. 
\begin{itemize}
\item Sample $\tilde \zv\sim\Nor [\Hmat^{-1}\muv,\Hmat^{-1}\Sigmamat(\Hmat^{-1})^T]$, and denote $\tilde \zv_1 = (z_1,\ldots,z_{k_1})^T$ and $\tilde \zv_2 = (z_{k_1+1},\ldots,z_k)$;
\item Let $\zv=(\zv_1^T,\zv_2^T)^T$, where $\zv_2=(\Gmat\Hmat_2)^{-1}\rv$ and $\zv_1 = \tilde \zv_1 - \Lambdamat_{11}^{-1} \Lambdamat_{12}( \zv_2 - \tilde \zv_2)$, and return 
\begin{align}
\xv &= \Hmat\zv=\Hmat_1 \zv_1 + \Hmat_2 (\Gmat\Hmat_2)^{-1}\rv\notag\\
 &= \Hmat_1\tilde \zv_1 + \Hmat_1\Lambdamat_{11}^{-1} \Lambdamat_{12} \tilde \zv_2 + (\Hmat_2 - \Hmat_1\Lambdamat_{11}^{-1} \Lambdamat_{12})(\Gmat\Hmat_2)^{-1}\rv\notag\\
& = (\Hmat_1, \Hmat_1\Lambdamat_{11}^{-1} \Lambdamat_{12})\tilde \zv + (\Hmat_2 - \Hmat_1\Lambdamat_{11}^{-1} \Lambdamat_{12})(\Gmat\Hmat_2)^{-1}\rv.
\end{align}

\end{itemize}
Therefore, we can equivalently generate $\xv$ as follows. 
\begin{itemize}
\item Sample $\yv\sim\Nor (\muv,\Sigmamat)$.
\item Return $\xv = (\Hmat_1, \Hmat_1\Lambdamat_{11}^{-1} \Lambdamat_{12}) \Hmat^{-1}\yv+ (\Hmat_2 - \Hmat_1\Lambdamat_{11}^{-1} \Lambdamat_{12})(\Gmat\Hmat_2)^{-1}\rv $.
\end{itemize}

The computation can be significantly simplified if $\Lambdamat_{12} = \mathbf{0}$, which means
$$
\Lambdamat_{12}=\Hmat^T_1\Sigmamat^{-1} \Hmat_2 = \mathbf{0}.
$$
Since $\Hmat_1^T\Gmat^T=0$ by definition, to make $\Lambdamat_{12} = \mathbf{0}$, if and only if we have $\Hmat_2$ as
$$
\Hmat_2 = \Sigmamat \Gmat^T \Mmat,
$$
where $\Mmat\in\mathbb{R}^{k_2\times k_2}$ is an arbitrary full rank matrix, under which we have
\begin{itemize}
\item Sample $\yv\sim\Nor (\muv,\Sigmamat)$.
\item Return $\xv = (\Hmat_1, \mathbf{0}_{k\times k_2} ) \Hmat^{-1}\yv+ \Sigmamat \Gmat^T(\Gmat \Sigmamat \Gmat^T)^{-1}\rv $, or return
\beq\label{eq:xv}
\xv = \yv - (\mathbf{0}_{k\times k_1},\Sigmamat \Gmat^T \Mmat ) \Hmat^{-1}\yv+ \Sigmamat \Gmat^T(\Gmat \Sigmamat \Gmat^T)^{-1}\rv .
\eeq
\end{itemize}
Let us denote $ (\mathbf{0}_{k\times k_1},\Sigmamat \Gmat^T \Mmat )\Hmat^{-1} = \Cmat $. We have $$ (\mathbf{0}_{k\times k_1},\Sigmamat \Gmat^T\Mmat ) = \Cmat \Hmat = (\Cmat \Hmat_1,\Cmat \Sigmamat \Gmat^T\Mmat)$$ and hence $\Cmat \Hmat_1 = \mathbf{0}$ and $\Cmat \Sigmamat \Gmat^T \Mmat= \Sigmamat \Gmat^T \Mmat$. Since $\Gmat \Hmat_1=0$, we have $\Cmat = \Sigmamat \Gmat^T(\Gmat \Sigmamat \Gmat^T)^{-1}\Gmat$. The proof is completed by substituting $ (\mathbf{0}_{k\times k_1},\Sigmamat \Gmat^T \Mmat )\Hmat^{-1}$ in \eqref{eq:xv} with $ \Sigmamat \Gmat^T(\Gmat \Sigmamat \Gmat^T)^{-1}\Gmat$.
%
\end{proof}



\begin{proof}[Alternative Proof of Theorem \ref{maintheorem1}]
To solve the problem in \eqref{eq:px}, one may solve an equivalent problem in \eqref{eq:pztruccond} by defining an invertible transformation matrix $\Hmat$ that satisfies
$\Gmat\Hmat_1 = \mathbf{0}_{k_2\times k_1}$. Let us denote $\Lambdamat = \Hmat^T\Sigmamat^{-1}\Hmat$
as a precision matrix that can be partitioned as in \eqref{eq:precision}. To simply the problem in \eqref{eq:pztruccond}, we choose the transformation matrix $\Hmat$ to make $\zv_1$ and $\zv_2$ be independent to each other. Since $\zv$ follows a MVN distribution, $\zv_1$ and $\zv_2$ are independent to each other if and only if 
%
%
%
%
$$
\Lambdamat_{12}=\Hmat^T_1\Sigmamat^{-1} \Hmat_2 = \mathbf{0}.
$$
Since $\Hmat_1^T\Gmat^T=\mathbf{0}$ by definition, to make $\Lambdamat_{12} = \mathbf{0}$, if and only if we have $\Hmat_2$ as
$$
\Hmat_2 = \Sigmamat \Gmat^T \Mmat,
$$
where $\Mmat\in\mathbb{R}^{k_2\times k_2}$ is an arbitrary full rank matrix.
Accordingly, we have
$$
\Hmat^{-1}\Sigmamat(\Hmat^{-1})^T = \left[
 \barr{cc}
 (\Hmat^T_1\Sigmamat^{-1} \Hmat_1)^{-1} & \mathbf{0} \\
 \mathbf{0} & (\Hmat^T_2\Sigmamat^{-1} \Hmat_2)^{-1}
 \earr \right] .
$$
Thus with $\Hmat$ satisfying $\Gmat\Hmat_1 = \mathbf{0}_{k_2\times k_1}$ and $\Hmat_2 = \Sigmamat \Gmat^T \Mmat$, one can transform the original problem in \eqref{eq:px} to that in \eqref{eq:pztruccond}, where 
$\zv_1$ and $\zv_2$ are independent and the restrictions $\Gmat\xv=\rv$ and $\zv_2 = (\Gmat\Hmat_2)^{-1}\rv$ imply each other.
Following the naive approach shown in Algorithm \ref{alg:1}, one can generate $\xv$ from \eqref{eq:px} as follows
\begin{itemize}
\item Find $\Hmat = (\Hmat_1,\Hmat_2)$ with $\Hmat_2 = \Sigmamat \Gmat^T \Mmat$ and with $\Hmat_1$ satisfying $\Gmat\Hmat_1 = \mathbf{0}_{k_2\times k_1}$;
\item Sample $\zv_1 \sim \Nor [(\Imat_{k_1},\mathbf{0}_{k_1\times k_2})\Hmat^{-1}\muv,
 (\Hmat^T_1\Sigmamat^{-1} \Hmat_1)^{-1}]$;
\item Return $\xv = \Hmat \left[
 \barr{c}
 \zv_1 \\
 \Mmat^{-1} (\Gmat \Sigmamat \Gmat^T )^{-1} \rv
 \earr \right]$.
\end{itemize}
However, this naive approach contains intermediate variables that could be computationally expensive to compute.
Below we present a method to bypass these intermediate variables.
Since the last step could be reexpressed as 
$$ \bali
 \xv &= \Hmat_1 \zv_1 + \Hmat_2
 \Mmat^{-1} (\Gmat \Sigmamat \Gmat^T )^{-1}\rv \\
 	 &= \Hmat_1 \zv_1 + \Hmat_2 \zv_2 + \Hmat_2
 [ \Mmat^{-1} (\Gmat \Sigmamat \Gmat^T )^{-1}\rv - \zv_2 ] \\
 &= \Hmat \zv + \Sigmamat \Gmat^T (\Gmat \Sigmamat \Gmat^T )^{-1}\rv - \Sigmamat \Gmat^T (\Mmat \zv_2),
\eali $$
where $\zv=(\zv_1^T,\zv_2^T)^T$ and $\zv_2\in\mathbb{R}^{k_2}$ is a vector whose values can be chosen arbitrarily. In addition, since $\Mmat$ is an arbitrary full-rank matrix, we can let
\beq
\zv_2 \sim \Nor [(\mathbf{0}_{k_2\times k_1},\Imat_{k_2})\Hmat^{-1}\muv,
 (\Hmat^T_2\Sigmamat^{-1} \Hmat_2)^{-1}], \notag
\eeq
which means $\zv\sim \Nor [\Hmat^{-1}\muv,
 \Hmat^{-1}\Sigmamat (\Hmat^{-1})^T]$,
and choose $\Mmat$ to make 
$$
\Gmat\xv = \Gmat\Hmat \zv  + \Gmat \Sigmamat \Gmat^T (\Gmat \Sigmamat \Gmat^T )^{-1}\rv - \Gmat \Sigmamat \Gmat^T (\Mmat \zv_2) = \rv,
$$
which means $ \Gmat \Sigmamat \Gmat^T (\Mmat \zv_2) =\Gmat\Hmat \zv  $.
Thus we have
$$
 \xv = \Hmat\zv + \Sigmamat \Gmat^T (\Gmat \Sigmamat \Gmat^T )^{-1} (\rv - \Gmat \Hmat\zv).
$$
In addition, since if $\zv\sim \Nor [\Hmat^{-1}\muv,
 \Hmat^{-1}\Sigmamat (\Hmat^{-1})^T]$, then $\yv = \Hmat \zv \sim \Nor(\muv,\Sigmamat)$. 
 Therefore, without the need to compute any intermediate variables, one may use Algorithm \ref{alg:2} to generate $\xv$ from the hyperplane truncated MVN distribution.
%

\end{proof}

\begin{proof}[Proof of Theorem \ref{thm_condition}]

Using the matrix inversion lemma on \eqref{eq:px1}, we have
\begin{align}\label{eq:px1_pdf}
&p(\xv_1) 
\propto \exp\left[ - \frac{1}{2}( \xv_1- \muv_1)^T \left(\Sigmamat_{11}- \Sigmamat_{12} \Sigmamat_{22}^{-1} \Sigmamat_{21}\right)^{-1} ( \xv_1- \muv_1) \right],\notag\\
 \propto& \exp\left\{ - \frac{1}{2}( \xv_1- \muv_1) ^T \left[\Sigmamat_{11}^{-1} +\Sigmamat_{11}^{-1} \Sigmamat_{12} \left(\Sigmamat_{22} - \Sigmamat_{21} \Sigmamat_{11}^{-1} \Sigmamat_{12}\right)^{-1} \Sigmamat_{21} \Sigmamat_{11}^{-1} \right] ( \xv_1- \muv_1) \right\}.
\end{align}
Using \eqref{eq:tilde_mu2}, we have $\Gmat\muv =\Gmat_1 \muv_1+ \Gmat_2 \muv_2= \rv$. Since $\Gmat \xv =\rv$, we further have $\Gmat (\xv-\muv)=0$ and hence $\Gmat_{1}(\xv_1- \muv_1) = - \Gmat_{2}(\xv_2- \muv_2)$. Therefore, given the construction of $ \muv_2$ as in \eqref{eq:tilde_mu2}, we can replace the equality constraint $\Gmat \xv =\rv$ on $\xv$ by requiring $ (\xv_2- \muv_2) = - \Gmat_{2}^{-1}\Gmat_{1}(\xv_1- \muv_1)$. Using this equivelent constraint together with \eqref{eq:px}, we have
\begin{align}\label{eq:px_pdf}
&p(\xv\given \muv,\tilde \Sigmamat, \Gmat,\rv )\notag\\
 \propto &\exp\left[ - \frac{1}{2}( \xv_1- \muv_1)^T \Sigmamat_{11}^{-1} ( \xv_1- \muv_1) - \frac{1}{2}( \xv_2- \muv_2)^T \tilde\Sigmamat_{22}^{-1} ( \xv_2- \tilde \muv_2)\right]\delta( \Gmat \xv = \rv ) \notag\\
 \propto& \exp\left\{ - \frac{1}{2}( \xv_1- \muv_1)^T \left[ \Sigmamat_{11}^{-1} + \Gmat_{1}^T(\Gmat_{2}^{-1})^T\tilde\Sigmamat_{22}^{-1} \Gmat_{2}^{-1}\Gmat_{1} \right]( \xv_1- \muv_1) \right\}\notag\\
 &\times\delta\left[\xv_2 = \muv_2- \Gmat_{2}^{-1}\Gmat_{1}(\xv_1- \muv_1)\right]\notag\\
 =&\mathcal{N}\left\{\xv_1; \muv_1, \left[\Sigmamat_{11}^{-1} + \Gmat_{1}^T(\Gmat_{2}^{-1})^T\tilde\Sigmamat_{22}^{-1} \Gmat_{2}^{-1}\Gmat_{1}\right]^{-1}\right\}\notag\\
 &\times \delta\left[\xv_2 = \muv_2- \Gmat_{2}^{-1}\Gmat_{1}(\xv_1- \muv_1)\right]
 \end{align}
 It is clear that the marginal distribution of $\xv_1$ in \eqref{eq:px_pdf} matches the conditional distribution of $\xv_1$ in \eqref{eq:px1_pdf} if we further construct 
 $\tilde \Sigmamat_{22}$ using \eqref{eq:tilde_sig22}.
\end{proof}

%

\begin{proof}[Proof of Corollary \ref{cor2}] Applying Theorem \ref{maintheorem1}
to Corollary \ref{cor1}, we can generate $\xv$ with
\begin{itemize}
\item Sample $\yv\sim\Nor \left( \muv,\tilde \Sigmamat \right)$; 
\item Return $\xv = \yv + \tilde\Sigmamat \Gmat^T [ \Gmat \tilde \Sigmamat \Gmat^T ]^{-1} (\rv - \Gmat \yv)$.
\end{itemize}
Since
$ \muv = \begin{bmatrix}
 \muv_1\\
\mathbf{0}
 \end{bmatrix}$, $
\tilde\Sigmamat = \begin{bmatrix}
 \Sigmamat_{11} & \mathbf{0} \\
\mathbf{0} & \Sigmamat_{22} - \Sigmamat_{21} \Sigmamat_{11}^{-1} \Sigmamat_{12}
 \end{bmatrix}
$, $\Gmat=(\Gmat_1,\Gmat_2)=(\Sigmamat_{21}\Sigmamat_{11}^{-1} ,\Imat_{k_2})$, and $\rv = \Sigmamat_{21}\Sigmamat_{11}^{-1} \muv_1$, we have
$\tilde \Sigmamat \Gmat^T =
 \begin{bmatrix}
 \Sigmamat_{12} \\
\Sigmamat_{22} - \Sigmamat_{21} \Sigmamat_{11}^{-1} \Sigmamat_{12}
 \end{bmatrix}
$ and $[ \Gmat \tilde \Sigmamat \Gmat^T ]^{-1} = \Sigmamat_{22}^{-1}$. Since $\tilde \Sigmamat$ is block diagonal, we can independently sample $\yv_1$ and $\yv_2$ as $\yv_1\sim\Nor ( \muv_1,\Sigmamat_{11})$ and $\yv_2\sim\Nor (\mathbf{0},\Sigmamat_{22}- \Sigmamat_{21} \Sigmamat_{11}^{-1} \Sigmamat_{12})$, respectively, with which we can further sample $\xv$ as
$$ \begin{bmatrix}
 \xv_1 \\
\xv_2
 \end{bmatrix} = \begin{bmatrix}
 \yv_1 \\
\yv_2
 \end{bmatrix} + \begin{bmatrix}
 \Sigmamat_{12} \\
\Sigmamat_{22} - \Sigmamat_{21} \Sigmamat_{11}^{-1} \Sigmamat_{12}
 \end{bmatrix} \Sigmamat_{22}^{-1} ( \Sigmamat_{21} \Sigmamat_{11}^{-1} \muv_1-\Sigmamat_{21} \Sigmamat_{11}^{-1}\yv_1 -\yv_2). $$
Thus we can let $\xv_1 =\muv_1+\yv'_1 - \Sigmamat_{12} \Sigmamat_{22}^{-1} \left( \Sigmamat_{21} \Sigmamat_{11}^{-1} \yv'_1 +\yv_2\right) $, where $\yv'_1 = \yv_1-\muv_1 \sim\Nor(\mathbf{0},\Sigmamat_{11})$.
\end{proof}

 \begin{proof}[Proof of Corollary \ref{cor3}]
 Using the matrix inversion lemma, we have
 \beq
 \Sigmamat_{\beta}= \Amat^{-1}-\Amat^{-1} \Phimat^T (\Omegamat^{-1}+\Phimat\Amat^{-1}\Phimat^T)^{-1}\Phimat\Amat^{-1}.
\eeq
The proof is completed by
using Corollary \ref{cor2} with $\muv_1 =\muv_{\beta}$, $\Sigmamat_{11}=\Amat^{-1}$, $\Sigmamat_{12} =\Amat^{-1}\Phimat^T $, and $\Sigmamat_{22} = \Omegamat^{-1}+\Phimat\Amat^{-1}\Phimat^T$. 
\end{proof}

\subsection*{Computational Complexity}

We present the computational complexities of all proposed algorithms in the following tables, where we highlight with bold the lowest complexity in each row. 


\begin{table}[H]
	\centering
	\caption{Computational complexity of Algorithm \ref{alg:1}.}
	\label{tab:CompAlg1} 
	\begin{tabular}{ccc}
		\hline \hline
		\multirow{2}{*}{Calculation}  & \multicolumn{2}{c}{Computational complexity}\\
		\cline{2-3}
		& Non-diagonal $\Sigmamat$ & Diagonal $\Sigmamat$\\
		\hline
		$\Hmat$ 	& 	$\Oc(k_2 k^2 )$ 	& 	$\Oc(k_2 k^2 )$	\\
		$\zv_2$ 	& 	$\Oc(k_2^2 k )$ 	& 	$\Oc(k_2^2 k)$	\\
		$\Sigmamat^{-1}$ 	& 	$\Oc(k^3)$ 		& 	\boldmath $\Oc(k)$	\\
		$\Lambdamat_{11}$ 	& 	$\Oc(k_1 k^2)$ 	& 	\boldmath $\Oc(k_1^2 k)$	\\
		$\Lambdamat_{12}$ 	& 	$\Oc(k_1 k_2 k)$ 	& 	$\Oc(k_1 k_2 k)$	\\
		$\muv_{\zv_1}$ 	& 	$\Oc(\max(k^2,k_1^3,k_1^2 k_2))$ 	& 	$\Oc(\max(k^2,k_1^3,k_1^2 k_2))$	\\
		$\zv_1$ 	& 	$\Oc( k_1^3 )$	 & 	$\Oc( k_1^3 )$	\\
		$\xv$ 	& 	$\Oc( \max(k_1 k , k_2 k) )$	 & 	$\Oc( \max(k_1 k , k_2 k) )$	\\
		\hline
		Summary 	& 	$\Oc( k^3 )$	&	\boldmath $\Oc( \max(k_2 k^2, k_1^2 k) )$	\\
		\hline 	\hline
	\end{tabular}
\end{table}

\begin{table}[H]
	\centering
	\caption{Computational complexity of Algorithm \ref{alg:2}.}
	\label{tab:CompAlg2} 
	\begin{tabular}{ccc}
		\hline \hline
		\multirow{2}{*}{Calculation}  & \multicolumn{2}{c}{Computational complexity}\\
		\cline{2-3}
		& Non-diagonal $\Sigmamat$ & Diagonal $\Sigmamat$\\
		\hline
		$\yv$ 		& 	$\Oc(k^3 )$ 		& 	\boldmath $\Oc(k)$	\\
		$\Gmat \Sigmamat \Gmat^T$ 		& 	$\Oc(k_2 k^2 )$ 	& 	\boldmath $\Oc(k_2^2 k)$	\\
		$\alphav$ 		& 	$\Oc( \max(k_2 k,k_2^3) )$ 	& 	$\Oc(\max(k_2 k,k_2^3))$	\\
		$\xv$ 	& 	$\Oc( k_2 k )$	 & 	$\Oc(k_2 k)$	\\
		\hline
		Summary 	& 	$\Oc( k^3 )$	&	\boldmath $\Oc( k_2^2 k )$	\\
		\hline 	\hline
	\end{tabular}
\end{table}

\begin{table}[H]
	\centering
	\caption{Computational complexity of naive simulation in Algorithm \ref{alg:3}.}
	\label{tab:CompNaive3} \makebox[\linewidth]{
\resizebox{\linewidth}{!}{%
	\begin{tabular}{ccccc}
		\hline \hline
		\multirow{3}{*}{Calculation}  & \multicolumn{4}{c}{Computational complexity}\\
		\cline{2-5}
		& Non-diagonal $\Sigmamat_{11}$ & Diagonal $\Sigmamat_{11}$ & Non-diagonal $\Sigmamat_{11}$ & Diagonal $\Sigmamat_{11}$\\
		& Non-diagonal $\Sigmamat_{22}$ & Non-diagonal $\Sigmamat_{22}$ & Diagonal $\Sigmamat_{22}$ & Diagonal $\Sigmamat_{22}$\\
		\hline
		$\Sigmamat_{22}^{-1}$	& 	$\Oc(k_2^3)$ 	& 	$\Oc(k_2^3)$	& 	\boldmath $\Oc(k_2)$ 	& 	\boldmath $\Oc(k_2)$	\\
		$\Sigmamat_{12} \Sigmamat_{22}^{-1} \Sigmamat_{21}$		& 	$\Oc(\max(k_1^2 k_2 , k_1 k_2^2) )$ 	& 	$\Oc(\max(k_1^2 k_2 , k_1 k_2^2) )$	& 	\boldmath $\Oc(k_1^2 k_2)$ 	& 	\boldmath $\Oc(k_1^2 k_2)$	\\
		$\xv_1$ 	& 	$\Oc( k_1^3 )$	& 	$\Oc( k_1^3 )$	& 	$\Oc( k_1^3 )$	 & 	$\Oc( k_1^3 )$	\\
		\hline
		Summary 	& 	$\Oc( \max(k_1^3 , k_2^3) )$	&	$\Oc( \max(k_1^3 , k_2^3) )$	& 	\boldmath $\Oc( \max(k_1^3 , k_1^2 k_2) )$	&	\boldmath $\Oc( \max(k_1^3 , k_1^2 k_2) )$	\\
		\hline 	\hline
	\end{tabular}}}
\end{table}

\begin{table}[H]
	\centering
	\caption{Computational complexity of Algorithm \ref{alg:3}.}
	\label{tab:CompAlg3} 
	\makebox[\linewidth]{
\resizebox{\linewidth}{!}{%
	\begin{tabular}{ccccc}
		\hline \hline
		\multirow{3}{*}{Calculation}  & \multicolumn{4}{c}{Computational complexity}\\
		\cline{2-5}
		& Non-diagonal $\Sigmamat_{11}$ & Diagonal $\Sigmamat_{11}$ & Non-diagonal $\Sigmamat_{11}$ & Diagonal $\Sigmamat_{11}$\\
		& Non-diagonal $\Sigmamat_{22}$ & Non-diagonal $\Sigmamat_{22}$ & Diagonal $\Sigmamat_{22}$ & Diagonal $\Sigmamat_{22}$\\
		\hline
		$\yv_1$	& 	$\Oc(k_1^3)$ 	& 	\boldmath $\Oc(k_1)$	& 	$\Oc(k_1^3)$ 	& 	\boldmath $\Oc(k_1)$	\\
		$\Sigmamat_{11}^{-1}$	&	$\Oc(k_1^3)$ 	& 	\boldmath $\Oc(k_1)$	& 	$\Oc(k_1^3)$ 	& 	\boldmath $\Oc(k_1)$	\\
		$\Sigmamat_{21} \Sigmamat_{11}^{-1} \Sigmamat_{12}$		& 	$\Oc( \max(k_1 k_2^2, k_1^2 k_2))$ 	& 	\boldmath $\Oc(k_1 k_2^2)$	& 	$\Oc( \max(k_1 k_2^2, k_1^2 k_2))$ 	& 	\boldmath $\Oc(k_1 k_2^2)$	\\
		$\yv_{2}$	& 	$\Oc(k_2^3)$ 	& 	$\Oc(k_2^3)$	& 	$\Oc(k_2^3)$ 	& 	$\Oc(k_2^3)$	\\
		$\alphav$	& 	$\Oc( \max(k_1 k_2, k_2^3) )$ 	& 	$\Oc( \max(k_1 k_2, k_2^3) )$	& 	\boldmath $\Oc(k_1 k_2)$ 	& 	 \boldmath $\Oc(k_1 k_2)$	\\
		$\xv_1$ 	& 	$\Oc( k_1 k_2 )$	& 	$\Oc( k_1 k_2 )$	& 	$\Oc( k_1 k_2 )$	 & 	$\Oc( k_1 k_2 )$	\\
		\hline
		Summary 	& 	$\Oc( \max( k_1^3, k_2^3) )$	&	\boldmath $\Oc( \max(k_1 k_2^2 , k_2^3) )$	& 	$\Oc( \max( k_1^3, k_2^3) )$	&	\boldmath $\Oc( \max(k_1 k_2^2 , k_2^3) )$	\\
		\hline 	\hline
	\end{tabular}}}
\end{table}

\begin{table}[H]
	\centering
	\caption{Computational complexity of naive simulation in Algorithm \ref{alg:4}.}
	\label{tab:CompNaive4} 
	\makebox[\linewidth]{
\resizebox{\linewidth}{!}{%
	\begin{tabular}{ccccc}
		\hline \hline
		\multirow{3}{*}{Calculation}  & \multicolumn{4}{c}{Computational complexity}\\
		\cline{2-5}
		& Non-diagonal $\Amat$ & Diagonal $\Amat$ & Non-diagonal $\Amat$ & Diagonal $\Amat$\\
		& Non-diagonal $\Omegamat$ & Non-diagonal $\Omegamat$ & Diagonal $\Omegamat$ & Diagonal $\Omegamat$\\
		\hline
		$\Phimat^T \Omegamat \Phimat$		& 	$\Oc( \max(n^2 p, n p^2) )$ 	& $\Oc( \max(n^2 p, n p^2) )$	& 	\boldmath $\Oc(n p^2)$ 	& 	\boldmath $\Oc(n p^2)$	\\
		$(\Amat + \Phimat^T \Omegamat \Phimat)^{-1}$	&	$\Oc(p^3)$ 	& 	$\Oc(p^3)$	& 	$\Oc(p^3)$ 	& 	$\Oc(p^3)$	\\
		$\betav$ 	& 	$\Oc(p^3)$ 	& 	$\Oc(p^3)$	& 	$\Oc(p^3)$ 	& 	$\Oc(p^3)$	\\
		\hline
		Summary 	& 	$\Oc( \max( n^2 p, p^3 ) )$	&	$\Oc( \max( n^2 p, p^3 ) )$	& 	\boldmath $\Oc( \max(n p^2 , p^3)$	&	\boldmath $\Oc( \max(n p^2 , p^3)$	\\
		\hline 	\hline
	\end{tabular}}}
\end{table}

\begin{table}[H]
	\centering
	\caption{Computational complexity of Algorithm \ref{alg:4}.}
	\label{tab:CompAlg4} 
	\makebox[\linewidth]{
\resizebox{\linewidth}{!}{%
	\begin{tabular}{ccccc}
		\hline \hline
		\multirow{3}{*}{Calculation}  & \multicolumn{4}{c}{Computational complexity}\\
		\cline{2-5}
		& Non-diagonal $\Amat$ & Diagonal $\Amat$ & Non-diagonal $\Amat$ & Diagonal $\Amat$\\
		& Non-diagonal $\Omegamat$ & Non-diagonal $\Omegamat$ & Diagonal $\Omegamat$ & Diagonal $\Omegamat$\\
		\hline
		$\Amat^{-1}$	& 	$\Oc(p^3)$ 	& 	\boldmath $\Oc(p)$	& 	$\Oc(p^3)$ 	& 	\boldmath $\Oc(p)$	\\
		$\yv_1$			& 	$\Oc(p^3)$ 	& 	\boldmath $\Oc(p)$	& 	$\Oc(p^3)$ 	& 	\boldmath $\Oc(p)$	\\
		$\Omegamat^{-1}$	&	$\Oc(n^3)$ 	& 	$\Oc(n^3)$	& 	\boldmath $\Oc(n)$ 	& 	\boldmath $\Oc(n)$	\\
		$\yv_2$			& 	$\Oc(n^3)$ 	& 	$\Oc(n^3)$	& 	\boldmath $\Oc(n)$ 	& 	\boldmath $\Oc(n)$	\\
		$\Omegamat^{-1} + \Phimat \Amat^{-1} \Phimat^T$		& 	$\Oc( \max(n p^2, n^2 p) )$ 	& 	\boldmath $\Oc(n^2 p)$	& 	$\Oc( \max(n p^2, n^2 p) )$ 	& 	\boldmath $\Oc(n^2 p)$	\\
		$\alphav$	& 	$\Oc(\max(np,n^3))$ 	& 	$\Oc(\max(np,n^3))$	& 	$\Oc(\max(np,n^3))$ 	& 	$\Oc(\max(np,n^3))$	\\
		$\betav$	& 	$\Oc(np)$ 	& 	$\Oc(np)$	& 	$\Oc(np)$ 	& 	$\Oc(np)$	\\
		\hline
		Summary 	& 	$\Oc( \max( n^3, p^3) )$	&	\boldmath $\Oc( \max(n^2 p, n^3) )$	& 	$\Oc( \max( n^3, p^3) )$	&	\boldmath $\Oc( \max(n^2 p, n^3) )$	\\
		\hline 	\hline
	\end{tabular}}}
\end{table}


\subsection*{Brief derivation of SG-MCMC for a simplex-constrained 
 vector} 

 	
 	Based on a comprehensive framework for constructing SG-MCMC algorithms in \citet{ma2015complete}, we have the mini-batch update rule for a global variable $\zv$ as 
	\beqs\label{eq:FramZmb} 
 	&{\zv_{t + 1}} \! = \! {\zv_t} + {\varepsilon _t} \! \left\{ - \big[ \Dmat \! \left( \zv_t \right) \! + \! \Qmat \! \left( \zv_t \right) \big]
 	\nabla \tilde H \! \left( \zv_t \right) \! + \! \Gamma \! \left( \zv_t \right) \right\}\notag \\
 	& ~+ \Nc \left( \bds 0,{\varepsilon _t} \left[ 2 \Dmat \left( \zv_t \right) - {\varepsilon _t} \hat\Bmat_t \right] \right),
 	\eeqs
 	where ${\varepsilon _t}$ are annealed step sizes, $\Dmat\left( \zv \right)$ is a positive semidefinite diffusion matrix, $\Qmat\left( \zv \right)$ is a skew-symmetric curl matrix, $\hat \Bmat_t$ is an estimate of the stochastic gradient noise variance satisfying a positive definite constraint as $2\Dmat\left( \zv_t \right) - {\varepsilon _t}\hat \Bmat_t \succ \bds0$, and $\Gamma_i(\zv)$, the $i^{th}$ element of the compensation vector $\Gamma(\zv)$, is defined as 
 	${\Gamma _i}\left( \zv \right) = \sum\nolimits_j \frac{\partial}{\partial z_j} \left[ \Dmat_{ij} \left( \zv \right) + \Qmat_{ij}\left( \zv \right) \right]$.
 	The mini-batch estimation of energy function is defined as $\tilde H\left( \zv \right) = - \ln p\left( \zv \right) - \rho \sum_{\xv \in \tilde X} {\ln p\left( {\xv\left| \zv \right.} \right)} $, with $\tilde X$ the mini-batch and $\rho $ the ratio of the dataset size to the mini-batch size.
 	
 	For simplicity, we adopt the same specifications that lead to the stochastic gradient Riemannian Langevin dynamics (SGRLD) inference algorithm for simplex-constrained model parameters \citep{patterson2013stochastic,ma2015complete}, 
	 namely $\Dmat \left( \zv \right) = \Gmat{\left( \zv \right)^{ - 1}}$, $\Qmat\left( \zv \right) = \bds 0$, and ${\hat \Bmat_t} = \bds0$, where $\Gmat\left( \zv \right)$ denotes the Fisher information matrix (FIM)~\citep{girolami2011riemann}. 
 	
 	With the multinomial likelihood $\nv_j \sim \Mult \left( n_{\cdotv j}, \phiv \right)$, and the reduced-mean parameterization $\varphiv \in\mathbb{R}_+^{V-1}$, where $j \in \{ 1, \cdots, N\}$ with $N$ the dataset size, it is straight to derive the FIM as 
 	\begin{align}\label{eq:RiemanPhiGBN}
	 	\Gmat \left( \varphiv \right) &  =  - \Ebb \left[ \frac{\partial^2}{\partial \varphiv^2}
	 	\ln \left( \prod\nolimits_j  \Mult \left[ \nv_j ; n_{\cdotv j} , \phiv \right] \right) \right] \notag\\
	 	& = M \left[ \diag \left( {1}/{\varphiv} \right) + \bds 1 \bds 1 ^T / ( {1 - \varphi_{\cdotv} } ) \right],
 	\end{align}
 	where $M := \Ebb \left[ \sum_{j=1}^N n_{\cdotv j} \right]$ is approximated along the updating as 
 	$ M = \left( 1 - \varepsilon_t \right) M + {\varepsilon_t} \rho E \left[ n_{ \cdotv \cdotv } \right] $.
 	Further with the Dirichlet prior $\phiv \sim \Dir \left( \eta \mathbf{1}_V\right)$, we have the conditional posterior of $\phiv$ as $(\phiv |-) \sim \Dir (\sum_j n_{1j}+\eta,\ldots, \sum_j n_{Vj}+\eta)$. 
 	Taking the gradient with respect to the reduced-mean parameterization $\varphiv \in\mathbb{R}_+^{V-1}$ on the mini-batch estimation of the negative energy function, we have
 	\begin{align}\label{eq:redgrad}
	 	\nabla _{\varphiv} \! \left[- \tilde H \! \left( {\varphiv} \right) \right]
	 	& \! = \! \frac{ \rho \bar \nv_{:\cdotv} \! + \! \eta \! - \! 1}{\varphiv} \! - \! \frac{\rho n_{V\cdotv} \! + \! \eta \! - \! 1}{1 - \varphi _{\cdotv}}.
 	\end{align}
 	Substituting both \eqref{eq:RiemanPhiGBN} and \eqref{eq:redgrad} into \eqref{eq:FramZmb}, we have \eqref{eq:upvarphi} as
 	$$
 	\varphiv_{t + 1} \!=\! \left[ 
 	\varphiv_t \!+\! \frac{\varepsilon _t}{M}
 	\!\left[ \left( \rho \bar \nv_{: \cdotv} \!+\! \eta \right) \!-\! \left(\rho n_{\cdotv \cdotv} \!+\! \eta V\right) \varphiv_t \right]
 	\!+\! \Nc \left( \bds 0,\frac{2\varepsilon _t}{M} \!\left[ \diag \left( \varphiv_t \right) \!-\! \varphiv_t \varphiv_t ^T \right] \right)
 	\right]_{\triangle} \!,
 	$$
 	where $[\cdotv]_{\triangle}$, denoting the constraint that $\varphiv \in \mathbb{R}_{+}^{V-1}$ and $\varphi_{\cdotv} \le 1$, ensures $\varphiv$ to be valid.

 	Next we prove that equation \eqref{eq:upvarphi} can be equivalently represented as \eqref{eq:upphi}, namely 
 	$$ 
 	\phiv_{t + 1} \!=\! \left[ 
 	\phiv_t \!+\! \frac{\varepsilon _t}{M}
 	\!\left[ \left( \rho \nv_{: \cdotv} \!+\! \eta \right) \!-\! \left(\rho n_{\cdotv \cdotv} \!+\! \eta V\right) \phiv_t \right]
 	\!+\! \Nc \left( \bds 0,\frac{2\varepsilon _t}{M} \diag \left( \phiv_t \right)  \right)
 	\right]_{\angle} \!,
 	$$
 	where $[\cdotv]_{\angle}$ represents the constraint that $\phiv \in \Rbb_{+}^{V}$ and $\bds 1 ^T \phiv = 1$. By substituting $\phiv = (\varphiv^T , 1 - \bds 1^T \varphiv)^T$ into \eqref{eq:upphi}, one can easily verify that the MVN simulation in \eqref{eq:upphi} is identical to that in \eqref{eq:upvarphi}. 
 	By further pointing out the fact that $[\cdotv]_{\triangle}$ is the same as $[\cdotv]_{\angle}$ under the reduced-mean parameterization, we conclude the proof.
 	


\end{document}